\documentclass[10pt,amsmath,amssymb,nofootinbib,superscriptaddress,twocolumn,pra]{revtex4-2}
\usepackage{graphicx}
\usepackage{dcolumn}
\usepackage{braket}
\usepackage{dsfont}
\usepackage{bm}
\usepackage{mathtools}
\usepackage{booktabs}
\usepackage{url}

\usepackage{hyperref}
\usepackage[dvipsnames]{xcolor}
\hypersetup{
    bookmarksnumbered=true, 
    unicode=false, 
    pdfstartview={FitH}, 
    pdftitle={}, 
    pdfauthor={}, 
    pdfsubject={}, 
    pdfcreator={}, 
    pdfproducer={}, 
    pdfkeywords={}, 
    pdfnewwindow=true, 
    colorlinks=true, 
    linkcolor=NavyBlue, 
    citecolor=NavyBlue, 
    filecolor=NavyBlue, 
    urlcolor=NavyBlue 
}

\usepackage{mathrsfs}
\newcounter{one}
\DeclareMathOperator{\id}{id}

\newcommand{\PF}{P_{\mathrm{fail}}}
\newcommand\Mem{M_{\mathrm{g}}}
\newcommand\Mcos{M_{\mathrm{r}}}

\newcommand\Aem{A_{\mathrm{g}}}
\newcommand\Acos{A_{\mathrm{r}}}
\newcommand\Cem{C_{\mathrm{g}}}
\newcommand\Ccos{C_{\mathrm{r}}}
\newcommand\Dem{\calD_{\mathrm{g}}}
\newcommand\Dcos{\calD_{\mathrm{r}}}

\newcommand{\ex}[1]{\langle #1 \rangle}
\newcommand{\dm}[1]{\ketbra{#1}{#1}}

\DeclareMathOperator{\Tr}{Tr}

\newtheorem{theorem}{Theorem}
\newtheorem{lemma}{Lemma}

\newtheorem{corollary}{Corollary}

\newcommand\calC{{\cal C}}
\newcommand\calD{{\cal D}}
\newcommand\calE{{\cal E}}
\newcommand\calF{{\cal F}}

\newcommand\calN{{\cal N}}

\newcommand\calP{{\cal P}}

\newcommand\calR{{\cal R}}

\newcommand\calU{{\cal U}}
\newcommand\calV{{\cal V}}

\newcommand\eps{\epsilon}

\newcommand{\eq}[1]{\begin{align} #1 \end{align}}
\newcommand{\ketbra}[2]{\ket{#1}\bra{#2}}

\def\QED{\mbox{\rule[0pt]{1.5ex}{1.5ex}}}

\def\endproof{\hspace*{\fill}~\QED\par\endtrivlist\unskip}

\newenvironment{proofof}[1]{\vspace*{5mm} \par \noindent
         {\bf Proof of #1:\hspace{2mm}}}{\endproof
}

\newcommand{\bal}{\begin{equation}\begin{aligned}}
\newcommand{\eal}{\end{aligned}\end{equation}}

\begin{document}
\title{Universal tradeoff relations between resource cost and irreversibility of channels: General-resource Wigner-Araki-Yanase theorems and beyond}
\author{Hiroyasu Tajima}
\email{hiroyasu.tajima@inf.kyushu-u.ac.jp}
\affiliation{Department of Informatics, Faculty of Information Science and Electrical Engineering,
Kyushu University, 744 Motooka, Nishi-ku, Fukuoka, 819-0395, Japan}
\affiliation{JST FOREST, 4-1-8 Honcho, Kawaguchi, Saitama, 332-0012, Japan}
\author{Koji Yamaguchi}
\affiliation{Department of Informatics, Faculty of Information Science and Electrical Engineering,
Kyushu University, 744 Motooka, Nishi-ku, Fukuoka, 819-0395, Japan}
\author{Ryuji Takagi}
\affiliation{Department of Basic Science, The University of Tokyo, 3-8-1 Komaba, Meguro-ku, Tokyo 153-8902, Japan}
\author{Yui Kuramochi}
\affiliation{Department of Informatics, Faculty of Information Science and Electrical Engineering,
Kyushu University, 744 Motooka, Nishi-ku, Fukuoka, 819-0395, Japan}
\begin{abstract}
Quantum technologies offer exceptional---sometimes almost magical---speed and performance, yet every quantum process costs physical resources. Designing next-generation quantum devices, therefore, depends on solving the following question: which resources, and in what amount, are required to implement a desired quantum process? Casting the problem in the language of quantum resource theories, we prove a universal cost-irreversibility tradeoff: the lower the irreversibility of a quantum process, the greater the required resource cost for its realization. The trade-off law holds for a broad range of resources---energy, magic, asymmetry, coherence, athermality,  and others---yielding lower bounds on resource cost of any quantum channel.
Its broad scope positions this result as a foundation for deriving the following key results: (1) we show a universal relation between the energetic cost and the irreversibility for arbitrary channels, encompassing the energy-error tradeoff for any measurement or unitary gate; (2) we extend the energy-error tradeoff to free energy and work costs; (3) we extend the Wigner-Araki-Yanase theorem, which is the universal limitation on measurements under conservation laws, to a wide class of resource theories: the probability of failure in distinguishing resourceful states via a measurement is inversely proportional to its resource cost; (4) we prove that infinitely many resource-non-increasing operations in fact require an infinite implementation cost. These findings reveal a universal relationship between quantumness and irreversibility, providing a first step toward a general theory that explains when---and how---quantumness can suppress irreversibility.
\end{abstract}
 
\maketitle

\textbf{\textit{Introduction}---}
Harnessing the laws of quantum mechanics offers quantum advantages that exceed classical limits in computation~\cite{Shor,Grover}, communication~\cite{comm1,comm2}, metrology~\cite{sensing1,sensing2}, and heat engine performance~\cite{TF,FT}. These advantages, however, are not cost‑free: every quantum channel costs physical resources. As practical quantum computers, sensors, and engines inch toward reality, the central question has sharpened: which resources, and how much, are required to realise a given quantum process?

Over the past decades this question has been framed within quantum resource theories~\cite{chitambar_quantum_2019}. Individual resources---entanglement~\cite{RToE1,RToE2,RToE3}, athermality~\cite{QT0,QT1,QT2}, coherence~\cite{aberg_quantifying_2006,baumgratz_quantifying_2014,winter_operational_2016}, asymmetry~\cite{gour_resource_2008,hansen_metric_2008,Marvian_thesis,yadin_general_2016,augusiak_asymptotic_2016,Marvian_distillation,marvian_operational_2022,YT,YT2,SMT,YMST}, magic~\cite{Bravyi2005universal,Veitch_2014stab,wang_efficiently_2020}, and others---have each been formalised, and bounds on the cost of implementing various channels have been studied. However, most results remain isolated by resource type \cite{Wigner1952,Araki-Yanase1960,OzawaWAY,Korzekwa_thesis,TN,Kuramochi-Tajima,ET2023,ozawaWAY_CNOT,Karasawa_2009,TSS,TSS2,TS,TTK,Tajima_Takagi2025,Yamasaki_energy} or constrained to special cases of channels, such as unitary channels~\cite{takagi_universal_2020,Yang_2021}. Only the asymmetry framework offers bounds for implementation costs of arbitrary channels~\cite{TTK}, but the framework is restricted to the single resource:
a unifying principle that compares a broad family of resources on an equal footing and applies to every quantum channel has remained elusive.

We close this gap by proving a universal trade-off between the irreversibility of a channel and the resources required to implement it. In essence, the less irreversibility a channel exhibits for some ensemble of states, the more resources one must supply. The tradeoff relation holds simultaneously for energy, magic, asymmetry, coherence, athermality, and any other resource theory that fits several general criteria, providing strict lower bounds on cost for any quantum channel.

Because of its universality, our inequality opens the door to many applications. We give four examples. First, specializing in energy, we derive an inverse-proportional law between the energy cost of any quantum channel and its irreversibility. Consequently, every measurement, unitary gate, and any other channel that changes the mean energy and is reversible for some state ensemble requires an energy cost inversely proportional to the implementation error.

Second, we apply our results to the athermality resource and derive a lower bound for the athermality cost of any quantum channel. Our bound provides lower bounds for the non-equilibrium free energy and work costs, indicating that the requirements of these thermodynamic resources by a channel can often diverge. Notably, this divergence occurs not only for measurements and unitary gates but also for Gibbs-preserving operations. Our results complement the results in Ref. \cite{Tajima_Takagi2025} that some Gibbs-preserving operations require infinite asymmetry, by demonstrating that a Gibbs-preserving operation requires infinite costs in terms of asymmetry, nonequilibrium free energy, and energy at the same time.

Third, we extend the celebrated Wigner-Araki-Yanase (WAY) theorem \cite{Wigner1952,Araki-Yanase1960,OzawaWAY,Korzekwa_thesis,TN,Kuramochi-Tajima,ET2023}, which is the universal constraint on quantum measurements imposed by conservation laws, to general resource theory. We prove that under a wide variety of resource theories the probability of failing to distinguish resourceful states via a measurement is bounded below by the reciprocal of the resource cost of the measurement. We can also give a similar cost-error tradeoff relation for arbitrary channels which are reversible for some state ensemble. Thus, the above energy error tradeoff for measurements and other channels is merely one facet of a much broader constraint landscape.

Fourth, we reveal that many resource-nonincreasing operations require infinite implementation costs, rendering them infeasible. Such cost-diverging operations include measurement-and-prepare channels, which are often employed in resource distillation and dilution tasks~\cite{brandao_reversible_2010,Brandao2011oneshot,Liu2019oneshot,Regula2020benchmarking}. This finding shows the necessity to classify operations with explicit attention to their physical expense.

Taken together, our results provide a common yardstick for quantifying when and how quantum resources can suppress irreversibility. They introduce a universal principle linking the resource cost and performance of quantum devices such as computers, sensors, and engines.

\textbf{\textit{Framework based on resource theory}---}
In a resource theory, one specifies a set of free states that can be prepared easily and a set of free operations that can be implemented easily.  States or operations outside these sets are called resource states and resource operations.  The sole requirement for choosing the free sets is that no combination of free states and free operations can create a resource.  Subject to this constraint, the choice is flexible, allowing theories that treat entanglement, athermality, coherence, asymmetry, energy, and many other quantities as resources.
To quantify the resource content of a state, one usually adopts a resource monotone: a function that is zero on all free states and never increases under free operations.  Such a monotone serves as the resource measure.

Building on this framework, we introduce our setting.  Let $\mathfrak{S}_{F}$ be the set of free states; any state outside $\mathfrak{S}_F$ is a resource state.  Let $\mathfrak{U}_F$ be the set of free unitary operations—unitaries that can be performed freely.

Because free unitaries form a subset of the free operations, specifying a full resource theory would automatically determine $\mathfrak{U}_F$. However, for our purposes the sets $\mathfrak{S}_F$ and $\mathfrak{U}_F$ suffice. We therefore avoid the unnecessary limitation on the applications of our theory imposed by fixing all free operations. We emphasize that our results nonetheless is valid even when a complete resource theory and its full set of free operations are given.

For the same reason, we do not impose monotonicity for all free operations on a generic resource measure. Instead, we consider a resource measure $M$ which satisfies the following three properties:
\begin{description}
\item[\textbf{(i) monotonicity}] $M$ is non-increasing under any operation written as a combination of (a) adding a free state, (b) performing an operation $\id\otimes\calE$, where $\calE$ is a free unitary or its dual, and (c) a partial trace. 
\item[\textbf{(ii) additivity for product states}] For any states $\rho$ and $\sigma$, the additivity $M(\rho\otimes\sigma)=M(\rho)+M(\sigma)$ holds.
\item[\textbf{(iii) H\"{o}lder continuity}] There is a real function $c(x)$ satisfying $\lim_{x\rightarrow0}c(x)=0$ and $c(x)\le c_{\max}$, and for any system $S$ there are constants $K_S\ge0$, $0<a_S\le\infty$ and  $0<b_S<\infty$ such that $a_S+b_S>1$ and  $\epsilon_m:=\|\rho^{\otimes m}-\sigma_m\|_1$ satisfies
\eq{
|M(\rho^{\otimes m})-M(\sigma_m)|&\le m^{a_S}K_S\epsilon^{b_S}_m+c(\epsilon_m)\label{regularity}
}
for any states $\rho$ on $S$ and $\sigma_m$ on $S_1...S_m$, where each $S_k$ is a copy of $S$. Here $a_S=\infty$ means that there exists a real number $a'_S$ such that $m^{a_S}$ in \eqref{regularity} is replaced by $e^{a'_Sm}$.
\end{description}
The condition (i) replaces the usual monotonicity requirement.  Because free unitaries, the partial trace, and the composition with free states typically belong to the set of free operations any monotone defined with respect to those operations can be used whenever convenient.

As shown later, a measure $M$ satisfying (i)--(iii) exists in many resource theories, including energy, magic, asymmetry, coherence, and athermality.

Once a resource measure for quantum states has been fixed, corresponding measures for operations can be introduced. In this paper, we employ the resource-increasing power. For a given state $\rho$ we define 
\eq{
M_{\rho}(\Lambda):=M(\Lambda(\rho))-M(\rho)
}
and set 
\eq{
M(\Gamma)\coloneqq \sup_{\sigma} M_\sigma(\id_R\otimes\Gamma) ,
}
where the optimization ranges over the reference system $R$.

\textbf{\textit{Implementation cost of a quantum channel---}}
We quantify the cost of implementation of a non-free operation as the amount of resource required to implement the desired operation $\Lambda$ with a free unitary.
We measure the accuracy of implementation by the channel-purified distance~\cite{Gilchrist2005distance}
\eq{
D_F(\Lambda_1,\Lambda_2) := \max_\rho D_F\bigl(\id\otimes\Lambda_1(\rho),\id\otimes\Lambda_2(\rho)\bigr),
}
where $D_F(\rho,\sigma)= \sqrt{1-F(\rho,\sigma)^2}$ and $F(\rho,\sigma)=\Tr\sqrt{\sqrt{\rho}\sigma\sqrt{\rho}}$.
We write $\Lambda_1\sim_\epsilon \Lambda_2$ when $D_F(\Lambda_1,\Lambda_2)\le\epsilon$.

The implementation cost is defined as the minimum resource cost for implementing a channel $\Lambda$ from $A$ to $A'$ with error $\epsilon$:
\eq{
M_c^\epsilon(\Lambda) \coloneqq \min\left\{ M(\eta)| \Lambda(\cdot) \sim_\epsilon \Tr_{B'}[\calU(\cdot\otimes \eta)],\ \calU\in\mathfrak{U}_F\right\},
}
where $\eta$ is a state on an ancillary system $B$, $\calU$ is a unitary operation $\calU(\cdot):=U\cdot U^\dagger$ on $AB$, and $\Tr_{B'}$ is the partial trace over $B'$ with $AB=A'B'$. In other words, the resource cost is the minimum amount of resource attributed to an ancilla state that---together with a free unitary---realises an $\epsilon$-approximate implementation of the target channel $\Lambda$.
When we treat the cost of zero-error implementation, we often use the abbreviation $M_c(\Lambda):=M_c^{\epsilon=0}(\Lambda)$.
For resource theories in which every free operation admits a Stinespring dilation by a free unitary and free states---such as those for energy, athermality, asymmetry, and coherence (when we employ Physically incoherent Operations (PIO) \cite{PIO} as free operations)---the cost $M_c^\epsilon(\Lambda)$ coincides with the minimal resource required when one implements $\Lambda$ using any free operation together with a resourceful ancillary state.

\textbf{\textit{Key quantities---}}
To evaluate the implementation cost, we introduce three quantities.
The first quantity is the irreversibility of a channel for an ensemble constructed by a pair of test states $\Omega:=(\rho_1,\rho_2)$ which are orthogonal to each other: $F(\rho_1,\rho_2)=0$. We then write $\Omega_p := (\rho_1,\rho_2)_p$ for the test ensemble in which $\rho_1$ and $\rho_2$ occur with prior probabilities $p_1 = p$ and $p_2 = 1-p$, respectively.
Each test state $\rho_k$ changes via the channel $\Lambda$, and we try to recover the test state with a recovery map $\calR:A'\to A$.
The irreversibility of the channel with respect to the ensemble $\Omega_p$ is defined as the minimum average error \cite{TTK}
\eq{
\delta(\Lambda,\Omega_p)^2:=\min_{\calR}\sum_k p_k D_F\bigl(\rho_k,\calR\circ\Lambda(\rho_k)\bigr)^2,
}
where $\calR$ runs over CPTP maps from $A'$ to $A$.

One key property of the irreversibility measure $\delta(\calE,\Omega_p)$ is that it equals the minimum failure probability for distinguishing the ensemble $\Omega'_{p,\Lambda}:=(\Lambda(\rho_1),\Lambda(\rho_2))_p$.
Concretely, with a binary Positive-Operator Valued Measure (POVM) $\mathbb{Q}:=\{Q_k\}_{k=1,2}$, the failure probability of distinguishing the states $\rho_1$ and $\rho_2$ in $\Omega_p$ is defined as
\eq{
\PF(\Omega_{p,\Lambda},\mathbb{Q}):=\sum_{k=1,2}p_k\Tr[(1-Q_k)\Lambda(\rho_k)],
}
and we have (see the Supplementary Materials)
\eq{
\delta(\Lambda,\Omega_p)^2=\min_{\mathbb{Q}}\PF(\Omega'_{p,\Lambda},\mathbb{Q}).
}
Furthermore, the irreversibility measure $\delta(\Lambda,\Omega_p)$ unifies many existing irreversibility measures \cite{TTK}: it encompasses recovery errors in quantum error correction and Petz-map recovery, and reproduces virtually all established error-and-disturbance metrics for quantum measurements \cite{ET2023}.

Next, we quantify both the resource gain obtained by state discrimination and the channel power required to perform such discrimination. 
To achieve this, we use a two-level quantum register system $K$ to store the classical outcome of the measurement, along with a particular orthonormal basis $\{\ket{k}\}_{k=1,2}$ on that register.
Below, we refer to the tuple $(K,\{\ket{k}\}_{k=1,2})$ as $\mathscr{K}$.

Using the register $\mathscr{K}$, the resource gain via the discrimination of the states in the ensemble $\Omega_p=(\rho_1,\rho_2)_p$ is defined as
\eq{
&\Mem(\Omega_p|\mathscr{K})\nonumber\\
&:=\max_{\rho,\mathbb{P}}\left\{M_\rho(\Lambda_{\mathbb{P}})\middle|\Lambda_{\mathbb{P}}(\rho)=\sum_{k=1,2}p_k\rho_k\otimes\ket{k}\bra{k}_K\right\},
}
where the maximization ranges over  $\rho$ on $A$ and the two-valued projective measurement channel 
$\Lambda_{\mathbb{P}}(\cdot):=\sum_{k=1,2}P_k\cdot P_k\otimes\ket{k}\bra{k}_K$ from $A$ to $AK$ whose two-valued Projection Valued Measure (PVM) $\mathbb{P}:=\{P_k\}$ describes the test states $\rho_1$ and $\rho_2$ with the probability 1.
This quantity describes the maximum resource gain from distinguishing the states $\rho_1$ and $\rho_2$.

Similarly, we define the channel power required to optimally distinguish the states in $\Omega'_{p,\Lambda}=(\Lambda(\rho_1),\Lambda(\rho_2))_p$ as
\eq{
&\Mcos(\Omega'_{p,\Lambda}|\mathscr{K})\nonumber\\
&:=\min\left\{M(\Lambda_{\mathbb{Q}})\middle|\delta(\Lambda,\Omega_p)^2=\PF(\Omega'_{p,\Lambda},\mathbb{Q})\right\},
}
where the minimization ranges over the two-valued POVM measurement channel $\Lambda_{\mathbb{Q}}(\cdot):=\sum_{k=1,2}\sqrt{Q_k}\cdot \sqrt{Q_k}\otimes\ket{k}\bra{k}_K$.
This quantity represents the minimal channel power required to optimally discriminate between $\Lambda(\rho_1)$ and $\Lambda(\rho_2)$.

\textbf{\textit{Main Results}---}
Now, let us introduce our main results. For simplicity, we focus on the case of $a_S=b_S=1$ in the main text.
For other general cases, see the Supplementary Materials.
\begin{theorem}\label{Thm:RItradeoff}
Let $A$ and $A'$ be quantum systems.
For any channel $\Lambda$ from $A$ to $A'$, any test ensemble $\Omega_p$, and any two-level register $\mathscr{K}$ the following inequality holds:
\eq{
M_c(\Lambda)\ge\frac{|\Mem(\Omega_p|\mathscr{K})-\Mcos(\Omega'_{p,\Lambda}|\mathscr{K})|^2_+}{16K_A\delta(\Lambda,\Omega_p)}-c',\label{eq:RItradeoff}
}
where $|x|_+$ returns $x$ when $x$ is positive, and returns $0$ otherwise. And $c':=c_{\max}+|\Mem(\Omega_p|\mathscr{K})-\Mcos(\Omega'_{p,\Lambda}|\mathscr{K})|_++\frac{|\Mem(\Omega_p|\mathscr{K})-\Mcos(\Omega'_{p,\Lambda}|\mathscr{K})|^2_+}{64K_A}$.
\end{theorem}
We present the proof of this theorem, along with a tighter bound, in the Supplementary Material.

Theorem \ref{Thm:RItradeoff} states that, given a suitable ensemble $\Omega_p$ and a suitable register $\mathscr{K}$, whenever the resource gain obtained by distinguishing the two states in $\Omega_p$ exceeds the channel power required to optimally distinguish the transformed ensemble $\Omega_{p,\Lambda}$, the cost of $\Lambda$ is inversely proportional to the irreversibility of $\Lambda$ on $\Omega_p$.
In particular, if there exists an ensemble $\Omega_p$ for which $\Mem(\Omega_p|\mathscr{K})-\Mcos(\Omega'_{p,\Lambda}|\mathscr{K})>0$ and $\delta(\Lambda,\Omega_p)=0$, then $M_c(\Lambda)$ diverges: the channel cannot be implemented unless one supplies an infinite amount of resource.

What is the cost of implementing a cost-diverging channel when a finite implementation error is allowed? Our results answer this question as well. When $\Mem(\Omega_p|\mathscr{K})-\Mcos(\Omega'_{p,\Lambda}|\mathscr{K})>0$ and $\delta(\Lambda,\Omega_p)=0$ hold, the following substitution is valid for \eqref{eq:RItradeoff}:
\eq{
M_c(\Lambda)\rightarrow M^\epsilon_c(\Lambda), \enskip \delta(\Lambda,\Omega_p)\rightarrow \epsilon.\label{eq:REtradeoff}
}
Namely, the cost of approximate implementation is inversely proportional to the error.

If $\Mem(\Omega_p|\mathscr{K})-\Mcos(\Omega'_{p,\Lambda}|\mathscr{K})\le0$ holds for every test ensemble $\Omega_p$, the cost and irreversibility need not be inversely proportional. Even in this situation, the following bound always holds:
\begin{theorem}\label{Thm:RPtradeoff}
Let $A$ and $A'$ be quantum systems.
For any channel $\Lambda$ from $A$ to $A'$, any test ensemble $\Omega_p$, any two-level register $\mathscr{K}$, and any two-valued POVM $\mathbb{Q}:=\{Q_k\}^{2}_{k=1}$ on $A'$ and its measurement channel $\Lambda_{\mathbb{Q}}(\cdot):=\sum_k\sqrt{Q_k}\cdot\sqrt{Q_k}\otimes\ket{k}\bra{k}_K$, the following inequality holds:
\eq{
M_c(\Lambda)\ge\frac{|\Mem(\Omega_p|\mathscr{K})-M(\Lambda_{\mathbb{Q}})|^2_+}{16K_A\sqrt{\PF(\Omega'_{p,\Lambda},\mathbb{Q})}}-c'_2,\label{eq:RPtradeoff}
}
where $c'_2:=c_{\max}+|\Mem(\Omega_p|\mathscr{K})-M(\Lambda_{\mathbb{Q}})|_++\frac{|\Mem(\Omega_p|\mathscr{K})-M(\Lambda_{\mathbb{Q}})|^2_+}{64K_A}$.
\end{theorem}

Our results also provide the cost of a channel $\Lambda$ that converts one pair of orthogonal states into another pair of orthogonal states with error $\epsilon$:
\begin{corollary}
Let $A$ and $A'$ be quantum systems.
Let a CPTP map $\Lambda:A\rightarrow A'$ approximately convert an orthogonal pair $\Omega:=(\rho_1, \rho_2)$ on $A$ to an orthogonal pair $\Omega'':=(\rho''_1, \rho''_2)$ on $A'$ as
\eq{
D_F(\Lambda(\rho_i),\rho''_i)\le\epsilon,\enskip (i=1,2).\label{cond_approx}
}
Then, for any two-level register $\mathscr{K}$, the following holds:
\eq{
M_c(\Lambda)\ge\frac{|\Mem(\Omega|\mathscr{K})-\Mcos(\Omega''|\mathscr{K})|^2_+}{16K_A\epsilon}-c'_3.\label{eq:REtradeoff_conversion}
}
Here, $\Mem(\Omega|\mathscr{K}):=\max_p \Mem(\Omega_p|\mathscr{K})$ and $\Mcos(\Omega''|\mathscr{K}):=\min_p \Mem(\Omega''_p|\mathscr{K})$.
And $c'_3:=c_{\max}+|\Mem(\Omega|\mathscr{K})-\Mcos(\Omega''|\mathscr{K})|_++\frac{|\Mem(\Omega|\mathscr{K})-\Mcos(\Omega''|\mathscr{K})|^2_+}{64K_A}$.
\end{corollary}

\textbf{\textit{Applications}---}
Our results apply to any resource theory satisfying conditions (i)--(iii). 
As listed in the Supplementary Materials, the applicable range includes energy, magic, asymmetry, coherence and athermality, among others.
Furthermore, our results apply to any quantum channel---including measurement channels, unitary gates, error correcting code channels, thermodynamic processes and others.
Therefore, they open the door to a remarkably broad range of applications.
In what follows, we illustrate several important use cases.


\textit{Example of application 1: Energy-irreversibility tradeoff for arbitrary channels}--
Our results provide a universal bound for the energy cost of arbitrary quantum channels by applying to the resource theory of energy.
Here, the resource theory of energy is a resource theory in which (a) Free states are ground states of the system Hamiltonian, (b) Free operations are generated by adding an ancillary system with an arbitrary Hamiltonian and its ground state; applying unitary operations that conserve total energy; taking partial traces.
We stress that with this definition, any projective measurement channel to energy eigenspaces is a free operation.

If the ground‑state energy of every Hamiltonian is set to zero, the energy expectation value itself becomes a resource measure that satisfies the conditions required by our results. 
Our framework can therefore be applied to obtain lower bounds on the energy cost of any quantum channel, i.e., the minimum energy required to realize the channel. Concretely, for an arbitrary channel $\Lambda$ and any test ensemble $\Omega_{1/2}$ consisting of two orthogonal pure states $\ket{\psi_1}$ and $\ket{\psi_2}$, the following inequality holds for the energy cost $E_c(\Lambda)$:
\eq{
E_{c}(\Lambda)\ge\frac{\calC(\Lambda,\Omega_{1/2})^2}{8(\Delta_{H}+\Delta_{H'})\delta(\Lambda,\Omega_{1/2})}-2\calC(\Lambda,\Omega_{1/2}).\label{EItradeoff}
}
Here, $\Delta_X$ is the difference between the maximum and minimum eigenvalues of the operator $X$ and  $\calC(\Lambda,\Omega_{1/2}):=|\bra{\psi_2}H-\Lambda^\dagger(H')\ket{\psi_1}|$  defined in Ref.~\cite{TTK}.

This inequality conveys a pivotal message: if a channel satisfies $\calC(\Lambda,\Omega_{1/2})>0$ and $\delta(\Lambda,\Omega_{1/2})=0$ for some orthogonal pure test ensemble $\Omega_{1/2}$, then implementing that channel demands infinite energy.
Moreover, for such energy-cost diverging channels, the energy cost of an approximate implementation is inversely proportional to the implementation error:
\eq{
E^{\epsilon}_{c}(\Lambda)\ge\frac{\calC(\Lambda,\Omega_{1/2})^2}{8(\Delta_{H}+\Delta_{H'})\epsilon}-c'_E,\label{EEtradeoff}
}
where $c'_{E}:=2\calC(\Lambda,\Omega_{1/2})+3\Delta_{H'}\epsilon$

Equations~\eqref{EItradeoff} and \eqref{EEtradeoff} are therefore of paramount importance: as we show below, the conditions $\calC(\Lambda,\Omega_{1/2})>0$ and $\delta(\Lambda,\Omega_{1/2})=0$ are met by a wide variety of channels---including unitary gates that change the system's energy and projective measurements that do not commute with the Hamiltonian.

\textit{Energy-error tradeoff for arbitrary projective measurements}---
Consider the projective‑measurement channel that outputs only classical data 
\eq{
\Gamma_{\mathbb{P}}(\cdot):=\sum_k\Tr[P_k\cdot] \ket{k}\bra{k}_R
}
and treat it within the energy resource theory, assuming a trivial Hamiltonian $H_R=0$ on the register $R$. Then, an energy‑projective measurement that commutes with the system Hamiltonian $H$ is a free operation, whereas it is not clear whether any channel for which at least one projector satisfies $[H,P_k]\neq0$ is free or not.
Our results show that all of such measurement channels are non-free and require infinite energy for zero‑error implementation, since we can always take a proper $\Omega_{1/2}$ satisfying $\calC(\Gamma_{\mathbb{P}},\Omega_{1/2})>0$ and $\delta(\Gamma_{\mathbb{P}},\Omega_{1/2})=0$. Therefore, even for approximate implementation, the energy costs of such measurements are inversely proportional to the error. Indeed, from \eqref{EEtradeoff}, we obtain the explicit bound
\eq{
E^{\epsilon}_c(\Gamma_{\mathbb{P}})\ge\max_k\frac{\|[P_k,H]\|^2_\infty}{2\Delta_H\epsilon}-2\Delta_{H}.
}
As a concrete illustration, consider a spin‑$1/2$ particle in a magnetic field with Hamiltonian $H=\tfrac{\hbar\omega}{2}\sigma_z$. A projective measurement in the $X$‑basis satisfies $[H,\sigma_x]\neq0$ and thus it is a cost‑diverging measurement. Measuring $\sigma_x$ perfectly would require an apparatus with access to an infinite energy reservoir; if an error $\epsilon$ is acceptable, the minimal energy cost reduces to approximately $\hbar\omega/(8\epsilon)$. This reformulates the Wigner–Araki–Yanase–Ozawa limit as an energy bound, making explicit the universal principle that higher measurement precision for observables non‑commuting with $H$ must be purchased with higher energy expenditure.

\textit{Energy-error tradeoff for arbitrary unitary operations}---
Let $\Lambda_{V}(\cdot):=V\cdot V^{\dagger}$ be a unitary channel on $A$. Within the energy resource theory, it is free iff the gate conserves energy, i.e., $[V,H]=0$. Whenever $[V,H]\neq0$, the gate is non‑free, and our general trade‑off again makes its cost explicit. Choosing an orthogonal pure ensemble $\Omega_{1/2}=(\ket{\psi_1},\ket{\psi_2})_{1/2}$ with $\ket{\psi_{1,2}}$ as the maximal and minimal eigenstates of $H-V^\dagger HV$, one finds $\delta(\Lambda_V,\Omega_{1/2})=0$ while $\calC(\Omega_{1/2},\Lambda_V)>0$.
From these facts, we can derive the following bound from \eqref{EEtradeoff}:
\eq{
E^{\epsilon}_c(V)\ge\frac{\Delta^{2}_{H-V^\dagger HV}}{64\Delta_{H}\epsilon}-\Delta_H(2+3\epsilon).
}
Therefore, the energy cost diverges as $1/\epsilon$ and becomes infinite in the zero‑error limit.

\textit{Example of application 2: free energy-irreversibility tradeoff for arbitrary channels}---
The above bound on energy cost is appealing, given the fundamental and ubiquitous role of energy in physics and related fields. However, in quantum thermodynamics, it is, at least in principle, possible to supply an unbounded amount of energy by repeatedly using an external system prepared in a Gibbs state. While this is clearly unfeasible in practice, this suggests that the constraints for the energy cost may not serve as a stringent theoretical constraint in such idealized scenarios.

Fortunately, our main results \eqref{eq:RItradeoff} and \eqref{eq:REtradeoff} work even in such scenarios: they provide another bound for the free energy cost of an arbitrary channel that is particularly useful in quantum thermodynamics. This new bound imposes stricter constraints than the energy-error trade-offs.
We take the Gibbs state $\tau:=e^{-\beta H}/\Tr[e^{-\beta H}]$, and take the relative entropy of athermality $A_R(\rho):=\frac{1}{\beta}D(\rho\|\tau)$ that is one of the standard measure of athermality as the resource measure $M$. Then, we can recast the bounds for the energy costs. The inequality \eqref{eq:RItradeoff} provides a lower bound, for any channel, on the implementation cost measured in the athermality. 
Importantly, the athermality is related to the non-equilibrium free energy $F(\rho):=\ex{H}_{\rho}-\frac{1}{\beta}S(\rho)$, where $H$ is the Hamiltonian of the system and $S(\rho):=-\Tr[\rho\log\rho]$ is the von-Neumann entropy, as follows:
\eq{
A_R(\rho)=F(\rho)-F(\tau).
}
Because the relative entropy of athermality naturally represents ``the amount of energy stored that can be converted into work via an isothermal process whose initial and final Hamiltonians are the same,'' it is an appropriate measure of energetic resources in a battery in the thermodynamic settings.

When there is a register $\mathscr{K}$ and a test ensemble $\Omega_p$ satisfying $\Aem(\Omega_p|\mathscr{K})-\Acos(\Omega'_{p,\Lambda}|\mathscr{K})>0$ and $\delta(\Lambda,\Omega_{p})=0$, the required cost diverges. In this case, \eqref{eq:REtradeoff} becomes
\eq{
A^{\epsilon}_c(\Lambda)\ge\frac{(\Aem(\Omega_p|\mathscr{K})-\Acos(\Omega'_{p,\Lambda}|\mathscr{K}))^2}{16K_A\epsilon}-h',
}
where $K_A=\Delta_{H}+\frac{2}{\beta}\log d_A$ and 
$h':=\frac{1}{\beta}\log2+(\Aem(\Omega_p|\mathscr{K})-\Acos(\Omega'_{p,\Lambda}|\mathscr{K})|_++\frac{(\Aem(\Omega_p|\mathscr{K})-\Acos(\Omega'_{p,\Lambda}|\mathscr{K}))^2}{64K_A}.$
Therefore, the implementation cost of $\Lambda$ scales inversely with the error, and perfect implementation is impossible unless an infinite amount of free energy is available. In the standard thermal‑operation with a work‑bit setup, where the free‑energy source is only the work bit, the inequalities \eqref{eq:RItradeoff} and \eqref{eq:REtradeoff} immediately give an analogous inverse relationship between the required work cost and the implementation error:
\eq{
W^{\epsilon}_c(\Lambda)\ge\frac{(\Aem(\Omega_p|\mathscr{K})-\Acos(\Omega'_{p,\Lambda}|\mathscr{K}))^2}{16K_A\epsilon}-h'-\frac{1}{\beta}\log2.
}

\textit{Free energy-error tradeoff for coherence erasure channels, including measurements and Gibbs-preserving operations}---
When do the conditions $\Aem(\Omega_p|\mathscr{K})-\Acos(\Omega'_{p,\Lambda}|\mathscr{K})>0$ and $\delta(\Lambda,\Omega_{p})=0$ hold? An important class is captured by the following result:
\begin{theorem}\label{Thm:c_erase}
Let $A$ be a $d$-level quantum system, and $\Lambda$ be a CPTP map from $A$ to $A$. Let $H$ be the Hamiltonian on $A$, and $\{\ket{j}\}$ be eigenbasis of $H$.
We choose $\Omega_{1/2}:=(\ket{\psi_{j,j'}},\ket{\psi^{\perp}_{j,j'}})_{1/2}$ as $\ket{\psi_{j,j'}}:=\frac{\ket{j}+\ket{j'}}{\sqrt{2}}$ and $\ket{\psi^{\perp}_{j,j'}}:=\frac{\ket{j}-\ket{j'}}{\sqrt{2}}$.
When $\Lambda$ satisfies
\eq{
\Lambda(\psi_{j,j'})=\ket{j_0}\bra{j_0},\enskip\Lambda(\psi^{\perp}_{j,j'})=\ket{j_1}\bra{j_1},\enskip j_0\ne j_1,
}
the following inequalities hold:
\eq{
\delta(\Lambda,\Omega_{1/2})&=0,\\
\Aem(\Omega_{1/2}|\mathscr{K})-\Acos(\Omega'_{1/2,\Lambda}|\mathscr{K})&>\frac{1}{2}|E_j-E_{j'}|-\frac{2}{\beta}\log2,
}
where $\mathscr{K}:=(K,\{\ket{k}\}_{k=1,2})$ whose Hamiltonian $H_K$ is the trivial Hamiltonian.
\end{theorem}
Namely, when $|E_j-E_{j'}|>\frac{4}{\beta}\log2$, the free energy cost of approximate implementation of $\Lambda$ is inversely proportional to the error.
This observation implies that erasing coherence between different energy levels reversibly costs infinite amount of free‑energy. From this observation, we can show that various channels require the free‑energy cost scales inversely with the allowed implementation error. Below we show several examples (For the details, see the Supplementary Materials):
\begin{description}
\item[Projective measurements]
Theorem \ref{Thm:c_erase} predicts that projective measurement channel $\Lambda_{\mathbb{P}}$ whose $\{P_k\}$ can distinguish $\ket{\psi_{j,j'}}$ and $\ket{\psi^{\perp}_{j,j'}}$ with probability 1 requires an infinite amount of free energy.
\item[Unitary operations] 
Theorem \ref{Thm:c_erase} predicts that any unitary channel $\calU$ which converts $\ket{\psi_{j,j'}}$ and $\ket{\psi^{\perp}_{j,j'}}$ to energy eigenvectors requires an infinite amount of free energy. It also implies that any $\calU'$ which converts energy eigenvectors to $\ket{\psi_{j,j'}}$ and $\ket{\psi^{\perp}_{j,j'}}$ also has the diverging free energy cost.
\item[Gibbs-preserving operations]
Theorem \ref{Thm:c_erase} predicts that some Gibbs-preserving operations require an infinite amount of free energy. This finding complements the result of Ref.~\cite{Tajima_Takagi2025}, which showed that certain Gibbs‑preserving operations require infinite energy coherence (i.e., asymmetry). Indeed, we explicitly construct channels that demand infinite coherence, free energy and energy at the same time. In order to underscore that this result does not directly follow from the prior study~\cite{Tajima_Takagi2025}, we also provide an example of a state with finite athermality and energy that exhibits diverging coherence.
\end{description}
We emphasize Theorem \ref{Thm:c_erase} provides just an example of high-cost class. We can obtain lower bounds of free energy cost and work cost for any quantum channel $\Lambda$ by directly calculating $\Aem(\Omega_p|\mathscr{K})-\Aem(\Omega'_{p,\Lambda}|\mathscr{K})$ and $\delta(\Lambda,\Omega_p)$ for arbitrary $\Omega_p$. 

\textit{Example of application 3: General-resource Wigner-Araki-Yanase theorem}---
By applying the main result of this paper to measurement processes, we obtain a Wigner-Araki-Yanase (WAY) theorem valid for general resource theories. The traditional WAY theorem shows that the presence of a conserved quantity---typically energy or angular momentum---limits measurement accuracy, yet the ``resource'' in question has essentially been restricted to quantum fluctuations of the conserved quantity itself. As shown in the applications 1 and 2, the WAY theorem can be recast as bounds of energy cost or free‑energy cost, but our framework yields a far broader limitation on measurement channels. In particular, the following theorem holds for any resource theory admitting a measure that satisfies conditions (i)--(iii):

\begin{theorem}
Let $\Omega:=(\rho_1,\rho_2)$ be two orthogonal test states on a system $A$.
Assume they are distinguished---up to error $\epsilon$---by an indirect measurement $(\eta^B,V,\{E_i\}_{i=1,2})$ involving an auxiliary system $B$,
\eq{
\Tr[(1-E_i)V\rho_i\otimes\eta V^\dagger]\le\epsilon^2,\enskip (i=1,2).
}
Suppose $V$ is a free unitary and $\Lambda_{\mathbb{E}}(...):=\sum_i\sqrt{E_i}...\sqrt{E_i}\otimes\ket{i}\bra{i}_R$ is a completely free channel, i.e., $\id_{R}\otimes\Lambda_{\mathbb{E}}$ is free for an arbitrary reference system $R$. Then
\eq{
M(\eta^B)\ge\max_p\frac{\Mem(\Omega_p|\mathscr{K})^2}{16K_A\epsilon}-c'_4,
}
where $c'_4:=\max_p\Mem(\Omega_p|\mathscr{K})-c_{\max}$.
\end{theorem}

The requirement that $\Lambda_{\mathbb{E}}(...):=\sum_i\sqrt{E_i}...\sqrt{E_i}\otimes\ket{i}\bra{i}_R$ be a completely free channel is precisely the Yanase condition in the WAY theorem—that the measurement performed on the probe commutes with the conserved quantity.

This theorem states that if the resource embedded in the test ensemble $\Omega_p$ on $A$, namely $\Mem(\Omega_p|\mathscr{K})$, is strictly positive, then perfect discrimination of the orthogonal pair $\Omega$ demands infinite resources in the apparatus. Hence, over a very wide class of resource theories, measurement accuracy and resource expenditure are unavoidably in trade-off.

\textit{Example of application 4: Resource-non-increasing maps requiring infinite resource}---
Our framework predicts that resource‑non‑increasing (RNI) operations---often classified as free---can require a diverging implementation cost. A canonical example is the measurement‑and‑prepare channel that is treated as a free operation in many resource-theoretic results. The inequality \eqref{eq:RItradeoff} shows that any resource-non-increasing operation $\Lambda$ for which a test ensemble $\Omega_p$ satisfies both $\delta(\Lambda,\Omega_p)=0$ and $\Mem(\Omega_p|\mathscr{K})-\Mcos(\Omega'_{p,\Lambda}|\mathscr{K})>0$ demands infinite resources. Even when a finite error~$\epsilon$ is tolerated, the inequality \eqref{eq:REtradeoff} implies that the corresponding channel cost diverges as $1/\epsilon$.

\begin{theorem}
The resource theories of energy, magic, asymmetry, coherence, and athermality each contain an RNI operation $\Lambda$ and a suitable ensemble $\Omega_p$ such that $\delta(\Lambda,\Omega_p)=0$ and $\Mem(\Omega_p|\mathscr{K})-\Mcos(\Omega'_{p,\Lambda}|\mathscr{K})>0$. Consequently, these theories possess RNI operations whose implementation costs diverge.
\end{theorem}

Thus, the formal condition ``non‑increasing in resource'' offers no guarantee of practical ease. Our result generalises the impossibility of implementing certain Gibbs‑preserving operations in quantum thermodynamics \cite{Tajima_Takagi2025} to a wide array of resource theories.

\textbf{\textit{Discussion}---}
In this paper, we link the ``difficulty'' of a quantum operation to the physical resources that must be invested, showing that the required resource cost is universally inversely proportional to the operation's irreversibility on a suitable test ensemble. All that is needed is a single resource monotone satisfying the minimal axioms of monotonicity, additivity, and mild regularity—conditions met by most practical measures. Consequently, our result simultaneously embraces the principal resource theories of energy, magic, asymmetry, coherence, and athermality, yielding a uniform inverse‑error bound on implementation cost. It furnishes lower bounds for both energy and free‑energy consumption of any quantum channel, lifts the Wigner–Araki–Yanase restriction far beyond conservation‑law settings, and exposes hidden cost of resource-non-increasing operations. In essence, the familiar engineering rule that precision trades off against resources is elevated to a rigorous, framework‑independent theorem.

This theorem opens a broad avenue for future investigation. The most crucial point is that it guarantees the existence of universal trade‑offs between irreversibility and resources in a quite wide class of resource theories, thereby constraining every quantum channel. Armed with this insight, one can now revisit individual resource theories and obtain resource-irreversibility trade‑off bounds for them that are both tighter and still universally valid for arbitrary channels. Indeed, within the framework of asymmetry, such tight---often optimal---bounds have already been obtained \cite{TTK}, suggesting that analogous results will appear in many other resource theories. Moreover, further extensions of the present methods could impose simultaneous constraints on an even broader family of theories. Collectively, these advances promise not only to clarify the universal link between irreversibility and quantumness but also to advance a central goal of resource theory: determining, with precision, the fundamental implementation cost of any quantum channel.

\let\oldaddcontentsline\addcontentsline
\renewcommand{\addcontentsline}[3]{}

\begin{acknowledgments}
We thank Bartosz Regula, Yosuke Mitsuhashi, Haruki Emori, Tomohiro Shitara, Shintaro Minagawa and Akihiro Hokkyo for fruitful discussions. H.T. was supported by JSPS Grants-in-Aid for Scientific Research 
No. JP25K00924, and MEXT KAKENHI Grant-in-Aid for Transformative
Research Areas B ``Quantum Energy Innovation” Grant Numbers 24H00830 and 24H00831, JST MOONSHOT No. JPMJMS2061, and JST FOREST No. JPMJFR2365.
R.T. acknowledges the support of JSPS KAKENHI Grant Number JP23K19028, JP24K16975, JP25K00924, JST, CREST Grant Number JPMJCR23I3, Japan, and MEXT KAKENHI Grant-in-Aid for Transformative
Research Areas A ``Extreme Universe” Grant Number JP24H00943.
K.Y. acknowledges support from JSPS KAKENHI Grant No. JP24KJ0085.
Y.K.\ acknowledges support from JSPS KAKENHI Grant No.~JP22K13977. 
\end{acknowledgments}

\bibliographystyle{apsrmp4-2}
\bibliography{RI_ref}


\widetext
\appendix

\let\addcontentsline\oldaddcontentsline

\clearpage

\widetext




\setcounter{section}{0}
\renewcommand{\thesection}{S.\arabic{section}}

\setcounter{theorem}{0}
\renewcommand{\thetheorem}{S.\arabic{theorem}}

\setcounter{figure}{0}
\renewcommand{\thefigure}{S.\arabic{figure}}

\makeatletter
\@removefromreset{equation}{section}
\makeatother

\setcounter{equation}{0}
\renewcommand{\theequation}{S.\arabic{equation}}

\begin{center}
 {\Large\bf Supplemental Material}
\end{center}

\tableofcontents

\section{Background and setting}\label{app:background}

\subsection{Framework based on resource theory}
In a resource theory, one specifies a set of free states that can be prepared easily and a set of free operations that can be implemented easily.  States or operations outside these sets are called resource states and resource operations.  The sole requirement for choosing the free sets is that no combination of free states and free operations can create a resource.  Subject to this constraint, the choice is flexible, allowing theories that treat entanglement, athermality, coherence, asymmetry, energy, and many other quantities as resources.
To quantify the resource content of a state, one usually adopts a resource monotone: a function that is zero on all free states and never increases under free operations.  Such a monotone serves as the resource measure.

Building on this framework, we introduce our setting.  Let $\mathfrak{S}_F$ be the set of free states; any state outside $\mathfrak{S}_F$ is a resource state.  Let $\mathfrak{U}_F$ be the set of free unitary operations—unitaries that can be performed freely.

Because free unitaries form a subset of the free operations, specifying a full resource theory would automatically determine $\mathfrak{U}_F$. However, for our purposes the sets $\mathfrak{S}_F$ and $\mathfrak{U}_F$ suffice. We therefore avoid the unnecessary limitation on the applications of our theory imposed by fixing all free operations. We emphasize that our results nonetheless is valid even when a complete resource theory and its full set of free operations are given.

For the same reason, we do not impose monotonicity for all free operations on a generic resource measure. Instead, for any resource indicator $M$ defined by the given set of free unitaries, we assume the following three properties:
\begin{description}
\item[\textbf{(i) monotonicity}] $M$ is non-increasing under any operation written as a combination of (a) adding a free state, (b) performing an operation $\id\otimes\calE$, where $\calE$ is a quantum operation free unitary or its dual, and (c) a partial trace. 
\item[\textbf{(ii) additivity for product states}] For any states $\rho$ and $\sigma$, the additivity $M(\rho\otimes\sigma)=M(\rho)+M(\sigma)$ holds.
\item[\textbf{(iii) H\"{o}lder continuity}] There is a real function $c(x)$ satisfying $\lim_{x\rightarrow0}c(x)=0$ and $c(x)\le c_{\max}$, and for any system $S$ there are constants $K_S\ge0$, $0<a_S\le\infty$ and  $0<b_S<\infty$ such that $a_S+b_S>1$ and  
\eq{
|M(\rho^{\otimes m})-M(\sigma_m)|\le m^{a_S}K_S\|\rho^{\otimes m}-\sigma_m\|^{b_S}_1+c(\|\rho^{\otimes m}-\sigma_m\|_1)\label{regularity_S}
}
for any states $\rho$ on $S$ and $\sigma_m$ on $S_1...S_m$, where each $S_k$ is a copy of $S$. Here $a_S=\infty$ means that there exists a real number $a'_S$ such that $m^{a_S}$ in \eqref{regularity_S} is replaced by $e^{a'_Sm}$.
\end{description}
The condition (i) replaces the usual monotonicity requirement.  Because free unitaries, the partial trace, and the composition with free states typically belong to the set of free operations any monotone defined with respect to those operations can be used whenever convenient.

To broaden the scope of application and to strengthen the bounds, sometimes the condition (ii) is replaced by the following alternative version:
\begin{description}
\item[\textbf{(ii-strong) full additivity}] 
Let $\{S^{(i)}\}_{(i=1,\dots,l)}$ be a collection of quantum systems. We say that a resource measure $M$ satisfies full additivity with respect to $\{S^{(i)}\}_{(i=1,\dots,l)}$ if the following condition holds: consider any two composite systems, labeled $A$ and $B$, each formed by choosing and combining an arbitrary number of copies of the systems in $\{S^{(i)}\}_{(i=1,\dots,l)}$, e.g., $A=S^{(1)}_1S^{(1)}_2S^{(1)}_3S^{(2)}_1$ and $B=S^{(1)}_1S^{(1)}_2S^{(2)}_1S^{(2)}_2S^{(3)}_{1}$, where $S^{(i)}_j$ the the $j$-th copy of $S^{(i)}$. Then, for any state $\rho_{AB}$ on the composite system $AB$, the quantity $M$ satisfies
\eq{
M(\rho_{AB})=M(\rho_A)+M(\rho_B),
}
where $\rho_A:=\Tr_{B}[\rho_{AB}]$ and $\rho_B:=\Tr_{A}[\rho_{AB}]$
\end{description}
Since (ii-strong) implies (ii), when a measure $M$ satisfies (i), (ii-strong) and (iii), it also satisfies (i)--(iii).
When we use the alternative version, we will explicitly mention it.

Whenever a resource measure $M$ satisfies the conditions (i)--(iii) or (ii-strong), the tradeoff relation between the resource cost measured by $M$ and the irreversibility holds for an arbitrary quantum channel.
And, as shown later, many types of resources have a measure $M$ satisfying the above (i)--(iii), including energy, magic, asymmetry, coherence, and athermality.
The condition (ii-strong) is satisfied by energy.

Once a resource measure for quantum states has been fixed, a corresponding measure for operations can be introduced. 
Again, there are various possible measures; we employ the following resource‑increasing power in this paper:
\eq{
M(\Lambda)\coloneqq \sup_{\rho} M(\id_R\otimes\Lambda(\rho)) - M(\rho),
}
where the optimization over $R$ is also taken. 

This measure satisfies the following property:
\begin{description}
\item[(p1) Subadditivity under concatenation] For any concatenation channel $\Lambda_1\circ\Lambda_2$,
\eq{
M(\Lambda_1\circ\Lambda_2)\le M(\Lambda_1)+M(\Lambda_2).
}
\end{description}

We give the proof of (p1)  below:
\begin{proofof}{the property (p1)}
Let $\rho_*$ be a state satisfying $M(\Lambda_1\circ\Lambda_2)=M(\id\otimes\Lambda_1\circ\Lambda_2(\rho_*))-M(\rho_*)$. Then, we obtain
\eq{
M(\Lambda_1\circ\Lambda_2)&=M(\id\otimes\Lambda_1\circ\Lambda_2(\rho_*))-M(\rho_*)\nonumber\\
&=M(\id\otimes\Lambda_1(\id\otimes\Lambda_2(\rho_*)))-M(\id_\otimes\Lambda_2(\rho_*))+M(\id\otimes\Lambda_2(\rho_*))-M(\rho_*)\nonumber\\
&\le M(\Lambda_1)+M(\Lambda_2).
}
\end{proofof}

Furthermore, when the measure for states $M$ satisfies the condition (i), its resource-generating power satisfies the following property:
\begin{description}
\item[(p2) Monotonicity for the quantum comb constructed by a free unitary and its dual] When $M$ satisfies condition (i), for an arbitrary comb $\Xi$ given by $\Xi(\Lambda)=\calE_2\circ\id_{R}\otimes\Lambda\circ\calE_1$, it holds that
\bal
 M(\Xi(\Lambda))\leq M(\Lambda). 
\eal
when $\calE_1$ and $\calE_2$ are written as $\calE_i=\id_{R'_i}\otimes\calV_i$ $(i=1,2)$ where $\calV_i$ is a free unitary or its dual.
\end{description}

\begin{proofof}{the property (p2)} 
\bal
 M(\Xi(\Lambda)) &= M(\calE_2\circ\id\otimes \Lambda\circ \calE_1)\nonumber\\
 &\le M(\id_{R'_2}\otimes\calV_2)+M(\id\otimes\Lambda) + M(\id_{R'_1}\otimes\calV_1)\nonumber\\
 &\le M(\calV_2)+M(\Lambda)+M(\calV_1)\nonumber\\
&\le M(\Lambda),
\eal
where the inequality is because of the property (p1), and the last equality is because $M(\id_{R'_i}\otimes\calV_i)=0$ due to the condition (i), as well as $M(\id_{R}\otimes \Lambda)\le M(\Lambda)$ because the definition of $M(\Lambda)$ involves the optimization over the ancillary system.  
\end{proofof}

\subsection{Implementation cost of quantum channels}

Let us define the amount of resource required to implement a desired non-free operation $\Lambda$ with a free unitary.
We employ the channel-purified distance~\cite{Gilchrist2005distance} as the measure of the accuracy of implementation
\eq{
D_F(\Lambda_1,\Lambda_2) \coloneqq \max_\rho D_F\bigl(\id\otimes\Lambda_1(\rho),\id\otimes\Lambda_2(\rho)\bigr),
}
where $D_F(\rho,\sigma)= \sqrt{1-F(\rho,\sigma)^2}$ and $F(\rho,\sigma)=\Tr\sqrt{\sqrt{\rho}\sigma\sqrt{\rho}}$.
We write $\Lambda_1\sim_\epsilon \Lambda_2$ when $D_F(\Lambda_1,\Lambda_2)\le\epsilon$.

Then, the minimum resource cost for implementing a channel $\Lambda$ from $A$ to $A'$ with error $\epsilon$ is defined as
\eq{
M_c^\epsilon(\Lambda) \coloneqq \min\left\{ M(\eta)| \Lambda(\cdot) \sim_\epsilon \Tr_{B'}[\calU(\cdot\otimes \eta)],\ \calU\in\mathfrak{U}_F\right\},
}
where $\eta$ is a state on an ancillary system $B$, $\calU$ is a unitary operation $\calU(\cdot):=U\cdot U^\dagger$ on $AB$, and $\Tr_{B'}$ is the partial trace over $B'$ with $AB=A'B'$. In other words, the resource cost is the minimum amount of resource attributed to an ancilla state that---together with a free unitary---realises an $\epsilon$-approximate implementation of the target channel $\Lambda$.
When we treat the cost of zero-error implementation, we often use the abbreviation $M_c(\Lambda):=M_c^{\epsilon=0}(\Lambda)$.
For resource theories in which every free operation admits a Stinespring dilation by a free unitary and free states---such as those for energy, athermality, asymmetry, and coherence (when we employ Physically incoherent Operations (PIO) \cite{PIO} as free operations)---the cost $M_c^\epsilon(\Lambda)$ coincides with the minimal resource required when one implements $\Lambda$ using any free operation together with a resourceful ancillary state.

\subsection{Irreversibility of channels and failure probability of state discrimination}

Next, we introduce a measure of irreversibility of quantum channels, originally introduced in Ref.~\cite{TTK} and used in Refs.~\cite{ET2023,Tajima_Takagi2025,NT2024}.
The irreversibility measure we introduce is defined as the minimum average state-recovery error for an ensemble consisting of two mutually orthogonal states.
To be concrete, we prepare a pair of test states $\Omega:=(\rho_1,\rho_2)$ which are orthogonal to each other, i.e., $F(\rho_1,\rho_2)=0$. 
We then write $\Omega_p := (\rho_1,\rho_2)_p$ for the test ensemble in which $\rho_1$ and $\rho_2$ occur with prior probabilities $p_1 = p$ and $p_2 = 1-p$, respectively.
Each test state $\rho_k$ changes via the channel $\Lambda$, and we try to recover the test state with a recovery map $\calR:A'\to A$.
The irreversibility of the channel with respect to the ensemble $\Omega_p$ is defined as the minimum average error
\eq{
\delta(\Lambda,\Omega_p)^2:=\min_{\calR}\sum_k p_k D_F\bigl(\rho_k,\calR\circ\Lambda(\rho_k)\bigr)^2,
}
where $\calR$ runs over all CPTP maps from $A'$ to $A$.

One key property of the irreversibility measure $\delta(\Lambda,\Omega_p)$ is that it equals the minimum failure probability to distinguish the states in the ensemble $\Omega'_{p,\Lambda}:=(\Lambda(\rho_1),\Lambda(\rho_2))_p$.
To describe this fact in concrete terms, we define the failure probability of discriminating the states in an ensemble $\Gamma_q:=(\sigma_1,\sigma_2)_q$ with a two-valued POVM $\mathbb{E}:=\{E_k\}_{k=1,2}$
as
\eq{
\PF(\Gamma_{q},\mathbb{E}):=\sum_{k=1,2}q_k\Tr[(1-E_k)\sigma_k],
}
where $q_1=q$ and $q_2=1-q$.
Then, for the emsemble $\Omega'_{p,\Lambda}:=(\Lambda(\rho_1),\Lambda(\rho_2))_p$, the following relation holds:
\eq{
\delta(\Lambda,\Omega_p)^2=\min_{\mathbb{Q}}\PF(\Omega'_{p,\Lambda},\mathbb{Q}),\label{PF_equal_delta}
}
where $\mathbb{Q}$ runs over two-valued POVMs on $A'$.

Furthermore, the irreversibility measure $\delta(\Lambda,\Omega_p)$ unifies many existing irreversibility measures: it encompasses recovery errors in quantum error correction and Petz-map recovery, and reproduces virtually all established error-and-disturbance metrics for quantum measurements \cite{ET2023}.

\begin{proofof}{\eqref{PF_equal_delta}}
First, we show $\min_{\mathbb{Q}}\PF(\Omega'_{p,\Lambda},\mathbb{Q})\ge\delta(\Lambda,\Omega_p)^2$.
For any two-valued POVM $\mathbb{Q}=\{Q_i\}_{i=1,2}$, we define a recovery map as
\eq{
\calR_{\mathbb{Q}}(...):=\sum_i\Tr[Q_i...]\rho_i.
}
For the convenience of writing, we introduce the complementary index $i_{\neg}$ of $i$ as the index opposite to $i$, i.e., when $i=1$, $i_{\neg}=2$, and when $i=2$, $i_{\neg}=1$. Then, we obtain
\eq{
\delta(\Lambda,\Omega_p)^2&\le\sum_{i}p_iD_F(\rho_i,\calR_{\mathbb{Q}}\circ\Lambda(\rho_i))^2\nonumber\\
&=1-\sum_ip_iF(\rho_i,\Tr[\Lambda^\dagger(Q_i)\rho_i]\rho_i+(1-\Tr[\Lambda^\dagger(Q_i)\rho_i])\rho_{i_{\neg}})^2\nonumber\\
&\stackrel{(a)}{=}1-\sum_ip_i\Tr[\Lambda^\dagger(Q_i)\rho_i]\nonumber\\
&=\sum_ip_i\Tr[(1-\Lambda^\dagger(Q_{i}))\rho_i]\nonumber\\
&=\PF(\Omega'_{p,\Lambda},\mathbb{Q}).
}
Here, in (a), we used 
\eq{
F(\rho_i,\Tr[\Lambda^\dagger(Q_i)\rho_i]\rho_i+(1-\Tr[\Lambda^\dagger(Q_i)\rho_i]\rho_{i_{\neg}})\rho_{i_{\neg}})&=\Tr[\sqrt{\sqrt{\rho_i}(\Tr[\Lambda^\dagger(Q_i)\rho_i]\rho_i+(1-\Tr[\Lambda^\dagger(Q_i)\rho_i]))\sqrt{\rho_i}}]\nonumber\\
&=\sqrt{\Tr[\Lambda^\dagger(Q_i)\rho_i]}.
} 
Therefore, we have obtained $\min_{\mathbb{Q}}\PF(\Omega'_{p,\Lambda},\mathbb{Q})\ge\delta(\Lambda,\Omega,p)^2$.

Next, let us show $\min_{\mathbb{Q}}\PF(\Omega'_{p,\Lambda},\mathbb{Q})\le\delta(\Lambda,\Omega_p)^2$.
We define $\calR_*$ as a recovery CPTP map satisfying
\eq{
\delta(\Lambda,\Omega_p)^2=\sum_{i}p_iD_F(\rho_i,\calR_*\circ\Lambda(\rho_i))^2.
}
Let $\mathbb{P}:=\{P_i\}_{i=1,2}$ be a two-valued PVM which discriminate $\rho_1$ and $\rho_2$ with probability 1. Using it, we define the following channel:
\eq{
\calP_\mathbb{P}(...):=\sum_i\Tr[P_i...]\rho_i.
}
Due to $\calP_\mathbb{P}(\rho_i)=\rho_i$,
\eq{
\delta(\Lambda,\Omega_p)^2\ge\sum_{i}p_iD_F(\rho_i,\calP_\mathbb{P}\circ\calR_*\circ\Lambda(\rho_i))^2.
}
Using $Q^*_i:=\calR^\dagger_*(P_i)$, we can rewrite the right-hand side as 
\eq{
\sum_{i}p_iD_F(\rho_i,\calP_\mathbb{P}\circ\calR_*\circ\Lambda(\rho_i))^2
=
1-\sum_ip_iF(\rho_i,\Tr[\Lambda^\dagger(Q^*_i)\rho_i]\rho_i+(1-\Tr[\Lambda^\dagger(Q^*_i)\rho_i])\rho_{i_{\neg}})^2.
}
Because $\rho_i$ and $\rho_{i_{\neg}}$ are orthogonal to each other, we obtain
\eq{
F(\rho_i,\Tr[\Lambda^\dagger(Q^*_i)\rho_i]\rho_i+(1-\Tr[\Lambda^\dagger(Q^*_i)\rho_i])\rho_{i_{\neg}})^2&=\Tr[\sqrt{\sqrt{\rho_i}(\Tr[\Lambda^\dagger(Q^*_i)\rho_i]\rho_i+(1-\Tr[\Lambda^\dagger(Q^*_i)\rho_i])\rho_{i_{\neg}})\sqrt{\rho_i}}]^2\nonumber\\
&=\Tr[\sqrt{\sqrt{\rho_i}(\Tr[\Lambda^\dagger(Q^*_i)\rho_i]\rho_i)\sqrt{\rho_i}}]^2\nonumber\\
&=\Tr[\Lambda^\dagger(Q^*_i)\rho_i].
}
Hence, we obtain
\eq{
\delta(\Lambda,\Omega_p)^2&\ge1-\sum_ip_i\Tr[\Lambda^\dagger(Q^*_i)\rho_i]\nonumber\\
&=\sum_ip_i\Tr[(1-\Lambda^\dagger(Q^*_i))\rho_i]\nonumber\\
&=\PF(\Omega'_{p,\Lambda},\mathbb{Q}^*).
}
Therefore, $\min_{\mathbb{Q}}\PF(\Omega'_{p,\Lambda},\mathbb{Q})\le\delta(\Lambda,\Omega_p)^2$ holds, and thus we obtain \eqref{PF_equal_delta}
\end{proofof}

\subsection{Resource gain and required channel power of state discrimination}
Next, we quantify both the resource gain obtained by state discrimination and the channel power required for such discrimination. 
To achieve this, we use a two-level quantum register system $K$ to store the classical outcome of the measurement, along with a particular orthonormal basis $\{\ket{k}\}_{k=1,2}$ on that register.
Below, we refer to the tuple $(K,\{\ket{k}\}_{k=1,2})$ as $\mathscr{K}$.

Using the notion of the register, we can define the resource gain and the required channel power to discriminate states in a two-valued ensemble.
We take a two-level register $\mathscr{K}=(K,\{\ket{k}\}_{k=1,2})$ and a two-valued ensemble $\Gamma_q:=(\sigma_1,\sigma_2)_q$ on a quantum system $S$, where $\sigma_1$ and $\sigma_2$ are orthogonal to each other.
Then, the resource gain via the discrimination of the states in $\Gamma_q$ is defined as
\eq{
\Mem(\Gamma_q|\mathscr{K}):=\max_{\sigma,\mathbb{E}}\left\{M_\sigma(\Lambda_{\mathbb{E}})\middle|\Lambda_{\mathbb{E}}(\sigma)=\sum_{k=1,2}q_k\sigma_k\otimes\ket{k}\bra{k}_K\right\} \label{eq:Memdef},
}
where the maximization ranges over $\sigma$ and the two-valued projective measurement channel 
$\Lambda_{\mathbb{E}}(\cdot):=\sum_{k=1,2}E_k\cdot E_k\otimes\ket{k}\bra{k}_K$ from $S$ to $SK$ whose two-valued PVM $\mathbb{E}:=\{E_k\}$ describes the test states $\sigma_1$ and $\sigma_2$ with the probability 1. 
Because of the definition, the quantity $\Mem(\Omega_p|\mathscr{K})$ describes the maximum resource gain from distinguishing the states $\rho_1$ and $\rho_2$ in $\Omega_p=(\rho_1,\rho_2)_p$.

Similarly, we define the channel power required to optimally distinguish the states in a two-valued ensemble $\Gamma'_{q}=(\sigma'_1,\sigma'_2)_q$ on whose $\sigma'_1$ and $\sigma'_2$ are not necessarily mutually orthogonal as
\eq{
\Mcos(\Gamma'_{q}|\mathscr{K}):=\min_{\mathbb{E}}\left\{M(\Lambda_{\mathbb{E}})\middle|\PF(\Gamma'_{q},\mathbb{E})=\min_{\mathbb{F}}\PF(\Gamma'_{q},\mathbb{F})\right\} ,\label{eq:Mcosdef}
}
where the minimization ranges over the two-valued POVM measurement channel $\Lambda_{\mathbb{E}}(\cdot):=\sum_{k=1,2}\sqrt{E_k}\cdot \sqrt{E_k}\otimes\ket{k}\bra{k}_K$.
Due to \eqref{PF_equal_delta}, the required power $\Mcos(\Omega'_{p,\Lambda}|\mathscr{K})$ satisfies
\eq{
\Mcos(\Omega'_{p,\Lambda}|\mathscr{K})=\min_{\mathbb{Q}}\left\{M(\Lambda_{\mathbb{Q}})\middle|\PF(\Omega'_{p,\Lambda},\mathbb{Q})=\delta(\Lambda,\Omega_p)^2\right\}.
}
This quantity represents the minimal channel power required to optimally discriminate between $\Lambda(\rho_1)$ and $\Lambda(\rho_2)$ in $\Omega'_{p,\Lambda}$.

To close this subsection, we establish the well-definedness of the quantities~\eqref{eq:Memdef} and \eqref{eq:Mcosdef}, i.e., we show that the maximization and minimizations in \eqref{eq:Memdef} and \eqref{eq:Mcosdef} are attainable.
We first consider the maximization in \eqref{eq:Memdef}.
For this, it suffices to show that the feasible region is compact and the objective function is contiunous.
Let $\mathcal{S}_S$ and $\mathcal{P}_S$ denote the sets of states (density operators) and of two-valued PVMs on $S$, respectively.
Due to the finite-dimensionality of $S$, the sets $\mathcal{S}_S$ and $\mathcal{P}_S$ are compact, and hence so is the product set $\mathcal{S}_S \times \mathcal{P}_S$.
Since the map $\mathcal{S}_S \times \mathcal{P}_S \ni (\sigma,\mathbb{E}) \mapsto \Lambda_\mathbb{E}(\sigma)$ is continuous, the feasible region of \eqref{eq:Memdef} is compact.
Moreover, since $\mathcal{S}_S \ni \rho \mapsto M(\rho)$ is continuous, we can see that $(\sigma, \Lambda) \mapsto M_\sigma (\Lambda)$ is continuous.
Therefore, since $\mathcal{P}_S \ni \mathbb{E} \mapsto \Lambda_\mathbb{E}$ is continuous, the objective function of \eqref{eq:Memdef} is continuous.

Next, we consider the minimization in \eqref{eq:Mcosdef}.
Since the set of two-valued POVMs on $S$ is compact and the function $\PF(\Gamma_q',\cdot)$ is continuous, the feasible region of \eqref{eq:Mcosdef} is nonempty and compact.
Let $\mathcal{CH}_S$ denote the set of quantum channels on $S$.
Since $\mathcal{CH}_S \ni \Lambda \mapsto M(\id_R \otimes \Lambda (\rho) )-M(\rho)$ is continuous for any fixed $R$ and $\rho \in \mathcal{S}_{RS}$, the function $\mathcal{CH}_S \ni \Lambda \mapsto M(\Lambda) = \sup_{R, \rho \in \mathcal{S}_{RS}}M(\id_R \otimes \Lambda (\rho) )-M(\rho)$ is lower semicontinuous.
From the continuity of $\mathbb{Q}\mapsto \Lambda_\mathbb{Q}$, this implies the lower semicontinuity of the objective function $\mathbb{Q}\mapsto M(\Lambda_\mathbb{Q})$.
Therefore, the minimum in \eqref{eq:Mcosdef} is attained.

\section{Lower bounds for the costs of arbitrary channels: Resource-irreversibility tradeoff and others}
\subsection{Useful lemma: Limitation of the conversion of measurement channels}
In this subsection, we show a lemma which limits the conversion of measurement channels.
This lemma is independent of resource theory: it describes a general feature of quantum measurements.
\begin{lemma}\label{Lemm:key_nonresource}
Let $A$ and $A'$ be two quantum systems.
For a two-valued PVM $\mathbb{P}:=\{P_i\}^{2}_{i=1}$ on $A$ and a two-valued POVM $\mathbb{Q}:=\{Q_i\}^{2}_{i=1}$, we define the error of approximation of $\mathbb{P}$ by $\Lambda^{\dagger}(\mathbb{Q}):=\{\Lambda^\dagger(Q_i)\}^{2}_{i=1}$ for a quantum system $\rho$ on $A$ as
\eq{
\epsilon(\rho,\mathbb{P},\mathbb{Q},\Lambda)^2:=\sum^{2}_{i=1}\ex{1-\Lambda^\dagger(Q_i)}_{P_i\rho P_i}.\label{def_eps}
}
We also assume that $\Lambda$ is realized by other systems $B$ and $B'$ satisfying $AB=A'B'$, a state $\eta^B$ on $B$ and a unitary $V$ on $AB$ as $\Lambda(\cdot)=\Tr_{B'}[V\cdot\otimes\eta^B V^\dagger]$.
Then, the following inequality holds:
\eq{
\|(\mathrm{id}^{R}\otimes\calV^{\dagger})\circ(\mathrm{id}^{B'}\otimes\Lambda_{\mathbb{Q}}^{A'})\circ\calV(\rho\otimes\eta)-\Lambda_{\mathbb{P}}^{A}(\rho)\otimes\eta^B\|_1&\le f(\epsilon(\rho,\mathbb{P},\mathbb{Q},\Lambda)),\label{target_lemma}\\
f(x)&:=4x+x^2.
}
Here, $\calV(...):=V...V^\dagger$, and $\Lambda_{\mathbb{P}}$ and $\Lambda_{\mathbb{Q}}$ are defined as
\eq{
\Lambda_{\mathbb{P}}(...)&:=\sum_iP_i(...)P_i\otimes\ket{j}\bra{j}_R,\\
\Lambda_{\mathbb{Q}}(...)&:=\sum_j\sqrt{Q_j}...\sqrt{Q_j}\otimes\ket{j}\bra{j}_R.
}
\end{lemma}

\begin{proof}
We firstly consider the case where $\rho$ and $\eta$ are pure states.
Let us introduce the following symbols:
\eq{
\ket{\Psi'}&:=V\ket{\rho}\otimes\ket{\eta},\\
\ket{\Psi'_i}&:=VP_i\ket{\rho}\otimes\ket{\eta}.
}
With using the symbols $\Psi':=\ket{\Psi'}\bra{\Psi'}$ and $\Psi'^{ij}:=\ket{\Psi'_i}\bra{\Psi'_j}$, we obtain
\eq{
\Psi'=\Psi'^{11}+\Psi'^{22}+\Psi'^{12}+\Psi'^{21}.
}
We also introduce
\eq{
\Psi''_{\mathrm{ideal}}&:=\calV\circ\Lambda_{\mathbb{P}}(\rho)\otimes\eta\nonumber\\
&=\sum^{2}_{i=1}\Psi'_{ii}\otimes\ket{i}\bra{i}_R\label{ideal_decom}.
}
Note that 
\eq{
\Lambda_{\mathbb{Q}}(\Psi')=(\mathrm{id}^{B'}\otimes\Lambda_{\mathbb{Q}}^{A'})\circ\calV(\rho\otimes\eta).
}
Thus, since the unitary $\calV^\dagger$ does not change the trace norm,  the following relation is equivalent to \eqref{target_lemma}:
\eq{
\|\Lambda_{\mathbb{Q}}(\Psi')-\Psi''_{\mathrm{ideal}}\|_1\le f(\epsilon(\rho,\mathbb{P},\mathbb{Q},\Lambda)).\label{keyeq_lemm}
}

Let us show \eqref{keyeq_lemm}. Due to $\|\sum_i\sigma_i\otimes\ket{i}\bra{i}-\sum_i\sigma'_i\otimes\ket{i}\bra{i}\|_1=\sum_i\|\sigma_i-\sigma'_i\|_1$, $\Lambda_{\mathbb{Q}}(\Psi')=\sum^{2}_{i=1}\sqrt{Q_i}\Psi'\sqrt{Q_i}\otimes\ket{i}\bra{i}_R$ and \eqref{ideal_decom}, we obtain
\eq{
\|\mathbb{Q}(\Psi')-\Psi''_{\mathrm{ideal}}\|_1= \sum^{2}_{i=1}\|\sqrt{Q_i}\Psi'\sqrt{Q_i}-\Psi'^{ii}\|_1.
}
Due to the triangle inequality, we get
\eq{
\|\sqrt{Q_i}\Psi'\sqrt{Q_i}-\Psi'^{ii}\|_1\le\|\sqrt{Q_i}\Psi'^{ii}\sqrt{Q_i}-\Psi'^{ii}\|_1+\|\sqrt{Q_i}\Psi'\sqrt{Q_i}-\sqrt{Q_i}\Psi'^{ii}\sqrt{Q_i}\|_1\label{target_evaluate}.
}
Let us evaluate each term in the right-hand side of the above inequality.
Again, for the convenience of writing, we introduce the complementary index $i_{\neg}$ of $i$ as the index opposite to $i$, i.e., when $i=1$, $i_{\neg}=2$, and when $i=2$, $i_{\neg}=1$. 

Let us evaluate the first term of the right-hand side of \eqref{target_evaluate}.
We define the probability $p_i:=\Tr[P_i\rho]$ and the conditional error $\epsilon^2_i:=\ex{\Lambda^\dagger(Q_{i_{\neg}})}_{\frac{P_i\rho P_i}{p_i}}$.
Then, we obtain
\eq{
\frac{1}{p_i}\Tr[\Psi'^{ii}Q_i]&=\frac{1}{p_i}\Tr[\Lambda(P_i\rho P_i)Q_i]\nonumber\\
&=\ex{\Lambda^\dagger(Q_i)}_{\frac{P_i\rho P_i}{p_i}}\nonumber\\
&=1-\epsilon^2_i.
}
Therefore, due to the gentle measurement lemma \cite{wilde_text},
\eq{
\left\|\frac{1}{p_i}\Psi'^{ii}-\sqrt{Q_i}\frac{\Psi'^{ii}}{p_i}\sqrt{Q_i}\right\|_1\le2\sqrt{\epsilon^2_i}=2\epsilon_i.
}
Hence, the first term is evaluated as
\eq{
\|\sqrt{Q_i}\Psi'^{ii}\sqrt{Q_i}-\Psi'^{ii}\|_1\le2p_i\epsilon_i.
}

Next, we evaluate the second term of the right-hand side of \eqref{target_evaluate}.
Due to $\ket{\Psi'}=\ket{\Psi'_1}+\ket{\Psi'_2}$, we obtain
\eq{
\sqrt{Q_i}\ket{\Psi'}=\sqrt{Q_i}\ket{\Psi'_1}+\sqrt{Q_i}\ket{\Psi'_2}.
}
In general, for two non-normalized pure states $\ket{\alpha}$ and $\ket{\beta}$, the state $\ket{\gamma}:=\ket{\alpha}+\ket{\beta}$ satisfies
\eq{
\|\gamma-\alpha\|_1\le2\sqrt{\braket{\alpha|\alpha}\braket{\beta|\beta}}+\braket{\beta|\beta},
}
which is derived as
\eq{
\|\gamma-\alpha\|_1&=\|\ket{\beta}\bra{\beta}+\ket{\alpha}\bra{\beta}+\ket{\beta}\bra{\alpha}\|_1\nonumber\\
&\le\|\ket{\beta}\bra{\beta}\|_1+2\|\ket{\alpha}\bra{\beta}\|_1\nonumber\\
&=\Tr[\ket{\beta}\bra{\beta}]+2\Tr[\sqrt{\braket{\alpha|\alpha}\ket{\beta}\bra{\beta}}]\nonumber\\
&=2\sqrt{\braket{\alpha|\alpha}\braket{\beta|\beta}}+\braket{\beta|\beta}.
}
Therefore, the second term is evaluated as
\eq{
\|\sqrt{Q_i}\Psi'\sqrt{Q_i}-\sqrt{Q_i}\Psi'^{ii}\sqrt{Q_i}\|_1&\le2\sqrt{
\bra{\Psi'^{i}}Q_i\ket{\Psi'^{i}}}\sqrt{
\bra{\Psi'^{i_{\neg}}}Q_i\ket{\Psi'^{i_{\neg}}}}
+\bra{\Psi'^{i_{\neg}}}Q_i\ket{\Psi'^{i_{\neg}}}\nonumber\\
&\le2\sqrt{
\ex{\Psi'^{i}|\Psi'^{i}}}\sqrt{
\bra{\Psi'^{i_{\neg}}}Q_i\ket{\Psi'^{i_{\neg}}}}
+\bra{\Psi'^{i_{\neg}}}Q_i\ket{\Psi'^{i_{\neg}}}\nonumber\\
&=2\sqrt{p_i}\sqrt{\Tr[\Psi'^{i_{\neg}i_{\neg}}Q_i]}
+\Tr[\Psi'^{i_{\neg}i_{\neg}}Q_i]\nonumber\\
&=2\sqrt{p_i}\sqrt{\Tr[\Lambda(P_{i_{\neg}}\rho P_{i_{\neg}})Q_i]}
+\Tr[\Lambda(P_{i_{\neg}}\rho P_{i_{\neg}})Q_i]\nonumber\\
&=2\sqrt{p_ip_{i_{\neg}}\epsilon^2_{i_{\neg}}}+p_{i_{\neg}}\epsilon^2_{i_{\neg}}.
}

Because of \eqref{def_eps}, $\sum_ip_i\epsilon^2_i=\epsilon(\rho,\mathbb{P},\mathbb{Q},\Lambda)^2$ holds.
Therefore, we obtain \eqref{target_lemma} as
\eq{
\|\Lambda_{\mathbb{Q}}(\Psi')-\Psi''_{\mathrm{ideal}}\|_1&\le2\sum_ip_i\epsilon_i+2\sum_i\sqrt{p_i(1-p_i)\epsilon^2_i}+\sum_ip_i\epsilon^2_i\nonumber\\
&\le2\sqrt{\sum_ip_i\epsilon^2_i}+2\sqrt{\sum_i(1-p_i)}\sqrt{\sum_ip_i\epsilon^2_i}+\epsilon(\rho,\mathbb{P},\mathbb{Q},\Lambda)^2\nonumber\\
&\le4\epsilon(\rho,\mathbb{P},\mathbb{Q},\Lambda)+\epsilon(\rho,\mathbb{P},\mathbb{Q},\Lambda)^2\nonumber\\
&=f(\epsilon(\rho,\mathbb{P},\mathbb{Q},\Lambda)).
}

Next, we consider the case where $\rho$ and $\eta$ are not necessarily pure.
In this case, we can take their purifications $\ket{\rho}_{AR_A}$ and $\ket{\eta}_{AR_B}$ with proper reference systems $R_A$ and $R_B$.
And, we just repeat the above discussion where we substitute $AR_A$, $BR_B$, $I^{R_AR_B}\otimes V$, $\id^{R_A}\otimes\Lambda$, $\id^{R_A}\otimes\Lambda_{\mathbb{P}}$, $\id^{R_A}\otimes\Lambda_{\mathbb{Q}}$, $\ket{\rho}_{AR_A}$ and $\ket{\eta}_{AR_B}$ for  $A$, $B$, $V$, $\Lambda$, $\Lambda_{\mathbb{P}}$, $\Lambda_{\mathbb{Q}}$, $\ket{\rho}$ and $\ket{\eta}$.
Then, we obtain the target relation \eqref{target_lemma} again.
\end{proof}

\subsection{Lower bounds of costs of channels}

\subsubsection{The case of $M$ satisfying the conditions (i)--(iii)}
We first treat the case where the measure $M$ satisfies the conditions (i)--(iii).
\begin{theorem}[Resource-Irreversibility tradeoff]\label{Thm:RItradeoff_SM}
Let $A$ and $A'$ be quantum systems.
For any channel $\Lambda$ from $A$ to $A'$, any test ensemble $\Omega_p:=(\rho_1,\rho_2)_p$ whose $\rho_1$ and $\rho_2$ are mutually orthogonal, any two-level register $\mathscr{K}:=(K,\{\ket{k}\}_{k=1,2})$, and any resource measure $M$ satisfying the conditions (i)--(iii), the following inequalities hold:
\eq{
M_c(\Lambda)&\ge g(a_A,b_A)\frac{|\Mem(\Omega_p|\mathscr{K})-\Mcos(\Omega'_{p,\Lambda}|\mathscr{K})|^{\frac{a_A+b_A}{a_A+b_A-1}}_+}{(K_Af(\delta(\Lambda,\Omega_p))^{b_A})^{\frac{1}{a_A+b_A-1}}}\nonumber\\
&-|\Mem(\Omega_p|\mathscr{K})-\Mcos(\Omega'_{p,\Lambda}|\mathscr{K})|_+-c_{\max},\enskip (\mbox{for }a_A<\infty)\label{eq:RItradeoff_SM}\\
M_c(\Lambda)&\ge\frac{|\Mem(\Omega_p|\mathscr{K})-\Mcos(\Omega'_{p,\Lambda}|\mathscr{K})|_+}{a'_A+b_A}\log\frac{|\Mem(\Omega_p|\mathscr{K})-\Mcos(\Omega'_{p,\Lambda}|\mathscr{K})|_+}{K_Af(\delta(\Lambda,\Omega_p))^{b_A}(a'_A+b_A)}\nonumber\\
&-|\Mem(\Omega_p|\mathscr{K})-\Mcos(\Omega'_{p,\Lambda}|\mathscr{K})|_+\left(1+\frac{1}{a'_A+b_A}\right)
-c_{\max},\enskip(\mbox{for }a_A=\infty)\label{eq:RItradeoff_ainf_SM}
}
where $g(a_A,b_A)=\left(\frac{1}{a_A+b_A}\right)^{\frac{1}{a_A+b_A-1}}\left(1-\frac{1}{a_A+b_A}\right)$, $f(x)=4x+x^2$, and $|x|_+$ returns $x$ when $x$ is positive, and returns $0$ otherwise. 
Furthermore, when $|\Mem(\Omega_p|\mathscr{K})-\Mcos(\Omega'_{p,\Lambda}|\mathscr{K})|_+>0$ and $\delta(\Lambda,\Omega_p)=0$, i.e. when the implementation cost $M_c(\Lambda)$ diverges, the following inequalities hold:
\eq{
M^\epsilon_c(\Lambda)\ge& g(a_A,b_A)\frac{|\Mem(\Omega_p|\mathscr{K})-\Mcos(\Omega'_{p,\Lambda}|\mathscr{K})|^{\frac{a_A+b_A}{a_A+b_A-1}}_+}{(K_Af(\epsilon)^{b_A})^{\frac{1}{a_A+b_A-1}}}-|\Mem(\Omega_p|\mathscr{K})-\Mcos(\Omega'_{p,\Lambda}|\mathscr{K})|_+-c_{\max},\enskip (\mbox{for }a_A<\infty)\label{eq:REtradeoff_SM}\\
M^\epsilon_c(\Lambda)\ge&\frac{|\Mem(\Omega_p|\mathscr{K})-\Mcos(\Omega'_{p,\Lambda}|\mathscr{K})|_+}{a'_A+b_A}\log\frac{|\Mem(\Omega_p|\mathscr{K})-\Mcos(\Omega'_{p,\Lambda}|\mathscr{K})|_+}{K_Af(\epsilon)^{b_A}(a'_A+b_A)}\nonumber\\
&-|\Mem(\Omega_p|\mathscr{K})-\Mcos(\Omega'_{p,\Lambda}|\mathscr{K})|_+\left(1+\frac{1}{a'_A+b_A}\right)
-c_{\max},\enskip(\mbox{for }a_A=\infty)\label{eq:REtradeoff_ainf_SM}
}
\end{theorem}

If $\Mem(\Omega_p|\mathscr{K})-\Mcos(\Omega'_{p,\Lambda}|\mathscr{K})\le0$ holds for every test ensemble $\Omega_p$, the cost and irreversibility need not be inversely proportional. Even in this situation, the following bound always holds:
\begin{theorem}\label{Thm:RPtradeoff_SM}
Let $A$ and $A'$ be quantum systems.
For any channel $\Lambda$ from $A$ to $A'$, any test ensemble $\Omega_p:=(\rho_1,\rho_2)_p$ whose $\rho_1$ and $\rho_2$ are mutually orthogonal,  any two-level register $\mathscr{K}:=(K,\{\ket{k}\}_{k=1,2})$, any resource measure $M$ satisfying the conditions (i)--(iii), and any two-valued POVM $\mathbb{Q}:=\{Q_k\}^{2}_{k=1}$ on $A'$ and its measurement channel $\Lambda_{\mathbb{Q}}(\cdot):=\sum_k\sqrt{Q_k}\cdot\sqrt{Q_k}\otimes\ket{k}\bra{k}_K$, the following inequality holds:
\eq{
M_c(\Lambda)\ge& g(a_A,b_A)\frac{|\Mem(\Omega_p|\mathscr{K})-M(\Lambda_{\mathbb{Q}})|^{\frac{a_A+b_A}{a_A+b_A-1}}_+}{\left(K_Af\left(\sqrt{\PF(\Omega'_{p,\Lambda},\mathbb{Q})}\right)^{b_A}\right)^{\frac{1}{a_A+b_A-1}}}-|\Mem(\Omega_p|\mathscr{K})-M(\Lambda_{\mathbb{Q}})|_+-c_{\max},\enskip(\mbox{for }a_A=\infty)\label{eq:RPtradeoff_SM}\\
M_c(\Lambda)\ge&\frac{|\Mem(\Omega_p|\mathscr{K})-M(\Lambda_{\mathbb{Q}})|_+}{a'_A+b_A}\log\frac{|\Mem(\Omega_p|\mathscr{K})-M(\Lambda_{\mathbb{Q}})|_+}{K_Af\left(\sqrt{\PF(\Omega'_{p,\Lambda},\mathbb{Q})}\right)^{b_A}(a'_A+b_A)}\nonumber\\
&-|\Mem(\Omega_p|\mathscr{K})-M(\Lambda_{\mathbb{Q}})|_+\left(1+\frac{1}{a'_A+b_A}\right)
-c_{\max},\enskip(\mbox{for }a_A=\infty).\label{eq:RPtradeoff_ainf_SM}
}
\end{theorem}

Theorem \ref{Thm:RPtradeoff_SM} restricts the cost of a channel that approximately causes a pairwise reversible conversion. 
\begin{corollary}\label{Coro:REtradeoff_conversion_SM}
Let $A$ and $A'$ be quantum systems.
Let a CPTP map $\Lambda:A\rightarrow A'$ approximately convert an orthogonal pair $\Omega:=(\rho_1, \rho_2)$ on $A$ to an orthogonal pair $\Omega'':=(\rho''_1, \rho''_2)$ on $A'$ as
\eq{
D_F(\Lambda(\rho_i),\rho''_i)\le\epsilon,\enskip (i=1,2).\label{cond_approx_SM}
}
Then, the following holds for any resource measure satisfying the conditions (i)--(iii) and any two-level register $\mathscr{K}:=(K,\{\ket{k}\}_{k=1,2})$:
\eq{
M_c(\Lambda)&\ge g(a_A,b_A)\frac{|\Mem(\Omega|\mathscr{K})-\Mcos(\Omega''|\mathscr{K})|^{\frac{a_A+b_A}{a_A+b_A-1}}_+}{(K_Af(\epsilon)^{b_A})^{\frac{1}{a_A+b_A-1}}}-|\Mem(\Omega|\mathscr{K})-\Mcos(\Omega''|\mathscr{K})|_+-c_{\max},\enskip(\mbox{for }a_A<\infty)\label{eq:REtradeoff_conversion_SM}\\
M_c(\Lambda)\ge&\frac{|\Mem(\Omega|\mathscr{K})-\Mcos(\Omega''|\mathscr{K})|_+}{a'_A+b_A}\log\frac{|\Mem(\Omega|\mathscr{K})-\Mcos(\Omega''|\mathscr{K})|_+}{K_Af(\epsilon)^{b_A}(a'_A+b_A)}\nonumber\\
&-|\Mem(\Omega|\mathscr{K})-\Mcos(\Omega''|\mathscr{K})|_+\left(1+\frac{1}{a'_A+b_A}\right)
-c_{\max},\enskip(\mbox{for }a_A=\infty)\label{eq:REtradeoff_conversion_SM_ainf}
}
Here, $\Mem(\Omega|\mathscr{K}):=\max_p \Mem(\Omega_p|\mathscr{K})$ and $\Mcos(\Omega''|\mathscr{K}):=\min_p \Mem(\Omega''_p|\mathscr{K})$.
\end{corollary}

\textbf{Deriving \eqref{eq:RItradeoff}, \eqref{eq:REtradeoff}, \eqref{eq:RPtradeoff}, and \eqref{eq:REtradeoff_conversion} in the main text from Theorems \ref{Thm:RItradeoff_SM} and \ref{Thm:RPtradeoff_SM} and Corollary \ref{Coro:REtradeoff_conversion_SM}:} When $a_S=b_S=1$, \eqref{eq:RItradeoff_SM} reduces to \eqref{eq:RItradeoff} in the main text as
\eq{
M_c(\Lambda)&\ge\frac{|\Mem(\Omega_p|\mathscr{K})-\Mcos(\Omega'_{p,\Lambda}|\mathscr{K})|^2_+}{16K_A\delta(\Lambda,\Omega_p)\left(1+\frac{1}{4}\delta(\Lambda,\Omega_p)\right)}-|\Mem(\Omega_p|\mathscr{K})-\Mcos(\Omega'_{p,\Lambda}|\mathscr{K})|_+-c_{\max}\nonumber\\
&\ge\frac{|\Mem(\Omega_p|\mathscr{K})-\Mcos(\Omega'_{p,\Lambda}|\mathscr{K})|^2_+}{16K_A\delta(\Lambda,\Omega_p)}-\frac{|\Mem(\Omega_p|\mathscr{K})-\Mcos(\Omega'_{p,\Lambda}|\mathscr{K})|^2_+}{64K_A}-|\Mem(\Omega_p|\mathscr{K})-\Mcos(\Omega'_{p,\Lambda}|\mathscr{K})|_+-c_{\max}\nonumber\\
&=\frac{|\Mem(\Omega_p|\mathscr{K})-\Mcos(\Omega'_{p,\Lambda}|\mathscr{K})|^2_+}{16K_A\delta(\Lambda,\Omega_p)}-c_{\max}'\label{eq:RItradeoff2_SM}.
}
Here we used $\frac{1}{1+x}\ge1-x$ for $x\ge0$ in the second line and $c'_{\max}:=c_{\max}+|\Mem(\Omega_p|\mathscr{K})-\Mcos(\Omega'_{p,\Lambda}|\mathscr{K})|_++\frac{|\Mem(\Omega_p|\mathscr{K})-\Mcos(\Omega'_{p,\Lambda}|\mathscr{K})|^2_+}{64K_A}$ in the third line.
In the same manner, when $a_S=b_S=1$, \eqref{eq:REtradeoff_SM},  \eqref{eq:RPtradeoff_SM} and \eqref{eq:REtradeoff_conversion_SM} reduce to \eqref{eq:REtradeoff}, \eqref{eq:RPtradeoff}, and \eqref{eq:REtradeoff_conversion}, respectively.

To derive Theorems \ref{Thm:RItradeoff_SM} and \ref{Thm:RPtradeoff_SM} and Corollary \ref{Coro:REtradeoff_conversion_SM}, we firstly show the following lemma:
\begin{lemma}\label{Lemm:key_resource}
Let $A$ and $A'$ be quantum systems.
Consider a CPTP map $\Lambda:A\rightarrow A'$. 
For a two-valued PVM $\mathbb{P}:=\{P_i\}^{2}_{i=1}$ on $A$ and a two-valued POVM $\mathbb{Q}:=\{Q_i\}^{2}_{i=1}$, we define the error of approximation of $\mathbb{P}$ by $\Lambda^{\dagger}(\mathbb{Q}):=\{\Lambda^\dagger(Q_i)\}^{2}_{i=1}$ for a quantum system $\rho$ on $A$ as
\eq{
\epsilon(\rho,\mathbb{P},\mathbb{Q},\Lambda)^2:=\sum^{2}_{i=1}\ex{1-\Lambda^\dagger(Q_i)}_{P_i\rho P_i}.
}
We also define a two-level register $\mathscr{K}=(K,\{\ket{k}\}_{k=1,2})$ and two quantum channels $\Lambda_{\mathbb{P}}:A\rightarrow AK$ and $\Lambda_{\mathbb{Q}}:A'\rightarrow A'K$ as
\eq{
\Lambda_{\mathbb{P}}(\cdot)&:=\sum_kP_k\cdot P_k\otimes\ket{k}\bra{k}_K,\\
\Lambda_{\mathbb{Q}}(\cdot)&:=\sum_k\sqrt{Q_k}\cdot\sqrt{Q_k}\otimes\ket{k}\bra{k}_K.
}
Then, for any resource measure $M$ satisfying the conditions (i)--(iii), the following inequalities hold:
\eq{
M_c(\Lambda)&\ge g(a_A,b_A)\frac{|M_\rho(\Lambda_{\mathbb{P}})-M(\Lambda_{\mathbb{Q}})|^{\frac{a_A+b_A}{a_A+b_A-1}}_+}{(K_Af(\epsilon(\rho,\mathbb{P},\mathbb{Q},\Lambda))^{b_A})^{\frac{1}{a_A+b_A-1}}}-|M_\rho(\Lambda_{\mathbb{P}})-M(\Lambda_{\mathbb{Q}})|_+-c_{\max},\enskip (\mbox{for }a_A<\infty)\label{eq:key_resource_SM}\\
M_c(\Lambda)&\ge\frac{|M_\rho(\Lambda_{\mathbb{P}})-M(\Lambda_{\mathbb{Q}})|_+}{a'_A+b_A}\log\frac{|M_\rho(\Lambda_{\mathbb{P}})-M(\Lambda_{\mathbb{Q}})|_+}{K_Af(\epsilon(\rho,\mathbb{P},\mathbb{Q},\Lambda))^{b_A}(a'_A+b_A)}\nonumber\\
&-|M_\rho(\Lambda_{\mathbb{P}})-M(\Lambda_{\mathbb{Q}})|_+\left(1+\frac{1}{a'_A+b_A}\right)
-c_{\max},\enskip(\mbox{for }a_A=\infty)\label{eq:key_resource_ainf_SM}
}
where $g(a_A,b_A)=\left(\frac{1}{a_A+b_A}\right)^{\frac{1}{a_A+b_A-1}}\left(1-\frac{1}{a_A+b_A}\right)$, $f(x)=4x+x^2$, and $|x|_+$ returns $x$ when $x$ is positive, and returns $0$ otherwise. 
\end{lemma}

\begin{proofof}{Lemma \ref{Lemm:key_resource}}
Let us take an arbitrary implementation $(\eta,V)$ of $\Lambda$ satisfying
\eq{
\Lambda(\cdot)=\Tr_{B'}[\calV(\cdot\otimes\eta)],
}
where $\calV(\cdot):=V\cdot V^\dagger$ and $V$ is a free unitary.
Then, to derive the inequalities \eqref{eq:key_resource_SM} and \eqref{eq:key_resource_ainf_SM}, we only have to derive
\eq{
M_c(\Lambda)&\ge g(a_A,b_A)\frac{|M_\rho(\Lambda_{\mathbb{P}})-M(\Lambda_{\mathbb{Q}})|^{\frac{a_A+b_A}{a_A+b_A-1}}_+}{(K_Af(\epsilon(\rho,\mathbb{P},\mathbb{Q},\Lambda))^{b_A})^{\frac{1}{a_A+b_A-1}}}-|M_\rho(\Lambda_{\mathbb{P}})-M(\Lambda_{\mathbb{Q}})|_+-c_{\max},\enskip (\mbox{for }a_A<\infty)\label{eq:key_resource_SM_prime}\\
M_c(\Lambda)&\ge\frac{|M_\rho(\Lambda_{\mathbb{P}})-M(\Lambda_{\mathbb{Q}})|_+}{a'_A+b_A}\log\frac{|M_\rho(\Lambda_{\mathbb{P}})-M(\Lambda_{\mathbb{Q}})|_+}{K_Af(\epsilon(\rho,\mathbb{P},\mathbb{Q},\Lambda))^{b_A}(a'_A+b_A)}\nonumber\\
&-|M_\rho(\Lambda_{\mathbb{P}})-M(\Lambda_{\mathbb{Q}})|_+\left(1+\frac{1}{a'_A+b_A}\right)
-c_{\max},\enskip(\mbox{for }a_A=\infty).\label{eq:key_resource_ainf_SM_prime}
}

Let us define
$\Lambda'_{\mathbb{P}}:AB\rightarrow ABK$ as:
\eq{
\Lambda'_{\mathbb{P}}(...):=(\calV^\dagger\otimes\mathrm{id}^{K})\circ(\Lambda_{\mathbb{Q}}\otimes\mathrm{id}^{B'})\circ\calV(...).
}
Here $\calV^\dagger(\cdot):=V^\dagger\cdot V$.
Because this channel is obtained by a comb constructed with the free unitary $V$ and its dual on  $\Lambda_{\mathbb{Q}}$, the following inequality holds due to the property (P2):
\eq{
M(\Lambda_{\mathbb{Q}})\ge M(\Lambda'_{\mathbb{P}}).\label{keya}
}
Also, due to Lemma \ref{Lemm:key_nonresource}, we obtain
\eq{
\|\Lambda'_{\mathbb{P}}(\rho\otimes\eta)-\Lambda_{\mathbb{P}}(\rho)\otimes\eta\|_1\le f(\epsilon(\rho,\mathbb{P},\mathbb{Q},\Lambda)).\label{keyb1}
} 

Here, let us prepare $m$ copies of $A$, and name them $A_1$,...,$A_m$. We also refer to a channel that performs $\Lambda'_{\mathbb{P}}$ on $A_k$ and $B$ as $\Lambda'_{\mathbb{P}:k}$.
We define 
\eq{
\Lambda_{\mathbb{P}}'^{(m)}:=\Lambda'_{\mathbb{P}:m}\circ\cdots\circ\Lambda'_{\mathbb{P}:1}.
}
Then, for $\rho^{\mathrm{tot}}_{\mathrm{ini}}:=(\bigotimes^{m}_{k=1}\rho^{A_k})\otimes\eta$,
\eq{
\left\|\Lambda_{\mathbb{P}}'^{(m)}(\rho^{\mathrm{tot}}_{\mathrm{ini}})-\left(\bigotimes^{m}_{k=1}\Lambda_{\mathbb{P}}(\rho)^{A_k}\right)\otimes\eta\right\|_1\le mf(\epsilon(\rho,\mathbb{P},\mathbb{Q},\Lambda))\label{keyb2}.
}
\textit{(Proof of \eqref{keyb2}:}
We firstly note that
\eq{
\|\Lambda_{\mathbb{P}:1}'(\rho^{\mathrm{tot}}_{\mathrm{ini}})-\Lambda_{\mathbb{P}}(\rho)^{A_1}\otimes\eta\otimes^{m}_{k=2}\rho^{A_k}\|_1\le f(\epsilon(\rho,\mathbb{P},\mathbb{Q},\Lambda)).
}
Therefore, due to the triangle inequality and the monotonicity of the trace distance,
\eq{
\|\Lambda_{\mathbb{P}:2}'\circ\Lambda_{\mathbb{P}:1}'(\rho^{\mathrm{tot}}_{\mathrm{ini}})-\otimes^{2}_{k'=1}\Lambda_{\mathbb{P}}(\rho)^{A_{k'}}\otimes\eta\otimes^{m}_{k=3}\rho^{A_k}\|_1&\le 
\|\Lambda_{\mathbb{P}:2}'\circ\Lambda_{\mathbb{P}:1}'(\rho^{\mathrm{tot}}_{\mathrm{ini}})-\Lambda_{\mathbb{P}:2}'(\Lambda_{\mathbb{P}}(\rho)^{A_{1}}\otimes\eta\otimes^{m}_{k=2}\rho^{A_k})\|_1\nonumber\\
&+
\|\Lambda_{\mathbb{P}:2}'(\Lambda_{\mathbb{P}}(\rho)^{A_{1}}\otimes\eta\otimes^{m}_{k=2}\rho^{A_k})-\otimes^{2}_{k'=1}\Lambda_{\mathbb{P}}(\rho)^{A_{k'}}\otimes\eta\otimes^{m}_{k=3}\rho^{A_k}\|_1\nonumber\\
&\le2f(\epsilon(\rho,\mathbb{P},\mathbb{Q},\Lambda)).
}
Suppose the above holds up to $l$. Namely, we assume the following inequality holds:  
\eq{
\|\Lambda_{\mathbb{P}:l}'\circ...\circ\Lambda_{\mathbb{P}:1}'(\rho^{\mathrm{tot}}_{\mathrm{ini}})-\otimes^{l}_{k'=1}\Lambda_{\mathbb{P}}(\rho)^{A_{k'}}\otimes\eta\otimes^{m}_{k=l+1}\rho^{A_k}\|_1\le lf(\delta(\rho,\mathbb{P},\mathbb{Q},\Lambda)).
}
Then, we can prove the same inequality whose $l$ is  substituted by $l+1$ as follows:
\eq{
&\|\Lambda_{\mathbb{P}:l+1}'\circ...\circ\Lambda_{\mathbb{P}:1}'(\rho^{\mathrm{tot}}_{\mathrm{ini}})-\otimes^{l+1}_{k'=1}\Lambda_{\mathbb{P}}(\rho)^{A_{k'}}\otimes\eta\otimes^{m}_{k=l+2}\rho^{A_k}\|_1\nonumber\\
&\le 
\|\Lambda_{\mathbb{P}:l+1}'\circ...\circ\Lambda_{\mathbb{P}:1}'(\rho^{\mathrm{tot}}_{\mathrm{ini}})-\Lambda_{\mathbb{P}':l+1}(\otimes^{l}_{k'=1}\Lambda_{\mathbb{P}}(\rho)^{A_{k'}}\otimes\eta\otimes^{m}_{k=l+1}\rho^{A_k})\|_1\nonumber\\
&+
\|\Lambda_{\mathbb{P}':l+1}(\otimes^{l}_{k'=1}\Lambda_{\mathbb{P}}(\rho)^{A_{k'}}\otimes\eta\otimes^{m}_{k=l+1}\rho^{A_k})
-\otimes^{l+1}_{k'=1}\Lambda_{\mathbb{P}}(\rho)^{A_{k'}}\otimes\eta\otimes^{m}_{k=l+2}\rho^{A_k}\|_1\nonumber\\
&\le (l+1)f(\epsilon(\rho,\mathbb{P},\mathbb{Q},\Lambda)).
}
Thus, by mathematical induction, we obtain \eqref{keyb2} for all $m$.
\textit{)}

Here, we evaluate $mM(\Lambda_{\mathbb{Q}})$ as 
\eq{
mM(\Lambda_{\mathbb{Q}})&\stackrel{(a)}{\ge} mM(\Lambda'_{\mathbb{P}})\nonumber\\
&\stackrel{(b)}{\ge}M(\Lambda'^{(m)}_{\mathbb{P}})\nonumber\\
&\stackrel{(c)}{\ge}M(\Lambda'^{(m)}_{\mathbb{P}}(\rho^{\mathrm{tot}}_{\mathrm{ini}}))-M(\rho^{\mathrm{tot}}_{\mathrm{ini}}),
}
where we used \eqref{keya} in (a), the property (p1) in (b) and the definition of the channel power in (c).

Therefore, when $a_A<\infty$ holds, 
\eq{
M(\rho^{\mathrm{tot}}_{\mathrm{ini}})+mM(\Lambda_{\mathbb{Q}})
&\ge M(\Lambda'^{(m)}_{\mathbb{P}}(\rho^{\mathrm{tot}}_{\mathrm{ini}}))\nonumber\\
&\ge M(\Tr_B\circ\Lambda'^{(m)}_{\mathbb{P}}(\rho^{\mathrm{tot}}_{\mathrm{ini}}))\nonumber\\
&\ge M\left(\bigotimes^{m}_{k=1}\Lambda_{\mathbb{P}}(\rho)^{A_k}\right)-m^{a_A+b_A}K_Af(\epsilon(\rho,\mathbb{P},\mathbb{Q},\Lambda))^{b_A}-c_{\max}\nonumber\\
&\ge m(M(\Lambda_{\mathbb{P}}(\rho)))-m^{a_A+b_A}K_Af(\epsilon(\rho,\mathbb{P},\mathbb{Q},\Lambda))^{b_A}-c_{\max}\label{resource_Q}.
}
Similarly, when $a_A=\infty$ holds,
\eq{
M(\rho^{\mathrm{tot}}_{\mathrm{ini}})+mM(\Lambda_{\mathbb{Q}})
&\ge  M\left(\bigotimes^{m}_{k=1}\Lambda_{\mathbb{P}}(\rho)^{A_k}\right)-m^{b_A}e^{a'_{A}m}K_Af(\epsilon(\rho,\mathbb{P},\mathbb{Q},\Lambda))^{b_A}-c_{\max}\nonumber\\
&\ge m(M(\Lambda_{\mathbb{P}}(\rho)))-m^{b_A}e^{a'_{A}m}K_Af(\epsilon(\rho,\mathbb{P},\mathbb{Q},\Lambda))^{b_A}-c_{\max}.
}

Hence, 
\eq{
M(\eta)&\ge m(M_\rho(\Lambda_{\mathbb{P}})-M(\Lambda_{\mathbb{Q}}))-m^{a_A+b_A}K_Af(\epsilon(\rho,\mathbb{P},\mathbb{Q},\Lambda))^{b_A}-c_{\max},\enskip(\mbox{for }a_A<\infty)\\
M(\eta)&\ge m(M_\rho(\Lambda_{\mathbb{P}})-M(\Lambda_{\mathbb{Q}}))-m^{b_A}e^{a'_{A}m}K_Af(\epsilon(\rho,\mathbb{P},\mathbb{Q},\Lambda))^{b_A}-c_{\max},\enskip(\mbox{for }a_A=\infty).
}
Since the above inequalities become trivial when $M_\rho(\Lambda_{\mathbb{P}})-M(\Lambda_{\mathbb{Q}})<0$, we can rewrite them as follows without loss of generality:
\eq{
M(\eta)&\ge m|M_\rho(\Lambda_{\mathbb{P}})-M(\Lambda_{\mathbb{Q}})|_+-m^{a_A+b_A}K_Af(\epsilon(\rho,\mathbb{P},\mathbb{Q},\Lambda))^{b_A}-c_{\max},\enskip(\mbox{for }a_A<\infty)\label{eq:bound_any_m}\\
M(\eta)&\ge m|M_\rho(\Lambda_{\mathbb{P}})-M(\Lambda_{\mathbb{Q}})|_+-m^{b_A}e^{a'_{A}m}K_Af(\epsilon(\rho,\mathbb{P},\mathbb{Q},\Lambda))^{b_A}-c_{\max},\enskip(\mbox{for }a_A=\infty)\label{eq:bound_any_m_ainf}.
}
These inequalities hold for any natural number $m$. Therefore, we only have to find a proper $m$ to provide \eqref{eq:key_resource_SM_prime} and \eqref{eq:key_resource_ainf_SM_prime}, respectively.

We first consider the case of $a_A<\infty$.
In this case, \eqref{eq:bound_any_m} is written as 
\eq{
M(\eta)\ge m\alpha-m^{a_A+b_A}\beta-c_{\max}\label{eq:simple1},
}
where $\alpha:=|M_\rho(\Lambda_{\mathbb{P}})-M(\Lambda_{\mathbb{Q}})|_+$ and $\beta:=K_Af(\epsilon(\rho,\mathbb{P},\mathbb{Q},\Lambda))^{b_A}$.
Then, for $g_1(x):=x\alpha-x^{a_A+b_A}\beta-c_{\max}$, 
\eq{
\frac{dg_1(x)}{dx}=\alpha-(a_A+b_A)\beta x^{(a_A+b_A)-1}.
}
Therefore, if $m$ can take a real value, the right-hand side of \eqref{eq:simple1} takes a peak at $m=\left(\frac{\alpha}{(a_A+b_A)\beta}\right)^{\frac{1}{(a_A+b_A)-1}}$.

Therefore, we define a natural number $m_{*}:=\left(\frac{\alpha}{(a_A+b_A)\beta}\right)^{\frac{1}{(a_A+b_A)-1}}-s$ where $0\le s<1$, and substitute it for \eqref{eq:bound_any_m}.
Then, we obtain 
\eq{
M(\eta)&\ge
\left(\left(\frac{\alpha}{(a_A+b_A)\beta}\right)^{\frac{1}{(a_A+b_A)-1}}-s\right)\alpha-\left(\left(\frac{\alpha}{(a_A+b_A)\beta}\right)^{\frac{1}{(a_A+b_A)-1}}-s\right)^{a_A+b_A}\beta-c_{\max}\nonumber\\
&\ge
\left(\left(\frac{\alpha}{(a_A+b_A)\beta}\right)^{\frac{1}{(a_A+b_A)-1}}\right)\alpha-s\alpha-\left(\frac{\alpha}{(a_A+b_A)\beta}\right)^{\frac{a_A+b_A}{(a_A+b_A)-1}}\beta-c_{\max}\nonumber\\
&=
\left(\left(\frac{\alpha}{a_A+b_A}\right)^{\frac{1}{(a_A+b_A)-1}}\alpha-\left(\frac{\alpha}{a_A+b_A}\right)^{\frac{a_A+b_A}{(a_A+b_A)-1}}\right)\beta^{-\frac{1}{a_A+b_A-1}}-s\alpha-c_{\max}\nonumber\\
&=\left(\left(\frac{1}{a_A+b_A}\right)^{\frac{1}{(a_A+b_A)-1}}-\left(\frac{1}{a_A+b_A}\right)^{\frac{a_A+b_A}{(a_A+b_A)-1}}\right)\alpha^{\frac{a_A+b_A}{a_A+b_A-1}}\beta^{-\frac{1}{a_A+b_A-1}}-s\alpha-c_{\max}\nonumber\\
&=\left(\frac{1}{a_A+b_A}\right)^{\frac{1}{a_A+b_A-1}}\left(1-\frac{1}{a_A+b_A}\right)\alpha^{\frac{a_A+b_A}{a_A+b_A-1}}\beta^{-\frac{1}{a_A+b_A-1}}-s\alpha-c_{\max}\nonumber\\
&=g(a_A,b_A)\alpha^{\frac{a_A+b_A}{a_A+b_A-1}}\beta^{-\frac{1}{a_A+b_A-1}}-s\alpha-c_{\max}\nonumber\\
&\ge g(a_A,b_A)\alpha^{\frac{a_A+b_A}{a_A+b_A-1}}\beta^{-\frac{1}{a_A+b_A-1}}-\alpha-c_{\max}.
}
Subsituting $\alpha=|M_\rho(\Lambda_{\mathbb{P}})-M(\Lambda_{\mathbb{Q}})|_+$ and $\beta=K_Af(\epsilon(\rho,\mathbb{P},\mathbb{Q},\Lambda))^{b_A}$, we obtain \eqref{eq:key_resource_SM_prime}.

Next, let us consider the case of $a_A=\infty$.
In this case, we can rewrite \eqref{eq:bound_any_m_ainf} as
\eq{
M(\eta)\ge m\alpha-m^b e^{a'm}\beta-c_{\max}.
}
where $\alpha:=|M_\rho(\Lambda_{\mathbb{P}})-M(\Lambda_{\mathbb{Q}})|_+$ and $\beta:=K_Af(\epsilon(\rho,\mathbb{P},\mathbb{Q},\Lambda))^{b_A}$.
Note that $m^b\le e^{bm}$, and thus
\eq{
M(\eta)\ge m\alpha-e^{(a'_A+b_A)m}\beta-c_{\max}.
\label{eq:simple2}
}
Then, for $g_2(x):=x\alpha-e^{(a'_A+b_A)x}\beta-c_{\max}$, 
\eq{
\frac{dg_2(x)}{dx}=\alpha-(a'_A+b_A)\beta e^{(a'_A+b_A)x}.
}
Therefore, if $m$ can take a real value, the right-hand side of \eqref{eq:simple2} takes a peak at $m=\frac{1}{a'_A+b_A}\log\frac{\alpha}{\beta(a'_A+b_A)}$.

Therefore, we define a natural number $m_{*}:=\frac{1}{a'_A+b_A}\log\frac{\alpha}{\beta(a'_A+b_A)}-s_2$ where $0\le s_2<1$, and substitute it for \eqref{eq:bound_any_m_ainf}.
Then, we obtain 
\eq{
M(\eta)&\ge
\left(\frac{1}{a'_A+b_A}\log\frac{\alpha}{\beta(a'_A+b_A)}-s_2\right)\alpha-e^{(a'_A+b_A)\frac{1}{a'_A+b_A}\log\frac{\alpha}{\beta(a'_A+b_A)}-s_2}\beta-c_{\max}\nonumber\\
&\ge
\frac{\alpha}{a'_A+b_A}\log\frac{\alpha}{\beta(a'_A+b_A)}-s_2\alpha
-e^{\log\frac{\alpha}{\beta(a'_A+b_A)}}\beta-c_{\max}\nonumber\\
&\ge
\frac{\alpha}{a'_A+b_A}\log\frac{\alpha}{\beta(a'_A+b_A)}-s_2\alpha
-\frac{\alpha}{a'_A+b_A}-c_{\max}\nonumber\\
&\ge\frac{\alpha}{a'_A+b_A}\log\frac{\alpha}{\beta(a'_A+b_A)}
-\alpha\left(1+\frac{1}{a'_A+b_A}\right)-c_{\max}.
}
Subsituting $\alpha=|M_\rho(\Lambda_{\mathbb{P}})-M(\Lambda_{\mathbb{Q}})|_+$ and $\beta=K_Af(\epsilon(\rho,\mathbb{P},\mathbb{Q},\Lambda))^{b_A}$, we obtain \eqref{eq:key_resource_ainf_SM_prime}.
\end{proofof}

From Lemma \ref{Lemm:key_resource}, we can derive Theorems \ref{Thm:RItradeoff_SM} and \ref{Thm:RPtradeoff_SM} and Corollary \ref{Coro:REtradeoff_conversion_SM}. We prove them one by one.

\begin{proofof}{Theorem \ref{Thm:RPtradeoff_SM}}
Let $\mathbb{P}_*:=\{P_k\}_{k=1,2}$ and $\rho_*$ be a two-valued PVM on $A$ that discriminates $\rho_1$ and $\rho_2$ in $\Omega_p$ with probability 1 and a state on $A$ satisfying
\eq{
M_{\rho_*}(\Lambda_{\mathbb{P}_*})=\Mem(\Omega_p|\mathscr{K}).\label{eq:deri1_thm2}
}
Then, by definition, for any two-valued POVM $\mathbb{Q}:=\{Q_k\}_{k=1,2}$ on $A'$, 
\eq{
\PF(\Omega'_{p,\Lambda},\mathbb{Q})=\epsilon(\rho_*,\mathbb{P}_*,\mathbb{Q},\Lambda)^2.\label{eq:deri2_them2}
}
Substituting \eqref{eq:deri1_thm2} and \eqref{eq:deri2_them2} in \eqref{eq:key_resource_SM} and \eqref{eq:key_resource_ainf_SM}, we obtain \eqref{eq:RPtradeoff_SM}
 and \eqref{eq:RPtradeoff_ainf_SM}
\end{proofof}

\begin{proofof}{Theorem \ref{Thm:RItradeoff_SM}}
Let $\mathbb{Q}_*:=\{Q_{k,*}\}_{k=1,2}$ be a POVM on $A'$ satisfying 
\eq{
M(\Lambda_{\mathbb{Q}_*})&=\Mcos(\Omega'_{p,\Lambda}|\mathscr{K}),\label{eq:deri1_ThmS1}\\
\delta(\Lambda,\Omega_p)^2&=\PF(\Omega'_{p,\Lambda},\mathbb{Q}_*)\label{eq:deri2_ThmS1},
}
where $\Lambda_{\mathbb{Q}_*}(\cdot)=\sum_k\sqrt{Q_{k,*}}\cdot\sqrt{Q_{k,*}}\otimes\ket{k}\bra{k}_K$.
Substituting \eqref{eq:deri1_ThmS1} and \eqref{eq:deri2_ThmS1} in \eqref{eq:RPtradeoff_SM}
 and \eqref{eq:RPtradeoff_ainf_SM}, we obtain \eqref{eq:RItradeoff_SM}
 and \eqref{eq:RItradeoff_ainf_SM}
\end{proofof}

\begin{proofof}{Corollary \ref{Coro:REtradeoff_conversion_SM}}
We firstly show that when \eqref{cond_approx_SM} holds, the following inequality is valid for any two-valued POVM $\mathbb{Q}:=\{Q_k\}_{k=1,2}$ discriminating $\rho'_1$ and $\rho'_2$ with probability 1:
\eq{
\PF(\Omega'_{p,\Lambda},\mathbb{Q})\le\epsilon^2.\label{target_coro1}
}
We define $p_1:=p$ and $p_2:=1-p$, and define two states
\eq{
\rho^{\mathrm{tot}}_{\Omega,p}&:=\sum_ip_i\rho_i\otimes\ket{i}\bra{i}_{K_{0}},\\
\rho^{\mathrm{tot}}_{\Omega',p}&:=\sum_ip_i\rho'_i\otimes\ket{i}\bra{i}_{K_{0}}.
}
Then, we obtain
\eq{
D_F(\Lambda(\rho^{\mathrm{tot}}_{\Omega,p}),\rho^{\mathrm{tot}}_{\Omega',p})^2&\stackrel{(a)}{=}1-(\sum_{i}p_iF(\Lambda(\rho_i),\rho'_i))^2\nonumber\\
&\le1-(\sum_ip_i\sqrt{1-\epsilon^2})^2\nonumber\\
&=\epsilon^2.
}
Here we used the following relation \cite{TTK} in (a):
\eq{
D_F(\sum_ip_i\rho_i\otimes\ket{i}\bra{i},\sum_ip_i\sigma\otimes\ket{i}\bra{i})^2
=1-(\sum_ip_iF(\rho_i,\sigma_i))^2.
}

Since $\mathbb{Q}$ discriminates $\rho'_1$ and $\rho'_2$ with probability 1, by using the index $k_{\neg}\ne k$ (i.e. when $k=1$, $k_{\neg}=2$ and when $k=2$, $k_{\neg}=1$), we get
\eq{
D_F(\Lambda(\rho^{\mathrm{tot}}_{\Omega,p}),\rho^{\mathrm{tot}}_{\Omega',p})^2&\ge 
D_F(\Lambda_{\mathbb{Q}}\circ\Lambda(\rho^{\mathrm{tot}}_{\Omega,p}),\Lambda_{\mathbb{Q}}(\rho^{\mathrm{tot}}_{\Omega',p}))^2\nonumber\\
&\ge D_F(\Tr_{A'}\circ\Lambda_{\mathbb{Q}}\circ\Lambda(\rho^{\mathrm{tot}}_{\Omega,p}),\Tr_{A'}\circ\Lambda_{\mathbb{Q}}(\rho^{\mathrm{tot}}_{\Omega',p}))^2\nonumber\\
&=  D_F(\sum_kp_k(1-\Tr[Q_{k_{\neg}}(\Lambda(\rho_k))])\ket{k}\bra{k}_K\otimes\ket{k}\bra{k}_{K_0}+\sum_kp_k\Tr[Q_{k_{\neg}}(\Lambda(\rho_k))]\ket{k_{\neg}}\bra{k_{\neg}}_K\otimes\ket{k}\bra{k}_{K_0},\nonumber\\
&\enskip\sum_kp_k\ket{k}\bra{k}_K\otimes\ket{k}\bra{k}_{K_0})^2\nonumber\\
&\ge \sum_kp_kD_F((1-\Tr[Q_{k_{\neg}}\Lambda(\rho_k)])\ket{k}\bra{k}_{K}+\Tr[Q_{k_{\neg}}\Lambda(\rho_k)]\ket{k_{\neg}}\bra{k_{\neg}}_K,\ket{k}\bra{k}_K)^2\nonumber\\
&=\sum_kp_k\Tr[(1-Q_k)\Lambda(\rho_k)]\nonumber\\
&=\PF(\Omega'_{p,\Lambda},\mathbb{Q}).
}
Therefore, we obtain \eqref{target_coro1}.

Now, let us derive \eqref{eq:REtradeoff_conversion_SM} and \eqref{eq:REtradeoff_conversion_SM_ainf}.
Due to \eqref{target_coro1}, \eqref{eq:RPtradeoff_SM}
and \eqref{eq:RPtradeoff_ainf_SM},
the following inequalities hold for any $0<p<1$ and any $\mathbb{Q}$ that discriminates $\rho''_1$ and $\rho''_2$ with probability 1:
\eq{
M_c(\Lambda)&\ge g(a_A,b_A)\frac{|\Mem(\Omega_p|\mathscr{K})-M(\Lambda_{\mathbb{Q}})|^{\frac{a_A+b_A}{a_A+b_A-1}}_+}{(K_Af(\epsilon)^{b_A})^{\frac{1}{a_A+b_A-1}}}-|\Mem(\Omega_p|\mathscr{K})-M(\Lambda_{\mathbb{Q}})|_+-c_{\max},\enskip(\mbox{for }a_A<\infty)\label{S103}\\
M_c(\Lambda)\ge&\frac{|\Mem(\Omega_p|\mathscr{K})-M(\Lambda_{\mathbb{Q}})|_+}{a'_A+b_A}\log\frac{|\Mem(\Omega_p|\mathscr{K})-M(\Lambda_{\mathbb{Q}})|_+}{K_Af(\epsilon)^{b_A}(a'_A+b_A)}\nonumber\\
&-|\Mem(\Omega_p|\mathscr{K})-M(\Lambda_{\mathbb{Q}})|_+\left(1+\frac{1}{a'_A+b_A}\right)
-c_{\max},\enskip(\mbox{for }a_A=\infty)\label{S104}.
}
Because $\rho''_1$ and $\rho''_2$ are orthogonal, any  $\mathbb{Q}$ that discriminates $\rho''_1$ and $\rho''_2$ with probability 1 satisfies
\eq{
\PF(\Omega''_p,\mathbb{Q})=\min_{\mathbb{E}}\PF(\Omega''_p,\mathbb{E})=0.
}
Therefore, when we take $\mathbb{Q}_*=\{Q_{*,k}\}_{k=1,2}$ on $A'$ as 
\eq{
M(\Lambda_{\mathbb{Q}_*})=\min_{\mathbb{Q}}\{M(\Lambda_{\mathbb{Q}})|\PF(\Omega''_{p},\mathbb{Q})=0\},
}
it satisfies
\eq{
M(\Lambda_{\mathbb{Q}_*})&=\Mcos(\Omega''_p|\mathscr{K})
}
for any $0< p<1$. (In other words, $\Mcos(\Omega''_p|\mathscr{K})$ has the same value for any $0<p<1$).
Therefore, substituting $\mathbb{Q}_*$ for $\mathbb{Q}$ in \eqref{S103} and \eqref{S104} and take the maximization runs over $0<p<1$, we obtain \eqref{eq:REtradeoff_conversion_SM} and \eqref{eq:REtradeoff_conversion_SM_ainf}.
\end{proofof}

\subsubsection{The case of $M$ satisfying the conditions (i), (ii-strong), and (iii)}

Next, we treat the case where the resource measure $M$ satisfies the conditions (i), (ii-strong), and (iii).
Since the condition (ii-strong) clearly implies (ii), the results in the previous subsection (e.g. Theorems \ref{Thm:RItradeoff_SM} and \ref{Thm:RPtradeoff_SM}) also hold in this case.
In the case of (ii-strong), the following quantity becomes important:
\eq{
M_{\rho}(\Lambda:\Lambda_{\mathbb{P}}):=
|M_\rho(\Lambda)+M_\rho(\Lambda_{\mathbb{P}})-M_{\rho}(\id_K\otimes\Lambda\circ\Lambda_{\mathbb{P}})|_+,
}
where $\Lambda_{\mathbb{P}}(\cdot)=\sum_kP_k\cdot P_k\otimes\dm{k}_K$ with the register $\mathscr{K}:=(K,\{\ket{k}\})$, and where $|x|_+$ returns $x$ when $x$ is positive, and returns $0$ otherwise. 
Using this quantity, we can derive an alternative version of Theorem \ref{Thm:RItradeoff_SM}:
\begin{theorem}[Resource-Irreversibility tradeoff for (ii-strong)]\label{Thm:RItradeoff_alter_SM}
Let $A$ and $A'$ be quantum systems.
We take an arbitrary channel $\Lambda$ from $A$ to $A'$, and its implementation $(V,\eta^B)$ such that $V$ is a free unitary and $\Lambda(\cdot)=\Tr_{B'}[V(\cdot\otimes\eta)V^\dagger]$ holds.
Then, for any test ensemble $\Omega_p:=(\rho_1,\rho_2)_p$ whose $\rho_1$ and $\rho_2$ are mutually orthogonal, any two-level register $\mathscr{K}:=(K,\{\ket{k}\}_{k=1,2})$, any resource measure $M$ satisfying the conditions (i), (ii-strong) for $\{A,A',B,B',K\}$ and (iii) for $a_S=b_S=1$ and $c=0$, any two-valued PVM $\mathbb{P}:=\{P_k\}_{k=1,2}$ discriminating $\rho_1$ and $\rho_2$ with probability 1, and any state $\rho$ satisfying $\Lambda_{\mathbb{P}}(\rho)=\sum_kp_k\rho_k\otimes\ket{k}\bra{k}_K$ for $\Lambda_{\mathbb{P}}(\cdot):=\sum_kP_k\cdot P_k\otimes\ket{k}\bra{k}_K$, the following inequality holds:
\eq{
M(\eta)&\ge \frac{M_{\rho}(\Lambda:\Lambda_{\mathbb{P}})^2}{4(K_A+K_{A'})f(\delta(\Lambda,\Omega_p))}-M_{\rho}(\Lambda:\Lambda_{\mathbb{P}})\label{eq:RItradeoff_alter_SM}.
}
\end{theorem}

Since (ii-strong) implies (ii), Theorem \ref{Thm:RItradeoff_SM} and \ref{Thm:RPtradeoff_SM} also hold when (ii-strong) is satisfied.
The advantage of Theorem \ref{Thm:RItradeoff_alter_SM} is that the numerator on the right-hand side of the trade-off relation, $\Mem(\Omega_{p}|\mathscr{K})-\Mcos(\Omega'_{p,\Lambda}|\mathscr{K})$, is replaced by $M{\rho}(\Lambda:\Lambda_{\mathbb{P}})$ that can be computed easily for any $\Lambda$ and $\Lambda_{\mathbb{P}}$.
As a result, Theorem \ref{Thm:RItradeoff_alter_SM} enables straightforward verification of the inverse proportionality between the implementation error and the resource cost for a broader class of quantum channels.
At present, the only known resource measure satisfying (ii-strong) is the expectation value of energy; however, the theorem remains applicable to any other measure satisfying this condition.

\begin{proofof}{Theorem \ref{Thm:RItradeoff_alter_SM}}
Due to \eqref{PF_equal_delta}, it is enough to show the following inequalities for any two-valued POVM $\mathbb{Q}:=\{Q_k\}_{k=1,2}$:
\eq{
M(\eta)&\ge \frac{M_{\rho}(\Lambda:\Lambda_{\mathbb{P}})^{2}}{4(K_A+K_{A'})f(\sqrt{\PF(\Omega'_{p,\Lambda},\mathbb{Q})})}-M_{\rho}(\Lambda:\Lambda_{\mathbb{P}}).\label{pre_ene}
}
Let us define
$\Lambda'_{\mathbb{P}}:AB\rightarrow ABK$ as
\eq{
\Lambda'_{\mathbb{P}}(...):=(\calV^\dagger\otimes\mathrm{id}^{K})\circ(\Lambda_{\mathbb{Q}}\otimes\mathrm{id}^{B'})\circ\calV(...).
}
Here $\calV(\cdot)=V\cdot V^\dagger$,  $\calV^\dagger(\cdot):=V^\dagger\cdot V$ and $\Lambda_{\mathbb{Q}}(\cdot):=\sum_{k}\sqrt{Q_k}\cdot\sqrt{Q_k}\otimes\dm{k}_{K}$.
Due to Lemma \ref{Lemm:key_nonresource}, we obtain
\eq{
\|\Lambda'_{\mathbb{P}}(\rho\otimes\eta)-\Lambda_{\mathbb{P}}(\rho)\otimes\eta\|_1\le f(\epsilon(\rho,\mathbb{P},\mathbb{Q},\Lambda)).
} 
Since $\mathbb{P}$ discriminates $\rho_1$ and $\rho_2$ with probability 1, and since $\Tr[P_k\rho]=p_k$, we obtain 
\eq{
\epsilon(\rho,\mathbb{P},\mathbb{Q},\Lambda)=\sqrt{\PF(\Omega'_{p,\Lambda},\mathbb{Q})}.
}
Therefore, we obtain
\eq{
\|\Lambda'_{\mathbb{P}}(\rho\otimes\eta)-\Lambda_{\mathbb{P}}(\rho)\otimes\eta\|_1\le f(\sqrt{\PF(\Omega'_{p,\Lambda},\mathbb{Q})}).\label{keyb1_en}
} 
Due to \eqref{keyb1_en}, the invariance of the trace distance under unitary transformations, and its monotonicity under partial trace, we obtain
\eq{
\|\Lambda_{\mathbb{Q}}\circ\Lambda(\rho)-\id^K\otimes\Lambda\circ\Lambda_{\mathbb{P}}(\rho)\|_1\le f(\sqrt{\PF(\Omega'_{p,\Lambda},\mathbb{Q})}).\label{keyb1prime_en}
}

Here, let us prepare $m$ copies of $A$, and name them $A_1$,...,$A_m$. We also refer to a channel that performs $\Lambda'_{\mathbb{P}}$ on $A_k$ and $B$ as $\Lambda'_{\mathbb{P}:k}$.
We define 
\eq{
\Lambda_{\mathbb{P}}'^{(m)}:=\Lambda'_{\mathbb{P}:m}\circ\cdots\circ\Lambda'_{\mathbb{P}:1}.
}
Then, for $\rho^{\mathrm{tot}}_{\mathrm{ini}}:=(\bigotimes^{m}_{k=1}\rho^{A_k})\otimes\eta$,
\eq{
\left\|\Lambda_{\mathbb{P}}'^{(m)}(\rho^{\mathrm{tot}}_{\mathrm{ini}})-\left(\bigotimes^{m}_{k=1}\Lambda_{\mathbb{P}}(\rho)^{A_k}\right)\otimes\eta\right\|_1\le mf(\sqrt{\PF(\Omega'_{p,\Lambda},\mathbb{Q})})\label{keyb2_en}.
}
Therefore,
\eq{
M(\Lambda'^{(m)}_{\mathbb{P}}(\rho^{\mathrm{tot}}_{\mathrm{ini}}))
&\ge
M(\Tr_B[\Lambda'^{(m)}_{\mathbb{P}}(\rho^{\mathrm{tot}}_{\mathrm{ini}})])
\nonumber\\
&\ge M\left(\bigotimes^{m}_{k=1}\Lambda_{\mathbb{P}}(\rho)^{A_k}\right)-mK_Amf(\sqrt{\PF(\Omega'_{p,\Lambda},\mathbb{Q})})\nonumber\\
&=
m(M(\Lambda_{\mathbb{P}}(\rho)))-m^{2}K_Af(\sqrt{\PF(\Omega'_{p,\Lambda},\mathbb{Q})}).\label{eq:lambda_P_rho_tot_lb}
}
Here, we define   $\rho^{\mathrm{tot}}_{l}:=\Lambda'_{\mathbb{P}:l}...\Lambda'_{\mathbb{P}:1}(\rho^{\mathrm{tot}}_{\mathrm{ini}})$ and $\rho^{\mathrm{tot}}_{0}:=\rho^{\mathrm{tot}}_{\mathrm{ini}}$. Then, 
\eq{
\rho^{\mathrm{tot}}_{l+1}=(\calV^{A_{l+1}B})^\dagger\circ\Lambda^{A'_{l+1}}_{\mathbb{Q}}\circ\calV^{A_{l+1}B}(\rho^{\mathrm{tot}}_l).
}
Due to the condition (i), $\calV$ and $\calV^\dagger$ preserves the amount of resource. Therefore,
\eq{
M(\rho^{\mathrm{tot}}_{l+1})=M(\rho^{\mathrm{tot}}_{l})+(M(\Lambda^{A'_{l+1}}_{\mathbb{Q}}\circ\calV^{A_{l+1}B}(\rho^{\mathrm{tot}}_l))-M(\calV^{A_{l+1}B}(\rho^{\mathrm{tot}}_l))).
}
Also, since $M$ satisfies (ii-strong) for $\{A,A',B,B',K\}$, we obtain
\eq{
M(\Lambda^{A'_{l+1}}_{\mathbb{Q}}\circ\calV^{A_{l+1}B}(\rho^{\mathrm{tot}}_l))-M(\calV^{A_{l+1}B}(\rho^{\mathrm{tot}}_l))=M(\Lambda^{A'_{l+1}}_{\mathbb{Q}}(\Tr_{\neg A'_{l+1}}[\calV^{A_{l+1}B}(\rho^{\mathrm{tot}}_l)]))-M(\Tr_{\neg A'_{l+1}}[\calV^{A_{l+1}B}(\rho^{\mathrm{tot}}_l)]),
}
where $\Tr_{\neg A'_l}:=\Tr_{A_1...A_{l-1}A_{l+1}...A_mB'K_1...K_{l}}$, and $K_1$,...,$K_l$ are copies of the register system $K$.
Since
\eq{
\left\|\rho^{\mathrm{tot}}_l-\left(\bigotimes^{l}_{k=1}\Lambda_{\mathbb{P}}(\rho)^{A_k}\right)\otimes\left(\bigotimes^{m}_{k'=l+1}\rho^{A_{k'}}\right)\otimes\eta\right\|_1\le lf(\sqrt{\PF(\Omega'_{p.\Lambda},\mathbb{Q})}),
}
we get
\eq{
\|\Tr_{\neg A'_{l+1}}[\calV^{A_{l+1}B}(\rho^{\mathrm{tot}}_l)]-\Lambda(\rho)\|_1\le lf(\sqrt{\PF(\Omega'_{p.\Lambda},\mathbb{Q})}),
}
which implies
\eq{
M(\Lambda^{A'_{l+1}}_{\mathbb{Q}}(\Tr_{\neg A'_{l+1}}[\calV^{A_{l+1}B}(\rho^{\mathrm{tot}}_l)]))-M(\Tr_{\neg A'_{l+1}}[\calV^{A_{l+1}B}(\rho^{\mathrm{tot}}_l)])
&\le
M(\Lambda_{\mathbb{Q}}\circ\Lambda(\rho))-M(\Lambda(\rho))
+2K_{A'}lf(\sqrt{\PF(\Omega'_{p,\Lambda},\mathbb{Q})})\nonumber\\
&=M_{\Lambda(\rho)}(\Lambda_{\mathbb{Q}})+2K_{A'}lf(\sqrt{\PF(\Omega'_{p,\Lambda},\mathbb{Q})}).
}
By repeatedly applying this inequality from $l=1$ to $m-1$, we obtain
\eq{
M(\Lambda'^{(m)}_{\mathbb{P}}(\rho^{\mathrm{tot}}_{\mathrm{ini}}))&\le M(\rho^{\mathrm{tot}}_{\mathrm{ini}})+mM_{\Lambda(\rho)}(\Lambda_{\mathbb{Q}})+(\sum^{m-1}_{l=1}l)2K_{A'}f(\sqrt{\PF(\Omega'_{p,\Lambda},\mathbb{Q})})\nonumber\\
&=
M(\rho^{\mathrm{tot}}_{\mathrm{ini}})+mM_{\Lambda(\rho)}(\Lambda_{\mathbb{Q}})+m(m-1)K_{A'}f(\sqrt{\PF(\Omega'_{p,\Lambda},\mathbb{Q})}).
}
Combining this inequality with \eqref{eq:lambda_P_rho_tot_lb}, we get
\eq{
m(M(\Lambda_{\mathbb{P}}(\rho)))-m^{2}K_Af(\sqrt{\PF(\Omega'_{\Lambda,p},\mathbb{Q})})\le M(\rho^{\mathrm{tot}}_{\mathrm{ini}})+mM_{\Lambda(\rho)}(\Lambda_{\mathbb{Q}})+m(m-1)K_{A'}f(\sqrt{\PF(\Omega'_{p,\Lambda},\mathbb{Q})})
}
Furthermore, since \eqref{keyb2_en} implies
\eq{
M_{\Lambda(\rho)}(\Lambda_{\mathbb{Q}})&\le M(\id^{K}\otimes\Lambda\circ\Lambda_{\mathbb{P}}(\rho))-M(\Lambda(\rho))+K_{A'}f(\sqrt{\PF(\Omega'_{p,\Lambda},\mathbb{Q})})\nonumber\\
&=M_{\rho}(\id^{K}\otimes\Lambda\circ\Lambda_{\mathbb{P}})-M_{\rho}(\Lambda)+K_{A'}f(\sqrt{\PF(\Omega'_{p,\Lambda},\mathbb{Q})}),
}
we have
\eq{
m(M(\Lambda_{\mathbb{P}}(\rho)))-m^2K_Af(\sqrt{\PF(\Omega'_{p,\Lambda},\mathbb{Q})})\le M(\rho^{\mathrm{tot}}_{\mathrm{ini}})+m(M_{\rho}(\id^{K}\otimes\Lambda\circ\Lambda_{\mathbb{P}})-M_{\rho}(\Lambda))+m^2K_{A'}f(\sqrt{\PF(\Omega'_{p,\Lambda},\mathbb{Q})}).
}
Therefore, we obtain
\eq{
M(\eta)\ge mM_\rho(\Lambda:\Lambda_{\mathbb{P}})-m^2(K_A+K_{A'})f(\sqrt{\PF(\Omega'_{p,\Lambda},\mathbb{Q})}).
}
By applying the same argument as in the proof of Lemma \ref{Lemm:key_resource} to the case $a_S=b_S=1$ and $c=0$, we complete the proof of \eqref{pre_ene}.
\end{proofof}


\section{List of applicable resource theories}\label{app:application_list}
As discussed in the main text, our framework applies to any resource theory that admits even a single resource measure $M$ satisfying Conditions (i)--(iii); consequently, the theories listed below are by no means exhaustive. 
Note that Conditions (ii) additivity and (iii) H\"{o}lder continuity depend only on the intrinsic properties of the measure $M$, allowing flexibility in choosing free operations, as long as $M$ satisfies monotonicity, i.e., Condition~(i). 

\subsection{Resource theory of energy}

Consider a system composed of subsystems $\{S_i\}_{i=1}^N$ each associated with a local Hamiltonian $H_i$ which acts nontrivially only on the corresponding subsystem $A_i$ and is assumed to be positive semi-definite, i.e., $H_i \geq 0$. Without loss of generality, we set the ground-state energy of the local Hamiltonian to zero by subtracting a constant. The total Hamiltonian is given by $H_{\mathrm{tot}} \coloneqq \sum_{i=1}^N H_i$. In the resource theory of energy, the free states are defined as the ground states of the local Hamiltonians, while the free operations are defined as those that can be implemented by a combination of the following procedures:
\begin{enumerate}
    \item Appending an ancillary system in the ground state.
    \item Performing a global unitary operation $U$ that preserves the total energy, i.e., satisfying $[U, H_{\mathrm{tot}}]=0$.
    \item Tracing out an arbitrary subset of the subsystems.
\end{enumerate}

The expectation value of the Hamiltonian, defined by $E(\rho_S)\coloneqq \Tr(\rho_SH_S)$, satisfies Conditions~(i), (ii-strong), and (iii) in this resource theory, as detailed below:
\begin{description}
    \item[(i)] For a free unitary operator $U$, its adjoint $U^\dag$ is also free since $[U,H_{\mathrm{tot}}]=0$ implies $[U^\dag,H_{\mathrm{tot}} ]=0$. Because both $\mathbb{I}\otimes U$ and $\mathbb{I}\otimes U^\dag$ are energy-conserving, the quantity $E$ is invariant under maps $\id\otimes \calE$ and $\id\otimes \calE^\dag$, where $\calE(\cdot)\coloneqq U(\cdot)U^\dag$. When a subset of the subsystems is traced out, $E$ is non-increasing due to the positive semi-definiteness of the local Hamiltonian. Consequently, $E$ is non-increasing under $\id\otimes\calE$ when $\calE$ is a partial trace operation. 
    \item[(ii-strong)] The additivity of the Hamiltonian implies $E(\rho_{AB})=E(\Tr_B(\rho_{AB}))+E(\Tr_A(\rho_{AB}))$. Therefore, $E$ satisfies (ii-strong) for any set of systems, provided there is no interaction Hamiltonian among these systems.
    \item[(iii)] For any two states $\rho_S$ and $\rho'_S$, we have 
        \begin{align}
            |E(\rho_S)-E(\rho'_S)|&=|\mathrm{Tr}(H_S(\rho_S-\rho'_S))|=|\mathrm{Tr}((H_S-\|H_S\|_\infty/2)(\rho_S-\rho'_S))|\\
            &\leq \|H_S-\|H_S\|_\infty/2\|_\infty \|\rho_S-\rho'_S\|_1= \frac{\Delta_{H_S}}{2} \|\rho_S-\rho'_S\|_1,\label{eq:continuity_energy_exp}
        \end{align}
        where $\|\cdot\|_\infty$ denotes the operator norm. Therefore, for $m$ copies of $S$ system, denoted by $S_1\cdots S_m$, it follows that
        \begin{align}
            |E(\rho^{\otimes m})-E(\sigma_m)|&\leq \frac{\Delta_{H_{S_1}+\cdots+H_{S_m}}}{2} \|\rho^{\otimes m}-\sigma_m\|_1=mK_S\|\rho^{\otimes m}-\sigma_m\|_1,\label{eq:Hoelder_continuity_energy_exp}
        \end{align}
        where we have defined $K_S\coloneqq\frac{\Delta_{H_S}}{2} $. Therefore, Condition~(iii) is satisfied with $a_S=b_S=1$, $c=0$, and $K_S=\|H_S\|_\infty$.
\end{description}

\subsection{Resource theory of asymmetry} 

The resource theory of asymmetry offers an alternative framework for addressing conservation laws. In particular, under time-translation symmetry generated by the unitary $e^{-iHt}$, free states are defined as those invariant under time evolution, while free operations are those that commute with the time-translation operation. Within this framework, quantum fluctuations with respect to the Hamiltonian $H$ serve as a resource.

The quantum Fisher information (QFI) for $\rho_t\coloneqq e^{-iHt}\rho e^{iHt}$ is a commonly used resource monotone in this context. By using the eigenvalue decomposition $\rho=\sum_{i}\lambda_i\ket{i}\bra{i}$, this quantity can be expressed as $\mathcal{F}_H(\rho)=\sum_{i,j}\frac{(\lambda_i-\lambda_j)^2}{\lambda_i+\lambda_j}|\braket{i|H|j}|^2$, where the summation is taken over pairs $(i,j)$ such that $\lambda_i+\lambda_j>0$. 
The QFI satisfies Conditions~(i), (ii), and (iii), as shown below:
\begin{description}
    \item[(i)] In the resource theory of asymmetry, a unitary operation $U$ is free if and only if it commutes with the Hamiltonian. In this case, both $\mathbb{I}\otimes U$ and $\mathbb{I}\otimes U^\dag$  also commute with the total Hamiltonian, implying that $\id\otimes \calE$ and $\id\otimes\calE^\dag$ are free operations, where $\calE(\cdot)\coloneqq U(\cdot)U^\dag$. It is also straightforward to verify that $\id\otimes \calE$ is free when $\calE$ is a partial trace operation. Since the QFI is non-increasing under free operations in the resource theory of asymmetry~\cite{yadin_general_2016}, Condition~(i) is satisfied. 
    \item[(ii)] The additivity is proven in~\cite{hansen_metric_2008}.
    \item[(iii)] In~\cite{augusiak_asymptotic_2016}, the following inequality is proven:
\begin{align}
     |\mathcal{F}_H(\rho_S)-\mathcal{F}_H(\rho'_S)|\leq 32\|H_S\|_\infty^2\sqrt{\|\rho_S-\rho'_S\|_1}.
\end{align}
Therefore, for $m$ copies of $S$ system, denoted by $S_1\cdots S_m$, we find
\begin{align}
    |\mathcal{F}_{H_{\mathrm{tot}}}(\rho^{\otimes m})-\mathcal{F}_{H_{\mathrm{tot}}}(\sigma_m)|\leq 32\|H_{\mathrm{tot}}\|_\infty^2\sqrt{\|\rho^{\otimes m}-\sigma_m\|_1}=m^2K_S\sqrt{\|\rho^{\otimes m}-\sigma_m\|_1},
\end{align}
where $K_S\coloneqq 32\|H_{S}\|_\infty^2$, implying that Eq.~\eqref{regularity_S} is satisfied with $a_S=2$, $b_S=1/2$, and $c=0$. Note that a constant shift in the Hamiltonian does not affect the QFI, and therefore, $\|H_S\|_\infty^2$ in the definition of $K_S$ can be replaced with $(\Delta_{H_S}/2)^2$.
\end{description}

\subsection{Resource theory of athermality}

Yet another important framework concerning energy conservation is the resource theory of athermality. In this resource theory, the free state is the Gibbs state $\tau$ at a fixed temperature $T$, and the free operations are those that can be implemented through an energy-conserving unitary interaction with a thermal bath prepared in the free state, or by tracing out a subsystem.

The relative entropy of athermality, i.e., $A_R(\rho)\coloneqq \frac{1}{\beta}D(\rho\|\tau)=F(\rho)-F(\tau)$, quantifies athermality of a state $\rho$. Here, we reguralize the relative entropy by $\beta$, to simplify the relation to the non-equlilibrium free energy $F(\rho)=\ex{H}_\rho-\frac{1}{\beta} S(\rho)$. This quantity satisfies Conditions~(i), (ii), and (iii), as shown below:
\begin{description}
    \item[(i)] Since $\mathbb{I}\otimes U$ and $\mathbb{I}\otimes U$ are energy-conserving when $U$ is energy conserving, maps $\id\otimes \mathcal{E}$ and $\id\otimes \mathcal{E}$ are free, where $\mathcal{E}(\cdot)\coloneqq U(\cdot)U^\dag$. Also, $\id\otimes \calE$ is free when $\calE$ is a partial trace operation. Since the relative entropy of athermality $A_R$ is non-increasing under free operations in the resource theory of athermality, Conditions(i) is satisfied. 
    \item[(ii)] By using the non-equilibrium free energy $F(\rho)$, the relative entropy of athermality can also be written as $A_R(\rho)=F(\rho)-F(\tau)$. Consequently, the additivity of $A_R$ follows from the additivity of the free energy. 
    \item[(iii)] By using Fannes' inequality:
\begin{align}
    |S(\rho_S)-S(\rho_S')|&\leq 2\|\rho_S-\rho_S'\|_1 (\log d_S-1)+h\left(\frac{\|\rho_S-\rho_S'\|_1}{2}\right)
\end{align}
with the entropy function $h(x)\coloneqq -x \log x -(1-x)\log (1-x)$, we get
\begin{align}
    |A_R(\rho_S)-A_R(\rho_S')|\leq  \left(\Delta H_S+\frac{2}{\beta}\log d_S\right) \|\rho_S-\rho_S'\|_1+\frac{1}{\beta}h\left(\frac{\|\rho_S-\rho_S'\|_1}{2}\right),
\end{align}
where we have used Eq.~\eqref{eq:continuity_energy_exp}.
Therefore, for $m$ copies of $S$ system, denoted by $S_1\cdots S_m$, we find
\begin{align}
    |A_R(\rho^{\otimes m})-A_R(\sigma_m)|\leq m K_S \|\rho^{\otimes m}-\sigma_m\|_1+\frac{1}{\beta}h\left(\frac{\|\rho^{\otimes m}-\sigma_m\|_1}{2}\right),
\end{align}
where $K_S\coloneqq \Delta H_S+\frac{2}{\beta}\log d_S$, implying that Eq.~\eqref{regularity_S} is satisfied with $a_S=b_S=1$ and $c(x)=\frac{1}{\beta}h(x/2)$. 
\end{description}


\subsection{Resource theory of coherence}

The resource theory of coherence investigates superposition with respect to a preferred basis $\{\ket{i}\}$. Its free states are incoherent states, defined as states that are diagonal in the basis $\{\ket{i}\}$, while there are various options for free operations, such as maximally incoherent operations~\cite{aberg_quantifying_2006}, incoherent operations~\cite{baumgratz_quantifying_2014}, and strictly incoherent operations~\cite{winter_operational_2016}. 

A commonly used resource monotone in this theory is the relative entropy of coherence $C_R(\rho)\coloneqq \min_{\sigma: \text{incoherent states}}D(\rho\|\sigma)$. This quantity satisfies Conditions~(i), (ii), and (iii), as shown below:
\begin{description}
    \item[(i)] Regardless of the choice of the above options of free operations, a unitary operation $U$ is free if and only if it is diagonal in the basis $\{\ket{i}\}$. This implies that maps $\id\otimes \mathcal{E}$ and $\id\otimes \mathcal{E}^\dag$ are free if $\calE(\cdot)\coloneqq U(\cdot)U^\dag$ is a free unitary map. In addition, $\id\otimes \calE$ maps an incoherent state into an incoherent state if $\calE$ is a partial trace operation. Indeed, for any incoherent state $\rho=\sum_{i,j}p_{i,j,k}\ket{i_A\otimes j_B\otimes k_C}\bra{i_A\otimes j_B\otimes k_C}$, the state $(\id_A\otimes \mathrm{Tr}_C)(\rho)=\sum_{i,j,k}p_{i,j,k}\ket{i_A\otimes j_B}\bra{i_A\otimes j_B}$ is incoherent. Therefore, the relative entropy of coherence $C_R$ satisfies Condition~(i). 
    \item[(ii)] As shown in~\cite{aberg_quantifying_2006}, the relative entropy of coherence $C_R$ is given by $C_R(\rho)= S(\Delta(\rho))-S(\rho)$, where $\Delta(\rho)\coloneqq \sum_i\braket{i|\rho|i}\ket{i}\bra{i}$ denotes the dephased state. Additivity of $C_R$ follows directly from this expression. 
    \item[(iii)] As shown in~\cite{takagi_universal_2020} using a general asymptotic continuity property in~\cite{Winter2016tight}, it holds
    \begin{align}
        |C_R(\rho_S)-C_R(\rho_S')|\leq 2\log d_S \|\rho_S-\rho_S'\|_1+(1+\|\rho_S-\rho_S'\|_1/2)h\left(\frac{\|\rho_S-\rho_S'\|_1/2}{1+\|\rho_S-\rho_S'\|_1/2}\right).
    \end{align}
    Therefore, for $m$ copies of $S$ system, denoted by $S_1\cdots S_m$, we get
    \begin{align}
         |C_R(\rho^{\otimes m})-C_R(\sigma_m)|\leq m K_S \|\rho^{\otimes m}-\sigma_m\|_1+(1+\|\rho^{\otimes m}-\sigma_m\|_1/2)h\left(\frac{\|\rho^{\otimes m}-\sigma_m\|_1/2}{1+\|\rho^{\otimes m}-\sigma_m\|_1/2}\right),
    \end{align}
    where $K_S\coloneqq 2\log d_S$, implying that Eq.~\eqref{regularity_S} is satisfied with $a_S=b_S=1$ and $c(x)=(1+x/2)h(x/(2+x))$.
\end{description}


\subsection{Resource theory of magic}

In the resource theory of magic, the free states are the stabilizer states, defined as the convex combination of pure states produced from the computational basis state by Clifford gates, while the free operations are Clifford operations. 
Among the three conditions of resource monotones required for our results, the additivity property is arguably the most non-trivial one. 
In the case of odd dimensions, measures based on min- and max-relative entropies with an extended set of free states based on Wigner negativity were shown to be additive~\cite{wang_efficiently_2020}. 
For the case of single-qubit systems, the logarithm of stabilizer fidelity was shown to be additive~\cite{Bravyi2019simulationofquantum} for pure states and was later extended to general mixed states and monotones based on more general quantum R\'enyi divergences, which include the stabilizer fidelity as a special case~\cite{Rubboli2024mixedstate}. 
In addition, the logarithm of generalized robustness (equivalent to the max-relative entropy of magic) and the dyadic negativity were shown to be additive~\cite{Seddon2021quantifying}.

Here, let us focus on the max-relative entropy of magic as an example.
The max-relative entropy of magic defined by $ \mathcal{D}_{\max}(\rho)\coloneqq \min_{\sigma:\text{stabilizer states}}D_{\max}(\rho\|\sigma)$, where $D_{\max}(\rho\|\sigma)\coloneqq \inf \{r\mid \rho\leq 2^r\sigma\}$ denotes the max-relative entropy, satisfies Conditions~(i), (ii), and (iii), as shown below:
\begin{description}
    \item[(i)] For a Clifford unitary operation $U$, both $\mathbb{I}\otimes U$ and $\mathbb{I}\otimes U^\dag$ are also Clifford. Consequently, both $\id\otimes \calE$ and $\id\otimes \calE^\dag$ are free if $\calE(\cdot)\coloneqq U(\cdot)U^\dag$ is a free unitary map. Since $\id\otimes \calE$ is also Clifford when $\calE$ is a partial trace operation, Condition~(i) follows from the monotonicity of the max-thauma under Clifford operations. 
    \item[(ii)] The additivity is proven in~\cite{wang_efficiently_2020}. 
    \item[(iii)] As shown in~\cite{takagi_universal_2020}, it holds
\begin{align}
    |\mathcal{D}_{\max}(\rho)-\mathcal{D}_{\max}(\rho')|\leq 2d_S\|\rho-\rho'\|_1.
\end{align}
Therefore, for $m$ copies of $S$ system, denoted by $S_1\cdots S_m$, we get
\begin{align}
    |\mathcal{D}_{\max}(\rho^{\otimes m})-\mathcal{D}_{\max}(\sigma_m)|\leq K_Se^{ a_S'm }\|\rho^{\otimes m}-\sigma_m\|_1,
\end{align}
where $K_S\coloneqq 2$ and $a_S'\coloneqq \log d_S$, implying that Eq.~\eqref{regularity_S} is satisfied with $b_S=1$ and $c=0$.
\end{description}


\section{Application 1: Lower bounds for energy costs of arbitrary channels}
\subsection{General formula: tradeoff between energy cost and irreversibility}
In this subsection, we apply our results to the energy cost of channels and give a universal bound for the energy cost for an arbitrary quantum channel.
In the resource theory of energy, the expectation value of energy is a resource measure satisfying the conditions (i), (ii), and (iii) with $a_S=b_S=1$ and $c=0$.
Furthermore, it satisfies the condition (ii-strong) for any set of systems, provided there is no interaction Hamiltonian among these systems.
Therefore, we can apply Theorem \ref{Thm:RItradeoff_alter_SM} to the resource theory of energy.
Then, for the energy cost $E_c(\Lambda)$, we prove the following theorem:
\begin{theorem}\label{Thm:EItradeoff_SM}
Let $A$ and $A'$ be quantum systems with Hamiltonians $H$ and $H'$. For any quantum channel $\Lambda$ from $A$ to $A'$ and any ensemble $\Omega_{1/2}:=(\ket{\psi_1},\ket{\psi_2})_{1/2}$ whose $\ket{\psi_1}$ and $\ket{\psi_2}$ are mutually orthogonal, the following relation holds:
\eq{
E_{c}(\Lambda)\ge\frac{\calC(\Lambda,\Omega_{1/2})^2}{2(\Delta_{H}+\Delta_{H'})f(\delta(\Lambda,\Omega_{1/2}))}-\calC(\Lambda,\Omega_{1/2}).\label{EItradeoff_SM}
}
Here, $\Delta_X$ is the difference between the maximum and minimum eigenvalues of $X$, and $\calC(\Lambda,\Omega_{1/2}):=|\bra{\psi_2}H-\Lambda^\dagger(H')\ket{\psi_1}|$.
\end{theorem}
This inequality conveys a crucial message: if a channel satisfies
$\mathcal{C}(\Omega_{1/2},\Lambda)>0 $ and $\delta(\Lambda,\Omega_{1/2})=0$ for at least one ensemble $\Omega_{1/2}$, its implementation requires infinite energy.

\textbf{Deriving \eqref{EItradeoff} in the main text:}
By substituting $f(x)=4x+x^2$ into \eqref{EItradeoff_SM} and using $\frac{1}{1+x}\ge1-x$, we get
\eq{
E_{c}(\Lambda)&\ge\frac{\calC(\Lambda,\Omega_{1/2})^2}{2(\Delta_{H}+\Delta_{H'})f(\delta(\Lambda,\Omega_{1/2}))}-\calC(\Lambda,\Omega_{1/2})\nonumber\\
&\ge\frac{\calC(\Lambda,\Omega_{1/2})^2}{8(\Delta_{H}+\Delta_{H'})\delta(\Lambda,\Omega_{1/2})}-\frac{\calC(\Lambda,\Omega_{1/2})^2}{32(\Delta_H+\Delta_{H'})}-\calC(\Lambda,\Omega_{1/2})\nonumber\\
&\ge\frac{\calC(\Lambda,\Omega_{1/2})^2}{8(\Delta_{H}+\Delta_{H'})\delta(\Lambda,\Omega_{1/2})}-\frac{65}{64}\calC(\Lambda,\Omega_{1/2})\label{middle}\\
&\ge\frac{\calC(\Lambda,\Omega_{1/2})^2}{8(\Delta_{H}+\Delta_{H'})\delta(\Lambda,\Omega_{1/2})}-2\calC(\Lambda,\Omega_{1/2})
}

Here, in the second-to-last line, we used $\calC(\Lambda,\Omega_{1/2})\le(\Delta_H+\Delta_{H'})/2$, which is derived as follows:
\eq{
\calC(\Lambda,\Omega_{1/2})&\le|\bra{\psi_1}H\ket{\psi_2}|+
|\bra{\psi_1}\Lambda^\dagger(H')\ket{\psi_2}|\nonumber\\
&\le\frac{\Delta_H}{2}+\frac{\Delta_{\Lambda^\dagger(H')}}{2}\nonumber\\
&\le\frac{\Delta_H}{2}+\frac{\Delta_{H'}}{2}.
}

\begin{proofof}{Theorem \ref{Thm:EItradeoff_SM}}
As shown in the section \ref{app:application_list}, in the resource theory of energy, $K_S=\Delta_{H_S}/2$ ($H_S$ is the Hamiltonian of $S$), $a_S=b_S=1$ and $c=0$ hold.
We also define the register $\mathscr{K}:=(K,\{\ket{k}\})$ whose Hamiltonian $H_K$ is trivial.
Since there are no interaction Hamiltonians between $A$,$A'$,$B$, $B'$ and $K$, the condition (ii-strong) holds for $\{A,A',B,B',K\}$.
Therefore, it suffices to prove that the following equality holds for a state $\rho$ satisfying $\Lambda_{\mathbb{P}}(\rho)=\sum_k\frac{1}{2}\ket{\psi_k}\bra{\psi_k}\otimes\ket{k}\bra{k}_K$, where $\mathbb{P}:=\{P_k\}_{k=1,2}$ is a PVM discriminating $\psi_1$ and $\psi_2$ with probability 1:
\eq{
M_\rho(\Lambda:\Lambda_{\mathbb{P}})=\calC(\Lambda,\Omega_{1/2})=|\bra{\psi_1}H-\Lambda^\dagger(H')\ket{\psi_2}|.\label{energy_key}
}

Consider a pure state $\rho=\ket{\rho}\bra{\rho}$, where $\ket{\rho}=\frac{\ket{\psi_1}+e^{-i\theta}\ket{\psi_2}}{\sqrt{2}}$. For any $\theta\in\mathbb{R}$, this state satisfies the condition  $\Lambda_{\mathbb{P}}(\rho)=\sum_k\frac{1}{2}\ket{\psi_k}\bra{\psi_k}\otimes\ket{k}\bra{k}_K$. From
\eq{
M_{\rho}(\Lambda_{\mathbb{P}})&=M(\Lambda_{\mathbb{P}}(\rho))-M(\rho)\nonumber\\
&=\frac{-e^{-i\theta}\bra{\psi_1}H\ket{\psi_2}-e^{i\theta}\bra{\psi_2}H\ket{\psi_1}}{2},
}
and
\eq{
M_{\rho}(\Lambda)-M_{\rho}(\id^{K}\otimes\Lambda\circ\Lambda_{\mathbb{P}})&=M(\Lambda(\rho))-M(\id^{K}\otimes\Lambda\circ\Lambda_{\mathbb{P}}(\rho))\nonumber\\
&=\frac{e^{-i\theta}\bra{\psi_1}\Lambda^\dagger(H')\ket{\psi_2}+e^{i\theta}\bra{\psi_2}\Lambda^\dagger(H')\ket{\psi_1}}{2},
}
we get
\eq{
M_{\rho}(\Lambda:\Lambda_{\mathbb{P}})&=M_{\rho}(\Lambda_{\mathbb{P}})+M_{\rho}(\Lambda)-M_{\rho}(\id^{K}\otimes\Lambda\circ\Lambda_{\mathbb{P}})\nonumber\\
&=Re(e^{i\theta}\bra{\psi_2}\Lambda^\dagger(H')-H\ket{\psi_1}).
}
Therefore, by choosing $\theta$ appropriately, we obtain
\eq{
M_{\rho}(\Lambda:\Lambda_{\mathbb{P}})=|\bra{\psi_2}H-\Lambda^\dagger(H')\ket{\psi_1}|.
}
\end{proofof}

\subsection{Energy-error bounds for  pairwise reversible channels}
The inequality \eqref{EItradeoff_SM} implies that any channel $\Lambda$ satisfying $\calC(\Lambda,\Omega_{1/2})>0$ and  $\delta(\Lambda,\Omega_{1/2})=0$ requires an infinite energy cost.
Furthermore, the following corollary of the inequality \eqref{EItradeoff_SM} restricts the energy cost of approximate implementation of $\Lambda$:
\begin{corollary}\label{Coro:EEtradeoff_SM}
Let $A$ and $A'$ be quantum systems with Hamitlonian $H$ and $H'$. 
Let $\Lambda$ be a quantum channel  from $A$ to $A'$ satisfying $\calC(\Lambda,\Omega_{1/2})>0$ and  $\delta(\Lambda,\Omega_{1/2})=0$  for an ensemble $\Omega_{1/2}:=(\ket{\psi_1},\ket{\psi_2})_{1/2}$ whose $\ket{\psi_1}$ and $\ket{\psi_2}$ are mutually orthogonal. Then, the energy cost $E^{\epsilon}_{c}(\Lambda)$ of approximate implementation of $\Lambda$ satisfies
\eq{
E^{\epsilon}_{c}(\Lambda)\ge\frac{\calC(\Lambda,\Omega_{1/2})^2}{8(\Delta_{H}+\Delta_{H'})\epsilon}-2\calC(\Lambda,\Omega_{1/2})-3\Delta_{H'}\epsilon\label{EEtradeoff_SM}.
}
\end{corollary}

\begin{proofof}{Corollary \ref{Coro:EEtradeoff_SM}}
Let $\Lambda'$ be a CPTP map from $A$ to $A'$ satisfying $\Lambda'\sim_{\epsilon}\Lambda$.
As proven in \cite{Tajima_Takagi2025}, the following relation holds for such a map $\Lambda'$:
\eq{
\delta(\Omega_{1/2},\Lambda')\le\epsilon.
}
Due to $\calC(\Omega_{1/2},\Gamma)=\|\psi_1(H-\Gamma^\dagger(H'))\psi_2\|_2$ and the triangle inequality of 2-norm, we obtain
\eq{
|\calC(\Lambda,\Omega_{1/2})-\calC(\Lambda',\Omega_{1/2})|\le
\|\psi_1(\Lambda^\dagger(H')-\Lambda'^\dagger(H'))\psi_2\|_2.
}
Combining the above and the following inequality in Ref.~\cite{Tajima_Takagi2025},
\eq{
\|\psi_1(\Lambda^\dagger(H')-\Lambda'^\dagger(H'))\psi_2\|_2\le2\Delta_{H'}\epsilon,
}
we obtain
\eq{
|\calC(\Lambda,\Omega_{1/2})-\calC(\Lambda',\Omega_{1/2})|\le2\Delta_{H'}\epsilon.
}
Therefore, substituting this relation into \eqref{middle}, we get
\eq{
E_{c}(\Lambda')&\ge\frac{|\calC(\Lambda,\Omega_{1/2})-2\Delta_{H'}\epsilon|^2_+}{8(\Delta_{H}+\Delta_{H'})\epsilon}-\frac{65}{64}\calC(\Lambda,\Omega_{1/2})-3\Delta_{H'}\epsilon\nonumber\\
&\ge
\frac{\calC(\Lambda,\Omega_{1/2})^2-4\Delta_{H'}\epsilon\calC(\Lambda,\Omega_{1/2})}{8(\Delta_{H}+\Delta_{H'})\epsilon}-\frac{65}{64}\calC(\Lambda,\Omega_{1/2})-3\Delta_{H'}\epsilon\nonumber\\
&\ge\frac{\calC(\Lambda,\Omega_{1/2})^2}{8(\Delta_{H}+\Delta_{H'})\epsilon}-2\calC(\Lambda,\Omega_{1/2})-3\Delta_{H'}\epsilon.
}
\end{proofof}

\subsection{Energy-Error tradeoff for specific classes of channels}
In this subsection, we present several classes of quantum channels satisfying $\calC(\Lambda,\Omega_{1/2})>0$ and  $\delta(\Lambda,\Omega_{1/2})=0$ for an ensemble $\Omega_{1/2}$.
\subsubsection{Projective measurements} 
The first example is the projective measurement.
For any quantum system $A$ with Hamiltonian $H$ and any projective measurement channel from $A$ to $K$, defined by
\eq{
\Gamma_{\mathbb{P}}(\cdot):=\sum^{r}_{k=1}\Tr[P_k\cdot]\ket{k}\bra{k}_K
}
whose register $\mathscr{K}:=(K,\{\ket{k}\}^{r}_{k=1})$ has a trivial Hamiltonian $H_K=0$, it is known that there is a $\Omega_{1/2}$ satisfying the following relations \cite{TTK}:
\eq{
\calC(\Gamma_{\mathbb{P}},\Omega_{1/2})&=2\max_{k}\|[P_k,H]\|_{\infty}\label{cos_meas1},\\
\delta(\Gamma_{\mathbb{P}},\Omega_{1/2})&=0.\label{cos_meas2}
}
Therefore, Corollary \ref{Coro:EEtradeoff_SM} implies
\eq{
E^{\epsilon}_c(\Gamma_{\mathbb{P}})\ge\max_{k}\frac{\|[P_k,H]\|^2_{\infty}}{2\Delta_{H}\epsilon}-2\Delta_{H}.
}
Here, we used $\Delta_{H_K}=0$ and 
\eq{
\|[P_k,H]\|_{\infty}&=|\ex{[P_k,H]}_{\psi}|\nonumber\\
&\le2\sqrt{V_{\psi}(P_k)}\sqrt{V_{\psi}(H)}\nonumber\\
&\le\frac{1}{2}\Delta_{H},
}
where $V_\sigma(X):=\ex{X^2}_\sigma-\ex{X}^2_\sigma$ denotes the variance.
Therefore, any projective measurement with projective operators that does not commute with $H$ requires an infinite amount of energy.

We remark on two points.
First, the number of outcomes of the measurement may be any $r$, not necessarily two.
Second, the same bound is valid for the channel 
\eq{
\Lambda_{\mathbb{P}}(\cdot):=\sum^{r}_{k=1}P_k\cdot P_k\otimes\ket{k}\bra{k}_K,
}
i.e.,
\eq{
E^{\epsilon}_c(\Lambda_{\mathbb{P}})\ge\max_{k}\frac{\|[P_k,H]\|^2_{\infty}}{2\Delta_{H}\epsilon}-2\Delta_{H}.
}
This is because the channel $\Gamma_{\mathbb{P}}$ is constructed from $\Lambda_{\mathbb{P}}$ without energy cost, simply by taking a partial trace. 

\subsubsection{Unitary gates and channels realized with an energy-preserving channel and a unitary gate} 
The second example is the unitary gate.
For any quantum system $A$ with Hamiltonian $H$ and any unitary channel $\calU(\cdot):=U\cdot U^\dagger$ on $A$, there is a $\Omega_{1/2}$ such that \cite{TTK}
\eq{
\calC(\Omega_{1/2},\calU)&=\frac{1}{2}\Delta_{H-U^\dagger HU},\label{cos_uni1}\\
\delta(\Omega_{1/2},\calU)&=0.\label{cos_uni2'}
}
Therefore, Corollary \ref{Coro:EEtradeoff_SM} implies
\eq{
E^{\epsilon}_c(\calU)\ge\frac{\Delta^2_{H-U^\dagger HU}}{64\Delta_{H}\epsilon}-2\Delta_{H}-3\Delta_{H}\epsilon,
}
where we used $\Delta_{H-U^\dagger HU}\le2\Delta_{H}$.

Furthermore, our results also restrict channels written in the following form:
\eq{
\Lambda=\calN\circ\calU,
}
where $U$ is a unitary such that $\braket{e_1|U^\dagger H U|e_2}\neq 0$ for eigenstates $\ket{e_1}$ and $\ket{e_2}$ of $H$, and $\calN$ is a channel from $A$ to $A$ satisfying $\calN(\dm{e_{1,2}})=\dm{e_{1,2}}$.
Indeed, it is shown \cite{TTK} that there is a $\Omega_{1/2}$ such that
\eq{
\calC(\Omega_{1/2},\calU)&=|\braket{e_1|U^\dagger H U|e_2}|,\label{cos_uni1'}\\
\delta(\Omega_{1/2},\calU)&=0.\label{cos_uni2}
}
Therefore, Corollary \ref{Coro:EEtradeoff_SM} implies
\eq{
E^{\epsilon}_c(\calN\circ\calU)\ge\frac{|\braket{e_1|U^\dagger H U|e_2}|^2}{16\Delta_{H}\epsilon}-\Delta_{H}-3\Delta_{H}\epsilon,
} 
where we used $|\braket{e_1|U^\dagger H U|e_2}|\le\frac{1}{2}\Delta_H$.

\subsubsection{Gibbs-preserving operations}
A Gibbs-preserving operation is a quantum operation that leaves a given Gibbs state invariant. Specifically, consider a quantum system described by a Hamiltonian $H$, and its corresponding Gibbs state $\gamma = e^{-\beta H}/\mathrm{Tr}(e^{-\beta H})$. A Gibbs-preserving operation $\mathcal{E}$ satisfies $\mathcal{E}(\gamma) = \gamma$, reflecting the physical constraint that such an operation cannot create or destroy thermal equilibrium at inverse temperature $\beta$. These operations play a central role in quantum thermodynamics and resource theories, serving as fundamental constraints for thermodynamic processes at the quantum level.

In Ref.~\cite{Tajima_Takagi2025}, it is shown that there are infinitely many Gibbs preserving operations satisfying $\calC(\Lambda,\Omega_{1/2})>0$ and  $\delta(\Lambda,\Omega_{1/2})=0$ for a proper pure orthogonal ensemble $\Omega_{1/2}$.
Therefore, such Gibbs-preserving operations require infinite energy cost, and the energy costs of their approximate implementations are inversely proportional to the implementation error. 

\section{Application 2: Lower bounds for athermality (and work) costs of arbitrary channels}
In this section, we apply our results to the resource theory of athermality.
In this application, we treat the relative entropy of athermality as the resource measure $M$:
\eq{
A_R(\rho):=\frac{1}{\beta}D(\rho\|\tau)=F(\rho)-F(\tau).
}
This quantity represents the difference between the non-equilibrium free energy and the equilibrium energy, and is sometimes called the ``free energy difference.''
This quantity also represents the maximum amount of energy extractable from the system by an isothermal process in which the Hamiltonian remains unchanged before and after the operation. In this sense, it quantifies the extrable energy from the system via an isothermal process.
We use the symbols $A_c$, $\Aem$ and $\Acos$ for $M_c$, $\Mem$ and $\Mcos$ whose $M$ are substituted by $A_R$.

\subsection{Athermality-irreversibitlity tradeoff and Athermality-error tradeoff}
In this case, our main theorem \ref{Thm:RItradeoff_SM} reduces to the following theorem:
\begin{theorem}\label{Thm:ATItradeoff}
Let $A$ and $A'$ be quantum systems with Hamiltonians $H$ and $H'$.
For any channel $\Lambda$ from $A$ to $A'$, any test ensemble $\Omega_p$, and any two-level register $\mathscr{K}:=(K,\{\ket{k}_K\}_{k=1,2})$ with Hamiltonian $H_K$ the following inequality holds:
\eq{
A_c(\Lambda)\ge\frac{|\Aem(\Omega_p|\mathscr{K})-\Acos(\Omega'_{p,\Lambda}|\mathscr{K})|^2_+}{16K_A\delta(\Lambda,\Omega_p)}-h',\label{eq:ATItradeoff}
}
where $K_A=\Delta_H+\frac{2}{\beta}\log d_A$, and the function $|x|_+$ returns $x$ when $x$ is positive, and returns $0$ otherwise. And 
\eq{
h':=\frac{1}{\beta}\left(\log2+|\Aem(\Omega_p|\mathscr{K})-\Acos(\Omega'_{p,\Lambda}|\mathscr{K})|_++\frac{|\Aem(\Omega_p|\mathscr{K})-\Acos(\Omega'_{p,\Lambda}|\mathscr{K})|^2_+}{64K_A}.\right)
}

Furthermore, when there is a test ensemble $\Omega_p$ and a register $\mathscr{K}$ satisfying
\eq{
|\Aem(\Omega_p|\mathscr{K})-\Acos(\Omega'_{p,\Lambda}|\mathscr{K})|_+&>0\label{AD-cond1},
\\
\delta(\Lambda,\Omega_p)&=0,\label{AD-cond2}
}
namely, when the cost of $\Lambda$ diverges, the following relation holds:
\eq{
A^{\epsilon}_c(\Lambda)\ge\frac{|\Aem(\Omega_p|\mathscr{K})-\Acos(\Omega'_{p,\Lambda}|\mathscr{K})|^2_+}{16K_A\epsilon}-h',\label{eq:ATEtradeoff}
}
\end{theorem}

\subsection{Work-irreversibility tradeoff and Work-error tradeoff}
Therefore, in the standard thermal‑operation with a ``work‑bit'' setup, where the free‑energy source is only the work bit, the inequalities \eqref{eq:ATItradeoff} and \eqref{eq:ATEtradeoff} give  analogous lower bounds for the required work cost of channels.

Let us consider the work bit $B_W$ described as a qubit system with Hamiltonian
\eq{
H_{B_W}:=W\ket{1}\bra{1}_{B_W}.
}
Then, the following relation holds:
\eq{
A_R(\ket{1}\bra{1}_{B_W})&=\frac{1}{\beta}(-S(\dm{1})-\Tr[\dm{1}\log \frac{e^{-\beta H_{B_W}}}{Z(\beta,H_{B_W})}])\nonumber\\
&=\frac{-1}{\beta}\log \frac{e^{-\beta W}}{1+e^{-\beta W}}\nonumber\\
&\le W+\frac{1}{\beta}\log2.
}
Therefore, when we implement a quantum channel with $\ket{1}_{B_W}$ on the work bit and the Gibbs state $\tau_{B_H}$ on an external heat bath $B_H$, we get
\eq{
W_c(\Lambda)&\ge A_c(\Lambda)-\frac{1}{\beta}\log2,\\
W^\epsilon_c(\Lambda)&\ge A^\epsilon_c(\Lambda)-\frac{1}{\beta}\log2.
}
These relations convert \eqref{eq:ATEtradeoff} and \eqref{eq:ATItradeoff} to the lower bound for the work cost.
For example, when there are $\Omega_p$ and $\mathscr{K}$ satisfying \eqref{AD-cond1} and \eqref{AD-cond2}, we have
\eq{
W^{\epsilon}_c(\Lambda)\ge\frac{|\Aem(\Omega_p|\mathscr{K})-\Acos(\Omega'_{p,\Lambda}|\mathscr{K})|^2_+}{16K_A\epsilon}-h'-\frac{1}{\beta}\log2.
}

\subsection{Athermality-Error tradeoff and Work-error tradeoff for specific classes of channels}
\subsubsection{General formula for coherence-erasure channels}
Just as in the case of constraints on energy cost, a broad class of channels exists for the athermality cost as well, for which there exist $\Omega_p$ and $\mathscr{K}$ satisfying \eqref{AD-cond1} and \eqref{AD-cond2} (thus causing the cost to diverge).
When do the conditions $\Aem(\Omega_p|\mathscr{K})-\Acos(\Omega'_{p,\Lambda}|\mathscr{K})>0$ and $\delta(\Lambda,\Omega_{p})=0$ hold?
Many channels that erase superpositions between energy eigenstates having different energies satisfy the conditions. In fact, the following theorem holds:
\begin{theorem}[Restatement of Theorem \ref{Thm:c_erase}]\label{Thm:c_erase_SM}
Let $A$ be a $d$-level quantum system, and $\Lambda$ be a CPTP map from $A$ to $A$. Let $H$ be the Hamiltonian of $A$, and $\{\ket{j}\}$ be eigenbasis of $H$.
We choose $\Omega_{1/2}:=(\ket{\psi_{j,j'}},\ket{\psi^{\perp}_{j,j'}})_{1/2}$ as
\eq{
\ket{\psi_{j,j'}}:=\frac{\ket{j}+\ket{j'}}{\sqrt{2}},\enskip
\ket{\psi^{\perp}_{j,j'}}:=\frac{\ket{j}-\ket{j'}}{\sqrt{2}}.
}
When $\Lambda$ satisfies
\eq{
\Lambda(\psi_{j,j'})=\ket{j_0}\bra{j_0},\enskip\Lambda(\psi^{\perp}_{j,j'})=\ket{j_1}\bra{j_1},\enskip j_0\ne j_1,
}
the following hold:
\eq{
\delta(\Lambda,\Omega_{1/2})&=0,\label{delta_0_thmS6}\\
\Aem(\Omega_{1/2}|\mathscr{K})-\Acos(\Omega'_{1/2,\Lambda}|\mathscr{K})&>\frac{1}{2}|E_j-E_{j'}|-\frac{2}{\beta}\log2,\label{Abound}
}
where $\mathscr{K}:=(K,\{\ket{k}\}_{k=1,2})$ whose Hamiltonian $H_K$ is the trivial Hamiltonian.
\end{theorem}
Namely, when $|E_j-E_{j'}|\ge\frac{4}{\beta}\log2$, the athermality cost (and work cost) of approximate implementation of $\Lambda$ is inversely proportional to the error.
This observation implies that erasing coherence between different energy levels reversibly costs an infinite amount of free‑energy. 

To derive Theorem \ref{Thm:c_erase_SM}, we firstly show the following lemma which gives an upper bound of the resource-generating power of projective measurements whose PVMs commute with the Hamiltonian:
\begin{lemma}\label{Lemm:cost_of_meas}
Let $A$ be a quantum system with Hamiltonian $H$.
Let $\Lambda_{\mathbb{P}}$ be a two-valued projective measurement channel $\Lambda_{\mathbb{P}}(\cdot):=\sum_{k=1,2}P_k\cdot P_k\otimes \ket{k}\bra{k}_K$, where $\mathbb{P}:=\{P_k\}_{k=1,2}$ is a two-valued PVM on $A$ satisfying $[P_k,H]=0$ and the two-level register $\mathscr{K}:=(K,\{\ket{k}_{K}\})$ has the trivial Hamiltonian $H_K=0$.
Then, 
\eq{
A_R(\Lambda_{\mathbb{P}})\le\frac{1}{\beta}\log4.
}
\end{lemma}

\begin{proofof}{Lemma \ref{Lemm:cost_of_meas}}
Note that the channel $\Lambda_{\mathbb{P}}$ can be realized by the implementation $(U,\ket{0}_K\otimes\ket{0}_{K'})$.
Here $K'$ is a copy of $K$ whose Hamiltonian is also the trivial Hamiltonian $H_{K'}=0$, and $U$ is a unitary on $AKK'$ defined as
\eq{
U:=P_0\otimes 1_K\otimes1_{K'}+P_1\otimes(\ket{1}\bra{0}+\ket{0}\bra{1})_{K}\otimes(\ket{1}\bra{0}+\ket{0}\bra{1})_{K'}.
}
The unitary $U$ satisfies
\eq{
U^\dagger U&=P_0\otimes1_K\otimes1_{K'}+P_1\otimes1_K\otimes1_{K'}=1_{AKK'}\\
UU^\dagger&=
P_0\otimes1_K\otimes1_{K'}+P_1\otimes1_K\otimes1_{K'}=1_{AKK'}\\
[U,H+H_K+H_{K'}]&=[U,H\otimes1_{KK'}]=0,
}
and $U$ and $\ket{0}_K\otimes\ket{0}_{K'}$ satisfy
\eq{
U\ket{0}_K\otimes\ket{0}_{K'}
=
P_0\otimes\ket{0}_K\otimes\ket{0}_{K'}
+
P_1\otimes\ket{1}_K\otimes\ket{1}_{K'},
}
and thus
\eq{
\Lambda_{\mathbb{P}}(\cdot)=\Tr_{K'}[U(\cdot\otimes\ket{0}\bra{0}_K\otimes\ket{0}\bra{0}_{K'})U^\dagger].
}
Therefore, $(U,\ket{0}_K\otimes\ket{0}_{K'})$ implements $\Lambda_{\mathbb{P}}$ and $U$ is an energy-preserving unitary.

By using
\eq{
A_R(\ket{0}\bra{0}_{K}\otimes \ket{0}\bra{0}_{K'})&=2A_R(\ket{0}\bra{0}_{K})\nonumber\\
&=\frac{2}{\beta}D(\ket{0}\bra{0}_{K}\|\tau_{K})\nonumber\\
&=\frac{2}{\beta}(S(\ket{0}\bra{0})-\Tr[\ket{0}\bra{0}\log\frac{I_K}{2}])\nonumber\\
&=-\frac{2}{\beta}\bra{0}\log\frac{1_K}{2}\ket{0}\nonumber\\
&=\frac{2}{\beta}\log2,
}
we prove $A_R(\Lambda_{\mathbb{P}})\le\frac{2}{\beta}\log2$. 
Let us take an arbitrary external reference system $R$ and its Hamiltonian $H_R$. For an energy-preserving (and thus free) unitary operator $U'=U\otimes1_R$, we have
\eq{
\id_R\otimes\Lambda_{\mathbb{P}}(\cdot)=\Tr_{K'}[U'(\cdot\otimes\ket{0}\bra{0}_K\otimes\ket{0}\bra{0}_{K'}])U'^\dagger]\label{contra1}.
}
Therefore, for any state $\sigma_{AR}$ on $AR$, we get
\eq{
A_R(\id_R\otimes\Lambda_{\mathbb{P}}(\sigma_{AR}))&\le A_R(\sigma_{AR}\otimes\ket{0}\bra{0}_K\otimes\ket{0}\bra{0}_{K'})\nonumber\\
&=A_R(\sigma_{AR})+A_R(\ket{0}\bra{0}_K\otimes\ket{0}\bra{0}_{K'})\nonumber\\
&\le A_R(\sigma_{AR})+\frac{2}{\beta}\log2,\label{contra2}
}
which implies
\eq{
A_R(\Lambda_{\mathbb{P}})\le\frac{2}{\beta}\log2.
}
\end{proofof}

Next, we give a lower bound of the resource gain via discrimination of states in $\Omega_p$:
\begin{lemma}\label{Lemm:gain_of_meas}
Let $A$ be a quantum system with Hamiltonian $H$ whose eigenbasis is $\{\ket{i}\}$.
Let $\Omega_p:=(\ket{\psi},\ket{\psi^{\perp}})_p$ is a test ensemble whose $\ket{\psi}$ and $\ket{\psi^{\perp}}$ are mutually orthogonal and defined as
\eq{
\ket{\psi}&:=a_j\ket{j}+a_{j'}\ket{j'},\\
\ket{\psi^{\perp}}&:=a^*_{j'}\ket{j}-a^*_{j}\ket{j'},
}
and $p$ is defined as
\eq{
p:=\frac{r_j}{r_j+r_{j'}}|a_j|^2+
\frac{r_{j'}}{r_j+r_{j'}}|a_{j'}|^2,\label{eq:p_definition}
}
where $r_i$ is the eigenvalue of the $i$-th eigenvector of the Gibbs state $\tau$ of $H$:
\eq{
r_i:=\bra{i}\tau\ket{i}.
}
Then, for a two-level register $\mathscr{K}:=(K,\{\ket{k}\}_{k=1,2})$ whose Hamiltonian is trivial $H_K=0$, the following inequality holds:
\eq{
\Aem(\Omega_p|\mathscr{K})\ge\frac{1}{\beta}(\log2
-h(p)
-(p|a_j|^2+(1-p)|a_{j'}|^2)\log r_j
-(p|a_{j'}|^2+(1-p)|a_j|^2)\log r_{j'}
+\log(r_j+r_{j'})
)\label{comb_res}.
}

\end{lemma}

\begin{proofof}{Lemma \ref{Lemm:gain_of_meas}}
Let $\mathbb{P}:=\{P_k\}_{k=1,2}$ be a two-valued PVM discriminating $\psi$ and $\psi^{\perp}$ with probability 1.
We define a projective measurement channel $\Lambda_{\mathbb{P}}$ as 
\eq{
\Lambda_{\mathbb{P}}(\cdot):=\sum_{k}P_k\cdot P_k\otimes\ket{k}\bra{k}_K.
}
Introducing a state
\eq{
\tilde{\tau}:=\frac{r_{j}}{r_j+r_{j'}}\dm{j}
+
\frac{r_{j'}}{r_j+r_{j'}}\dm{j'},
}
we obtain
\eq{
\Lambda_{\mathbb{P}}(\tilde{\tau})
&=
\left(\frac{r_j}{r_j+r_{j'}}|a_j|^2+
\frac{r_{j'}}{r_j+r_{j'}}|a_{j'}|^2\right)\dm{\psi}_A\otimes\dm{0}_K
+
\left(\frac{r_j}{r_j+r_{j'}}|a_{j'}|^2+
\frac{r_{j'}}{r_j+r_{j'}}|a_{j}|^2\right)\dm{\psi^{\perp}}_A\otimes\dm{1}_K\nonumber\\
&=
p\dm{\psi}_A\otimes\dm{0}_K
+
(1-p)\dm{\psi^{\perp}}_A\otimes\dm{1}_K.
}
Therefore, we find
\eq{
\Aem(\Omega_p|\mathscr{K})\ge A_R(\Lambda_{\mathbb{P}}(\tilde{\tau}))
-
A_R(\tilde{\tau}).
}

From the definition of $A_R$, $A_R(\Lambda_{\mathbb{P}}(\tilde{\tau}))$ is evaluated as
\eq{
A_R(\Lambda_{\mathbb{P}}(\tilde{\tau}))&=
\frac{1}{\beta}D\left(\Lambda_{\mathbb{P}}(\tilde{\tau})\middle\|\tau\otimes\frac{1_K}{2}\right)\nonumber\\
&=\frac{1}{\beta}\left(-S(\Lambda_{\mathbb{P}}(\tilde{\tau}))
-\Tr\left[\Lambda_{\mathbb{P}}(\tilde{\tau})\log\left(\tau\otimes\frac{1_K}{2}\right)\right]\right)\nonumber\\
&=\frac{1}{\beta}\left(-h(p)
-p\Braket{\psi\otimes 0|\log\left(\tau\otimes\frac{1_K}{2}\right)|\psi\otimes0}
-(1-p)\Braket{\psi^{\perp}\otimes 1|\log\left(\tau\otimes\frac{1_K}{2}\right)|\psi^{\perp}\otimes 1}
\right).
}
Each term is evaluated as follows:
\eq{
\Braket{\psi\otimes 0|\log\left(\tau\otimes\frac{1_K}{2}\right)|\psi\otimes0}
&=
\Braket{\psi\otimes0 |\left(\log\frac{r_j}{2}\vert j\rangle\langle j\vert_A\otimes\vert 0\rangle\langle 0\vert_K
+
\log\frac{r_{j'}}{2}\vert j'\rangle\langle j'\vert_A\otimes\vert0\rangle\langle 0\vert_K
\right)|\psi\otimes 0}\nonumber\\
&=
\log\frac{r_j}{2}|\braket{j|\psi}|^2
+
\log\frac{r_{j'}}{2}|\braket{j'|\psi}|^2\nonumber\\
&=|a_j|^2\log r_j+|a_{j'}|^2\log r_{j'}-\log2,\\
\Braket{\psi^{\perp}\otimes 1|\log\left(\tau\otimes\frac{1_K}{2}\right)|\psi^{\perp}\otimes 1}
&=
\Braket{\psi^\perp \otimes 1|\left(\log\frac{r_j}{2}\vert j\rangle\langle j\vert_A\otimes\vert 1\rangle\langle 1\vert_K
+
\log\frac{r_{j'}}{2}\vert j'\rangle\langle j'\vert_A\otimes\vert1\rangle\langle 1\vert_K
\right)|\psi^\perp\otimes 1}\nonumber\\
&=
\log\frac{r_j}{2}|\braket{j|\psi^{\perp}}|^2
+
\log\frac{r_{j'}}{2}|\braket{j'|\psi^{\perp}}|^2\nonumber\\
&=|a_{j'}|^2\log r_j+|a_{j}|^2\log r_{j'}-\log2.
}
Therefore, we obtain
\eq{
A_R(\Lambda_{\mathbb{P}}(\tilde{\tau}))&=
\frac{1}{\beta}
(\log2-h(p)-p(|a_j|^2\log r_j+|a_{j'}|^2\log r_{j'})
-
(1-p)(|a_{j'}|^2\log r_j+|a_{j}|^2\log r_{j'}))\nonumber\\
&=\frac{1}{\beta}
(\log2-h(p)-(p|a_j|^2+(1-p)|a_{j'}|^2)\log r_j
-
(p|a_{j'}|^2+(1-p)|a_{j}|^2)\log r_{j'})\label{comb1}.
}

Also, $A_R(\tilde{\tau})$ is evaluated as
\eq{
A_R(\tilde{\tau})&=\frac{1}{\beta}D(\tilde{\tau}\|\tau)\nonumber\\
&=\frac{1}{\beta}(-S(\tilde{\tau})
-\Tr[\tilde{\tau}\log\tau])\nonumber\\
&=\frac{1}{\beta}\left(-h\left(\frac{r_j}{r_j+r_{j'}}\right)
-\frac{r_j}{r_j+r_{j'}}\log r_j
-\frac{r_{j'}}{r_j+r_{j'}}\log r_{j'}
\right)\nonumber\\
&=
\frac{1}{\beta}\left(\frac{r_j}{r_j+r_{j'}}\log\frac{r_j}{r_j+r_{j'}}
+
\frac{r_{j'}}{r_j+r_{j'}}\log\frac{r_{j'}}{r_j+r_{j'}}
-\frac{r_j}{r_j+r_{j'}}\log r_j
-\frac{r_{j'}}{r_j+r_{j'}}\log r_{j'}
\right)\nonumber\\
&=-\frac{1}{\beta}\log(r_j+r_{j'}).
\label{comb2}
}
Combining \eqref{comb1} and \eqref{comb2}, we obtain \eqref{comb_res}.
\end{proofof}

Using these two lemmas, we prove Theorem \ref{Thm:c_erase_SM}.
\begin{proofof}{Theorem \ref{Thm:c_erase_SM}}
We define a CPTP recovery map $R_*$ from $A$ to $A$ as
\eq{
\calR_*(\cdot)=\bra{j_0}\cdot\ket{j_0}\dm{\psi_{j,j'}}+\bra{j_1}\cdot\ket{j_1}\dm{\psi^{\perp}_{j,j'}}+\Tr[(1-\dm{j_0}-\dm{j_1})\cdot]\dm{j_0}.
}
Then, 
\eq{
\delta(\Lambda,\Omega_{1/2})^2&\le\frac{D_F(\dm{\psi_{j,j'}},R_*\circ\Lambda(\dm{\psi_{j,j'}})^2+
D_F(\dm{\psi^\perp_{j,j'}},R_*\circ\Lambda(\dm{\psi_{j,j'}})^2
}{2}\nonumber\\
&=0.
}
Thererefore, \eqref{delta_0_thmS6} is valid.

Let us show \eqref{Abound}.
Substituting $a_j=a_{j'}=\frac{1}{\sqrt{2}}$ into the definition of $p$ in \eqref{eq:p_definition}, we obtain 
\eq{
p=\frac{1}{2}\frac{r_j}{r_{j}+r_{j'}}+\frac{1}{2}\frac{r_{j'}}{r_{j}+r_{j'}}=\frac{1}{2}.
}
Therefore, Lemma \ref{Lemm:gain_of_meas} implies
\eq{
\Aem(\Omega_{1/2}|\mathscr{K})&\ge\frac{1}{\beta}\left(\log2-\log2-\frac{1}{2}\log r_j-\frac{1}{2}\log r_{j'}-\log\frac{1}{r_j+r_{j'}}\right)\nonumber\\
&=\frac{1}{\beta}\log\frac{r_j+r_{j'}}{\sqrt{r_jr_{j'}}}.
}

Since $r_{j'}=e^{-\beta(E_{j'}-E_{j})}r_j$ follows from the definition of $r_j$ and $r_{j'}$, we get
\eq{
\frac{r_j+r_{j'}}{\sqrt{r_jr_{j'}}}&=\frac{1+e^{-\beta(E_{j'}-E_{j})}}{\sqrt{e^{-\beta(E_{j'}-E_{j})}}}\nonumber\\
&=\sqrt{e^{\beta(E_{j'}-E_{j})}}
+
\sqrt{e^{\beta(E_{j}-E_{j'})}}\nonumber\\
&>e^{\frac{\beta}{2}|E_{j'}-E_{j}|}
}
Therefore, we obtain
\eq{
\Aem(\Omega_{1/2}|\mathscr{K})>\frac{1}{2}|E_{j'}-E_{j}|.
}

Note that since $\Lambda(\psi_{j,j'})$ and $\Lambda(\psi_{j,j'})$ are energy eigenstates, they can be discriminated with probability 1 by a measurement channel $\Lambda_{\mathbb{P}}$ from $A$ to $AK$ whose two-level register $\mathscr{K}=(K,\{\ket{k}\}_{k=1,2})$ has trivial Hamiltonian $0_K$ and whose PVM $\mathbb{P}:=\{P_k\}_{k=1,2}$ satisies $[P_k,H]=0$.  
Therefore, from Lemma \ref{Lemm:cost_of_meas}, we get
\eq{
\Acos(\Omega'_{1/2,\Lambda}|\mathscr{K})\le\frac{1}{\beta}\log4,
}
which completes the proof of \eqref{Abound}.
\end{proofof}

From Theorem \ref{Thm:c_erase_SM}, we can show that, for various channels, the athermality cost scales inversely with the allowed implementation error. In the subsequent subsections, we present several examples.

\subsubsection{Projective measurements}
The first example is projective measurements.
Using Theorem \ref{Thm:c_erase_SM}, we show that projective measurement channel $\Gamma_{\mathbb{P}}(\cdot):=\sum_{k=1,2}\Tr[P_k\cdot ]\ket{k}\bra{k}_K$ whose $\{P_k\}_{k=1,2}$ can distinguish $\ket{\psi_{j,j'}}$ and $\ket{\psi^{\perp}_{j,j'}}$ with probability 1 requires an infinite amount of athermality.

\begin{proof}
Let $A_2$ be a copy of $A$ with the same Hamiltonian $H$.
Let $U$ be a unitary on $AA_2K$ defined as
\eq{
U=\dm{0}_K\otimes1_{AA_2}+\dm{1}_K\otimes (\mbox{SWAP})_{AA_2},
}
where $(\mbox{SWAP})_{AA_2}$ is the swap unitary between $A$ and $A_2$.
Since $A_2$ is a copy of $A$, the operator $U$ is an energy-preserving unitary.

We define $\Lambda'$ as a channel from $K$ to $A$ by
\eq{
\Lambda'(\cdot):=\Tr_{A_2K}[U\cdot\otimes\dm{j_0}_A\otimes\dm{j_1}_{A_2}U^\dagger],
}
where $\ket{j_0}$ and $\ket{j_1}$ are energy eigenstates of $H$.
Since this channel satisfies
\eq{
\Lambda'\circ\Gamma_{\mathbb{P}}(\dm{\psi_{j,j'}})=\dm{j_0},\enskip
\Lambda'\circ\Gamma_{\mathbb{P}}(\dm{\psi^{\perp}_{j,j'}})=\dm{j_1},
}
Theorem \ref{Thm:c_erase_SM} implies that implementing the channel $\Lambda'\circ\Gamma_{\mathbb{P}}$ requires an infinite amount of athermality.
Since the definition of $\Lambda'$ shows its cost is finite (lower than or equal to $A_R(\dm{j_0})+A_R(\dm{j_1})$), we conclude that the cost of $\Gamma_{\mathbb{P}}$ diverges.
In the same manner, we also find
\eq{
A^{\epsilon}_c(\Gamma_{\mathbb{P}})\propto\frac{1}{\epsilon}.
}
Here, we used the symbol $\propto$ to state $A^{\epsilon}_c(\Gamma_{\mathbb{P}})$ has a lower bound inversely proportional to $\epsilon$. 
\end{proof}

\subsubsection{Unitary gates}
Theorem \ref{Thm:c_erase_SM} also shows that any unitary channel $\calU$ which converts $\ket{\psi_{j,j'}}$ and $\ket{\psi^{\perp}_{j,j'}}$ to energy eigenvectors requires an infinite amount of athermality. 
Indeed, when a unitary operator $U$ satisfies $\bra{j_0}U\ket{\psi_{j,j'}}=\bra{j_1}U\ket{\psi^{\perp}_{j,j'}}=1$, the following holds:
\eq{
\calU(\dm{\psi_{j,j'}})=\dm{j_0},\enskip\calU(\dm{\psi^{\perp}_{j,j'}})=\dm{j_1},
}
where $\calU(\cdot):=U\cdot U^\dagger$. Therefore, $A^{\epsilon}_c(\calU)$ is inversely proportional to $\epsilon$.

\subsubsection{Gibbs-preserving operations}
Theorem \ref{Thm:c_erase_SM} also shows that some Gibbs-preserving operations require an infinite amount of athermality.
Interestingly, some Gibbs-preserving operations require infinite amounts of athermality, asymmetry, and energy at the same time.
Below, we provide an example of such Gibbs-preserving operations:
\begin{theorem}\label{Thm:infinite_cost_GPO}
Let $A$ be a four-level system.
We take an orthonormal basis of $A$, which is $\{\ket{0_a},\ket{0_b},\ket{1},\ket{2}\}$.
We also introduce the Hamiltonian of $A$ as
\eq{
H&=E_1\dm{1}+E_2\dm{2},\\
0=E_{0}&<E_1<E_2.
}
Namely, $\ket{0_a}$ and $\ket{0_b}$ are zero-energy degenerate states. 
Next, we fix an arbitrary inverse temperature $\beta$, and define the Gibbs state for the inverse temperature $\beta$ and the Hamiltonian $H$ as
\eq{
\tau:=\frac{e^{-\beta H}}{Z(\beta,H)},\quad Z(\beta,H)\coloneqq \Tr(e^{-\beta H}).
}
In addition, we introduce two states $\ket{+_{12}}$ and $\ket{-_{12}}$ as 
\eq{
\ket{\pm_{12}}:=\frac{\ket{1}\pm\ket{2}}{\sqrt{2}},
}
and two real positive numbers $\eps_{+}$ and $\eps_-$ as
\eq{
\eps_+:=\Tr[P_+\tau],\enskip\eps_-:=\Tr[P_-\tau]
}
where $P_{\pm}:=\dm{\pm_{12}}$.
Then, the following channel $\Lambda$ is Gibbs-preserving and $E^\epsilon_c(\Lambda)$ and $\sqrt{\calF^\epsilon_c(\Lambda)}$ obey the lower bounds that are inversely proportional to $\epsilon$:
\eq{
\Lambda(...)&:=\Tr[P_+ ...]\dm{0_a}+\Tr[P_- ...]\dm{0_b}+\Tr[(1-(P_++P_-))...]\tau',\\
\tau'&:=\frac{\tau-(\eps_+\dm{0_a}+\eps_-\dm{0_b})}{1-(\eps_++\eps_-)}.\label{def_cost_div}
}
Furthermore, when $|E_1-E_2|\ge\frac{4}{\beta}\log2$,  $A^\epsilon_c(\Lambda)$ also obeys a lower bound that is  inversely proportional to $\epsilon$.
\end{theorem}

\begin{proofof}{Theorem \ref{Thm:infinite_cost_GPO}}
First, we show $\Lambda$ is a CPTP map.
Because this is a kind of state-preparation-map, it suffices to show that $\tau'$ is actually a normalized state.
We verify that $\Tr[\tau']=1$ as follows:
\eq{
\Tr[\tau']&=\frac{\Tr[\tau]-(\eps_+\Tr[\dm{0_a}]+\eps_-\Tr[\dm{0_b}])}{1-(\eps_++\eps_-)}\nonumber\\
&=\frac{1-(\eps_++\eps_-)}{1-(\eps_++\eps_-)}\nonumber\\
&=1.
}
To show $\tau'\ge0$, it suffices to show that $\tau-(\eps_+\dm{0_a}+\eps_-\dm{0_b})\ge0$ since $1-(\eps_++\eps_-)>0$. From the definition of $\tau$, we have
\eq{
\tau-(\eps_+\dm{0_a}+\eps_-\dm{0_b})
&=\frac{1}{Z(\beta,H)}\left(e^{-\beta E_0}-\frac{e^{-\beta E_1}+e^{-\beta E_2}}{2}\right)\dm{0_a}+\frac{1}{Z(\beta,H)}\left(e^{-\beta E_0}-\frac{e^{-\beta E_1}+e^{-\beta E_2}}{2}\right)\dm{0_b}\nonumber\\
&+\frac{e^{-\beta E_1}}{Z(\beta,H)}\dm{1}+\frac{e^{-\beta E_2}}{Z(\beta,H)}\dm{2},
}
where we have used 
\eq{
\eps_-=\eps_+=\frac{e^{-\beta E_1}+e^{-\beta E_2}}{2Z(\beta,H)}.
}
Because of $E_0<E_1<E_2$, we get $e^{-\beta E_0}-\frac{e^{-\beta E_1}+e^{-\beta E_2}}{2}>0$, and thus $\tau-(\eps_+\dm{0_a}+\eps_-\dm{0_b})\ge0$. Therefore, $\Lambda$ is a CPTP map.

Next, we verify that $\Lambda$ is Gibbs-preserving as follows:
\eq{
\Lambda(\tau)&=\Tr[P_+ \tau]\dm{0_a}+\Tr[P_- \tau]\dm{0_b}+\Tr[(1-(P_++P_-))\tau]\tau'\nonumber\\
&=\eps_+\dm{0_a}+\eps_-\dm{0_b}+(1-(\eps_++\eps_-))\frac{\tau-(\eps_+\dm{0_a}+\eps_-\dm{0_b})}{1-(\eps_++\eps_-)}\nonumber\\
&=\tau.
}

Now, we show that infinite amounts of costs are required to realize $\Lambda$. 
First, we analyze the athermality cost. From the definition of the channel $\Lambda$, we have
\eq{
\Lambda(\dm{+_{12}})=\dm{0_a},\enskip
\Lambda(\dm{-_{12}})=\dm{0_b}.
}
Since $\dm{0_a}$ and $\dm{0_b}$ are energy eigenstates, as long as $|E_1-E_2|\ge\frac{4}{\beta}\log2$, Theorems \ref{Thm:ATItradeoff} and \ref{Thm:c_erase_SM} imply that
\eq{
A^\epsilon_c(\Lambda)\propto\frac{1}{\epsilon}.
}

Next, let us investigate the energy cost and asymmetry cost of $\Lambda$.
For ensemble $\Omega_{1/2}=(\dm{+_{1,2}},\dm{-_{1,2}})_{1/2}$, we have $\delta(\Lambda,\Omega_{1/2})=0$. Therefore, it suffices to show that 
\eq{
\calC(\Lambda,\Omega_{1/2})>0.
}
We here explicitly evaluate $\calC(\Lambda,\Omega_{1/2})$ to prove this inequality. Since the test states are orthogonal pure states, we have
\eq{
\calC(\Lambda,\Omega_{1/2})^2=|\bra{+_{12}}(H-\Lambda^\dagger(H))\ket{-_{12}}|^2.
}
Since $\Lambda^\dagger$ is given by
\eq{
\Lambda^\dagger(...)=\ex{...}_{\dm{0_a}}P_++\ex{...}_{\dm{0_b}}P_-+\ex{...}_{\tau'}(1-(P_++P_-)),
}
we get
\eq{
\Lambda^\dagger(H)&=E_0(P_++P_-)+\ex{H}_{\tau'}(1-(P_++P_-))\nonumber\\
&=E_0P_{1,2}+\ex{H}_{\tau'}P_{0},
}
where we have defined $P_{0}:=\dm{0_a}+\dm{0_b}$ and $P_{1,2}:=\dm{1}+\dm{2}$. By using  
\eq{
\bra{+_{12}}P_{0}\ket{-_{12}}=\bra{+_{12}}P_{1,2}\ket{-_{12}}=0,
} 
we find $\bra{+_{12}}\Lambda^\dagger(H)\ket{-_{12}}=0$. Consequently, we obtain
\eq{
\calC(\Lambda,\Omega_{1/2})^2&=|\bra{+_{12}}H\ket{-_{12}}|^2\nonumber\\
&=|\bra{+_{12}}(E_2\dm{2}+E_1\dm{1})\ket{-_{12}}|^2\nonumber\\
&=\left(\frac{E_2-E_1}{2}\right)^2\nonumber\\
&=\frac{(E_2-E_1)^2}{4}.
}
Therefore, $\calC(\Lambda,\Omega_{1/2})>0$ follows from $E_2>E_1$, which completes the proof. 

\end{proofof}

\subsection{Example of a state with finite athermality and energy but diverging asymmetry}
In the prior study~\cite{Tajima_Takagi2025}, it was shown that certain Gibbs-preserving operations require infinite asymmetry. The example presented in the previous subsection complements this finding by providing a Gibbs-preserving operation that necessitates not only infinite asymmetry but also unbounded athermality and energy. To underscore that this result is not a straightforward consequence of~\cite{Tajima_Takagi2025}, we provide examples in which both athermality and energy remain finite while asymmetry diverges.

Let us consider a system with Hamiltonian $H=\sum_{n=0}^\infty n\ket{n}\bra{n}$, i.e., the harmonic oscillator system. We introduce a pure state
\begin{align}
    \ket{\psi}\coloneqq\frac{1}{\zeta(3)}\sum_{m=1}^\infty\sqrt{\frac{1}{m^3}}\ket{m},
\end{align}
where the zeta function is defined by $\zeta(s)=\sum_{n=1}^\infty \frac{1}{n^s}$. The energy expectation value of this state is
\begin{align}
    \braket{H}_\psi=\frac{1}{\zeta(3)}\sum_{m=1}^\infty \frac{1}{m^2}=\frac{\zeta(2)}{\zeta(3)},
\end{align}
while the second moment of energy diverges:
\begin{align}
    \braket{H^2}_\psi=\frac{1}{\zeta(3)}\sum_{m=1}^\infty \frac{1}{m}=\infty.
\end{align}
Therefore, the pure state $\psi$ has finite energy $\braket{H}_\psi$, and finite relative entropy of athermality
\begin{align}
    A_R(\psi)=\braket{H}_\psi-\frac{1}{\beta}S(\psi)-F(\tau)=\frac{\zeta(2)}{\zeta(3)}-F(\tau),
\end{align}
where $F(\tau)$ denotes the free energy of the Gibbs state $\tau$. However, since the QFI is four times the variance of the Hamiltonian for pure states, we obtain
\begin{align}
    \mathcal{F}_H(\psi)=4(\braket{H^2}_\psi-\braket{H}_\psi^2)=\infty.
\end{align}
This example thus demonstrates that the result in the previous subsection cannot be directly inferred from~\cite{Tajima_Takagi2025}.

\section{Application 3: General-resource Wigner-Araki-Yanase theorem}

By applying Theorem \ref{Thm:RItradeoff} to measurement processes, we obtain a Wigner-Araki-Yanase (WAY) theorem valid for general resource theories. The traditional WAY theorem shows that the presence of a conserved quantity---typically energy or angular momentum---limits measurement accuracy, yet the ``resource'' in question has essentially been restricted to quantum fluctuations of the conserved quantity itself. As shown above, the WAY theorem can be recast as bounds of energy cost or free‑energy cost, but both remain rooted in the energy conservation law. Our framework yields a far broader limitation on measurement channels. In particular, the following theorem holds for any resource theory admitting a measure that satisfies conditions (i)--(iii):

\begin{theorem}\label{Thm:G-WAY_SM}
Let $\Omega:=(\rho_1,\rho_2)$ be two orthogonal test states on a system $A$.
Assume they are distinguished---up to error $\epsilon$---by an indirect measurement $(\eta^B,V,\{E_i\}_{i=1,2})$ involving an auxiliary system $B$, i.e.,
\eq{
\Tr[(1-E_i)V\rho_i\otimes\eta V^\dagger]\le\epsilon^2,\enskip (i=1,2).\label{eq:error_assumption_in_GWAY}
}
Suppose $V$ is a free unitary. Then, the following inequality holds:
\eq{
M(\eta^B)\ge\max_{p}\frac{|\Mem(\Omega_p|\mathscr{K})-M(\Lambda_{\mathbb{E}})|^2_+}{16K_A\epsilon}-\max_{p}|\Mem(\Omega_p|\mathscr{K})-M(\Lambda_{\mathbb{E}})|_+-c_{\max},
}
where $\Lambda_{\mathbb{E}}(...):=\sum_i\sqrt{E_i}...\sqrt{E_i}\otimes\ket{i}\bra{i}_R$.
\end{theorem}

We remark that when $\Lambda_{\mathbb{E}}(...)$ is a completely free channel, i.e., when $\id_{R}\otimes\Lambda_{\mathbb{E}}$ is free for an arbitrary reference system $R$, the above inequality becomes
\eq{
M(\eta^B)\ge\max_{p}\frac{\Mem(\Omega_p|\mathscr{K})^2}{16K_A\epsilon}-\max_{p}\Mem(\Omega_p|\mathscr{K})-c_{\max}.
}

The requirement that $\Lambda_{\mathbb{E}}(...):=\sum_i\sqrt{E_i}...\sqrt{E_i}\otimes\ket{i}\bra{i}_R$ be a completely free channel precisely corresponds to the Yanase condition in the WAY theorem---that the measurement performed on the probe commutes with the conserved quantity.

This theorem states that if there is a positive resource gain for distinguishing two states in the ensemble $\Omega_p$ on $A$, namely $\Mem(\Omega_p|\mathscr{K})>0$, then perfect discrimination of the orthogonal pair $\Omega$ demands infinite resources in the apparatus. 
Hence, over a very wide class of resource theories, measurement accuracy and resource expenditure are unavoidably in trade-off.

\begin{proofof}{Theorem \ref{Thm:G-WAY_SM}}
Let us define a channel $\Lambda$ from $A$ to $B$ as 
\eq{
\Lambda(\cdot):=\Tr_{A}[V\cdot\otimes\eta V^\dagger].
}
For an ensemble $\Omega'_{p,\Lambda}:=(\Lambda(\rho_1),\Lambda(\rho_2))_{p}$, the condition \eqref{eq:error_assumption_in_GWAY} implies 
\eq{
\PF(\Omega'_{p,\Lambda},\mathbb{E})\le\epsilon^2
}
holds for any $p$. 
Therefore, we obtain
\eq{
M(\eta^B)&\ge M_c(\Lambda)\nonumber\\
&\ge\max_{p}\frac{|\Mem(\Omega_p|\mathscr{K})-M(\Lambda_{\mathbb{E}})|^2_+}{16K_A\epsilon}-\max_{p}|\Mem(\Omega_p|\mathscr{K})-M(\Lambda_{\mathbb{E}})|_+-c_{\max}.
}
\end{proofof}

\section{Application 4: Resource-non-increasing operations requiring infinite amount of resource cost}
Applying Theorem \ref{Thm:RItradeoff_SM} to the resource‑non‑increasing operations (RNI operations), we can show that in RNI operations, there are many cost-diverging channels.
To be concrete, the following theorem holds:
\begin{theorem}\label{Thm:nogo_RNIO}
The resource theories of energy, magic, asymmetry, coherence, and athermality each contain an RNI operation $\Lambda$ and a suitable ensemble $\Omega_p$ such that $\delta(\Lambda,\Omega_p)=0$ and $\Mem(\Omega_p|\mathscr{K})-\Mcos(\Omega'_{p,\Lambda}|\mathscr{K})>0$. Consequently, these theories possess RNI operations whose implementation cost scales as $M^\epsilon_c(\Lambda)\propto 1/\epsilon$.
\end{theorem}

\begin{proofof}{Theorem \ref{Thm:nogo_RNIO}}
We present explicit examples of cost-diverging resource-non-increasing operations for each resource theory in turn.

For asymmetry and athermality, an example of a cost-diverging RNI channel is given in Theorem \ref{Thm:infinite_cost_GPO}. Indeed, the channel defined in \eqref{def_cost_div} is Gibbs-preserving, and thus athermality-non-increasing. It is also asymmetry-non-increasing, because $\dm{0_a}$, $\dm{0_b}$ and $\tau'$ are symmetric states.
As shown in Theorem \ref{Thm:infinite_cost_GPO}, this channel requires infinite amounts of asymmetry and athermality.

Next, we consider the resource theory of energy.
In this case, again we consider the four-level system considered in Theorem \ref{Thm:infinite_cost_GPO}, and define the following channel:
\eq{
\Lambda_{E}(...)&:=\Tr[P_+ ...]\dm{0_a}+\Tr[P_- ...]\dm{0_b}+\Tr[(1-(P_++P_-))...]\dm{0_a}.\label{def_cost_div_en}
}
Clearly, this channel is energy-non-increasing since $\dm{0_a}$ and $\dm{0_b}$ are ground states. Implementing this channel requires an infinite amount of energy since it satisfies
\eq{
\calC(\Omega_{1/2},\Lambda_E)&=\frac{(E_2-E_1)^2}{4},\\
\delta(\Lambda,\Omega_{1/2})&=0,
}
which are derived in the same manner as the proof of Theorem \ref{Thm:infinite_cost_GPO}.

Next, we analyze the resource theory of coherence.
Let us consider a qubit system $A_C$ and its free basis $\{\ket{i}_{A_C}\}$. 
We introduce a coherence-non-increasing CPTP map
\eq{
\Lambda_{C}(\cdot):=\Tr[\dm{+}\cdot]\dm{0}
+
\Tr[\dm{-}\cdot]\dm{1},
}
where $\ket{\pm}:=\frac{\ket{0}\pm\ket{1}}{\sqrt{2}}$. Below, we prove that this channel $\Lambda_{C}$ is cost-diverging. 

Introducing an ensemble $\Omega^{(C)}_p$ by $\Omega^{(C)}_p:=(\ket{+},\ket{-})_{1/2}$, we find that
\eq{
\delta(\Lambda,\Omega^{(C)}_p)=0
}
holds since a channel $\calR_*(\cdot):=\Tr[\dm{0}\cdot]\dm{+}
+
\Tr[\dm{1}\cdot]\dm{-}$ satisfies
\eq{
D_F(\dm{+},\calR_*\circ\Lambda(\dm{+}))=D_F(\dm{-},\calR_*\circ\Lambda(\dm{-}))=0.\label{RNIM_C1}
}
Furthermore, since $\Omega'^{(C)}_{p,\Lambda}=(\dm{0},\dm{1})_{1/2}$, the following measurement channel from $A_C$ to $K$ discriminates the states in $\Omega'^{(C)}_{p,\Lambda}$ with probability 1:
\eq{
\Gamma^{(C)}_{\mathbb{P}}(\cdot)
:=\sum^{2}_{i=1}\Tr[\dm{i}_{A_C}\cdot]\dm{i}_K,
}
where the register is $\mathscr{K}:=(K,\{\dm{k}\}_{k=0,1})$ and $\{\ket{k}\}_{k=0,1}$ denotes the free basis of the qubit $K$.
This channel is realized by the implementation $(U,\ket{0}_K\otimes\ket{0}_{K'})$, where $K'$ is a copy of $K$ and $\{\ket{k}_{K'}\}_{k=0,1}$ denotes its free basis, and $U$ is defined as
\eq{
U:=\dm{0}_{A_C}\otimes 1_K\otimes1_{K'}+\dm{1}_{A_C}\otimes(\ket{1}\bra{0}+\ket{0}\bra{1})_{K}\otimes(\ket{1}\bra{0}+\ket{0}\bra{1})_{K'}.
}
That is, they satisfy
\eq{
\Gamma^{(C)}_{\mathbb{P}}(\cdot)=\Tr_{A_CK'}[U(\cdot\otimes\dm{0}_K\otimes\dm{0}_{K'})U^\dagger]
}
From this expression, we find $\Gamma^{(C)}_{\mathbb{P}}$ is a free operation since $\ket{0}_K\otimes\ket{0}_{K'}$ is a free state, and $U$ is a free unitary as it is a permutation among free basis vectors for $A_CKK'$ system. 
Moreover, $\Gamma^{(C)}_{\mathbb{P}}$ is completely free since $\id_R\otimes\Gamma^{(C)}_{\mathbb{P}}$ is free for any auxiliary system $R$, as $1_R\otimes U$ is also a permutation among free basis vectors for $A_CKK'R$. Thus, we get $C_R(\Gamma^{(C)}_{\mathbb{P}})=0$, which implies 
\eq{
\Ccos(\Omega'^{(C)}_{p,\Lambda}|\mathscr{K})=0.
}
Therefore, if $\Cem(\Omega^{(C)}_p|\mathscr{K})>0$, $\Lambda_C$ is cost-diverging. 

Let us now prove $\Cem(\Omega^{(C)}_p|\mathscr{K})>0$. 
We introduce a channel $\Lambda_{\mathbb{F}}(\cdot):=\sum^{1}_{k=0}F_k\cdot F_k\otimes\dm{k}_K$, where $\mathbb{F}:=\{F_k\}:=\{\dm{\pm}\}$. Since this channel satisfies
\eq{
\Lambda_{\mathbb{F}}\left(\frac{I_A}{2}\right)=\frac{\dm{+}\otimes\dm{0}+\dm{-}\otimes\dm{1}}{2},\label{S252}
}
we get
\eq{
\Cem(\Omega^{(C)}_p|\mathscr{K})&\geq C_R\left(\frac{\dm{+}\otimes\dm{0}+\dm{-}\otimes\dm{1}}{2}\right)-C_R\left(\frac{I_A}{2}\right)\nonumber\\
&=C_R\left(\frac{\dm{+}\otimes\dm{0}+\dm{-}\otimes\dm{1}}{2}\right),
}
where we have used $C_R(I_A/2)=0$. 
Since $\sum_{i=0,1}\sum_{j=0,1}\bra{i,j}_{A_CK}\frac{\dm{+}\otimes\dm{0}+\dm{-}\otimes\dm{1}}{2}\ket{i,j}_{A_CK}=\frac{I_{A_CK}}{4}$, the right-hand side is evaluated as
\eq{
C_R\left(\frac{\dm{+}\otimes\dm{0}+\dm{-}\otimes\dm{1}}{2}\right)&=S\left(\frac{I_{A_CK}}{4}\right)-S\left(\frac{\dm{+}\otimes\dm{0}+\dm{-}\otimes\dm{1}}{2}\right)\nonumber\\
&=\log4-\log2\nonumber\\
&=\log2.
}
Therefore, $\Cem(\Omega^{(C)}_p|\mathscr{K})>0$, and thus the channel $\Lambda_C$ is cost-diverging.

Finally, we anlayze the resource theory of magic.
Let us consider a qubit system $A_M$ and its $Z$ basis $\{\ket{i}_{A_M}\}$.
We define $T$ state and its $Z$ rotation as
\eq{
\ket{T}&:=\frac{\ket{0}+e^{i\pi/4}\ket{1}}{\sqrt{2}},\\
\ket{T'}&:=\frac{\ket{0}+e^{-i3\pi/4}\ket{1}}{\sqrt{2}}.
}
We introduce a magic-non-increasing CPTP map 
\eq{
\Lambda_{M}(\cdot):=\Tr[\dm{T}\cdot]\dm{0}
+
\Tr[\dm{T'}\cdot]\dm{1}.
}
Below, we show that $\Lambda_{M}$ is cost-diverging. 

For an ensemble $\Omega^{(M)}_p:=(\ket{T},\ket{T'})_{1/2}$, a channel $\calR_*(\cdot):=\Tr[\dm{0}\cdot]\dm{T}
+
\Tr[\dm{1}\cdot]\dm{T'}$ satisfies
\eq{
D_F(\dm{T},\calR_*\circ\Lambda(\dm{T}))=D_F(\dm{T'},\calR_*\circ\Lambda(\dm{T'}))=0.\label{RNIM_M1}
}
Therefore, we get
\eq{
\delta(\Lambda,\Omega^{(M)}_p)=0.
}

Furthermore, since $\Omega'^{(C)}_{p,\Lambda}=(\dm{0},\dm{1})_{1/2}$, the following measurement channel from $A_M$ to $K$ discriminates the states in $\Omega'^{(C)}_{p,\Lambda}$ with probability 1:
\eq{
\Gamma^{(M)}_{\mathbb{P}}(\cdot)
:=\sum^{2}_{i=1}\Tr[\dm{i}_{A_M}\cdot]\dm{i}_K,
}
where the register is $\mathscr{K}:=(K,\{\dm{k}\}_{k=0,1})$ and $\{\ket{k}\}_{k=0,1}$ denotes the $Z$-basis states of the qubit $K$.
This channel is realized by the  implementation $(V,\ket{0}_K\otimes\ket{0}_{K'})$, where $K'$ is a copy of $K$ and $\{\ket{k}_{K'}\}_{k=0,1}$ denotes its $Z$-basis states, and $V$ is defined as
\eq{
V:=\dm{0}_{A_M}\otimes 1_K\otimes1_{K'}+\dm{1}_{A_M}\otimes(\ket{1}\bra{0}+\ket{0}\bra{1})_{K}\otimes(\ket{1}\bra{0}+\ket{0}\bra{1})_{K'}.\label{eq:unitary_in_RTM}
}
Namely, the channel $\Gamma^{(M)}_{\mathbb{P}}$ satisfies
\eq{
\Gamma^{(M)}_{\mathbb{P}}(\cdot)=\Tr_{A_MK'}[V(\cdot\otimes\dm{0}_K\otimes\dm{0}_{K'})V^\dagger].
}
Since $\ket{1}\bra{0}+\ket{0}\bra{1}$ is the Pauli-$X$ operator, Eq.~\eqref{eq:unitary_in_RTM} implies that $V$ is the $X$ operation on $K$ and $K'$ system controlled by the qubit $A_M$ in the $Z$ basis, which is a Clifford (=free) unitary. Additionally, $\ket{0}_K\otimes\ket{0}_{K'}$ is a free state, implying that $\Gamma^{(M)}_{\mathbb{P}}$ is a free operation.
Moreover, $\Gamma^{(M)}_{\mathbb{P}}$ is completely free, since $I_R\otimes V$ is also a Clifford unitary as it is a controlled-$XX$ operation acting on $A_MKK'$ qubits, leaving the auxiliary system $R$ unaffected. Thus, we get $\calD_{\max}(\Gamma^{(M)}_{\mathbb{P}})=0$, implying that 
\eq{
\Dcos(\Omega'_{p,\Lambda}|\mathscr{K})=0,
}
where $\Dcos$ and $\Dem$ are $\Mcos$ and $\Mem$ whose $M$ is the max-relative entropy of magic $\calD_{\max}$.
Therefore, if $\Dem(\Omega_p|\mathscr{K})>0$, the channel $\Lambda_M$ is cost-diverging.

We now prove $\Dem(\Omega_p|\mathscr{K})>0$. For a quantum channel $\Lambda_{\mathbb{E}}(\cdot):=\sum^{1}_{k=0}E_k\cdot E_k\otimes\dm{k}_K$, where $\mathbb{E}:=\{E_k\}:=\{\dm{T},\dm{T'}\}$, we have
\eq{
\Lambda_{\mathbb{E}}\left(\dm{0}\right)=\frac{\dm{T}\otimes\dm{0}+\dm{T'}\otimes\dm{1}}{2}.\label{S264}
}
Therefore, by using $\calD_{\max}(\dm{0})=0$, we get
\eq{
\Dem(\Omega_p|\mathscr{K})&\geq \calD_{\max}\left(\frac{\dm{T}\otimes\dm{0}+\dm{T'}\otimes\dm{1}}{2}\right)-\calD_{\max}\left(\dm{0}\right)\nonumber\\
&=\calD_{\max}\left(\frac{\dm{T}\otimes\dm{0}+\dm{T'}\otimes\dm{1}}{2}\right).
}
To evaluate $\calD_{\max}\left(\frac{\dm{T}\otimes\dm{0}+\dm{T'}\otimes\dm{1}}{2}\right)$, note that the state $\sigma:=\frac{\dm{T}\otimes\dm{0}+\dm{T'}\otimes\dm{1}}{2}$ and the T state $\dm{T}$ can be converted to each other via free operations.
Indeed, from the T state, $\sigma$ can be made by adding an additional qubit $K$ whose state is the maximally mixed state $I/2$ and performing the controlled-$Z$ gate whose control qubit is $K$ and the target qubit is $A_M$.
Conversely, from the state $\sigma$, the T state can be obtained by performing the controlled-$Z$ gate whose control qubit is $K$ and the target qubit is $A_M$, and then tracing over the qubit $K$.
Therefore,
\eq{
\calD_{\max}(\sigma)&=\calD_{\max}(\dm{T})\nonumber\\
&\stackrel{(a)}{=}\log(1+2\sin(\pi/18))\nonumber\\
&>0.
}
Here, in (a), we used the value of $\calD_{\max}(\dm{T})$ calculated in Ref.~\cite{PhysRevLett.124.090505}.
Therefore, $\Dem(\Omega_p|\mathscr{K})>0$, and thus the channel $\Lambda_M$ is cost-diverging.
\end{proofof}

\end{document}